\documentclass[11pt]{article}

\usepackage[margin=1in]{geometry}
\usepackage{amsmath, amssymb, amsthm}
\usepackage{graphicx}
\usepackage{tabularx}
\usepackage{tikz}
\usetikzlibrary{calc,decorations.pathreplacing}
\usepackage{booktabs}
\usepackage{multirow}
\usepackage{hyperref}
\usepackage{natbib}
\usepackage{algorithm}
\usepackage{enumitem}
\usepackage{subcaption}
\usepackage{times}
\usepackage{appendix}
\usepackage{placeins}

\usetikzlibrary{arrows.meta, shadows.blur, calc}
\usepackage{xcolor}
\definecolor{coldblue}{RGB}{214,234,248}
\definecolor{coldframe}{RGB}{52,120,180}
\definecolor{hotred}{RGB}{248,215,210}
\definecolor{hotframe}{RGB}{185,70,60}
\definecolor{neutralbg}{RGB}{240,245,240}
\definecolor{neutframe}{RGB}{90,135,100}
\definecolor{arrowcol}{RGB}{70,90,110}
\definecolor{ceilcol}{RGB}{25,95,165}
\definecolor{floorcol}{RGB}{180,45,35}
\definecolor{outdoororange}{RGB}{220,100,20}
\definecolor{indoorcyan}{RGB}{20,150,180}

\usepackage{nomencl}
\usepackage{multicol}
\makenomenclature

\renewcommand{\nompreamble}{\begin{multicols}{2}}
	\renewcommand{\nompostamble}{\end{multicols}}

\DeclareMathOperator*{\argmin}{arg\,min}

\newcommand{\tsum}{\sum\limits}

\usepackage{pgffor}
\usepackage{setspace}

\title{Robust tests for log-normal lifetimes under cyclic-stress accelerated tests and interval monitoring}
\author{
	María Jaenada\\
	\small Department of Statistics and O.R., UNED, Madrid, Spain
    \and Leandro Pardo \\
	\small Department of Statistics and O.R., Complutense University of Madrid, Spain
	\and
    Kiran Prajapat\thanks{Corresponding author. Email: kiranprajapat92@gmail.com, kiran.prajapat@newcastle.ac.uk (Kiran Prajapat)}\\
    \small School of Mathematics, Statistics and Physics, Newcastle University, Newcastle upon Tyne, UK
}\date{}

\begin{document}
	
	\maketitle
    \begin{abstract}
        Some products are, by nature, designed to operate under stress conditions that repeatedly cycle between two levels, such as batteries undergoing repeated charge and discharge, or components exposed to recurring thermal or pressure cycles. When such products are tested under accelerated conditions, the applied stress is naturally elevated in the same cyclic manner rather than held constant or increased in a single step. Cyclic-stress accelerated life testing (CyALT) reflects this by alternating the applied stress between specified levels. Units are often inspected only at fixed times, so the data consist of interval failure counts rather than exact lifetimes. Hypothesis testing plays a crucial role in this setting. It can assess whether the applied stress significantly affects lifetime, or validate the parameter values used to design the experiment. Classical tests rely on the maximum likelihood estimator, which is highly sensitive to contamination. Since interval probabilities are estimated jointly across stress conditions, a single anomalous failure count can distort the entire parameter vector and severely affect the test's significance level and power. This paper develops Wald-type and Rao-type test statistics for log-normal CyALT lifetimes under interval monitoring, based on the weighted minimum density power divergence estimator. Both tests are derived for simple and composite null hypotheses. A Monte Carlo study shows that classical tests lose control of the significance level rapidly under contamination, while the robust tests stay close to nominal with only a modest loss of power under clean data. An application to air-conditioner evaporator lifetime data further illustrates the proposed tests.
    \end{abstract}
    \providecommand{\keywords}[1]{\textbf{Keywords:} #1}
    \keywords{Accelerated life testing; Cumulative exposure model; Inverse power law; Cyclic-stress loading; Weighted minimum density power divergence estimator; Hypothesis testing; Wald and Rao tests; Robust estimation.}\\[2pt]
\textbf{MSC 2020:} 62F03, 62F35, 62F12, 62N05, 90B25
	
	\newtheorem{theorem}{Theorem}
    \newtheorem{acknowledgement}[theorem]{Acknowledgement}
    \newtheorem{axiom}[theorem]{Axiom}
    \newtheorem{case}[theorem]{Case}
    \newtheorem{claim}[theorem]{Claim}
    \newtheorem{conclusion}[theorem]{Conclusion}
    \newtheorem{condition}[theorem]{Condition}
    \newtheorem{conjecture}[theorem]{Conjecture}
    \newtheorem{corollary}[theorem]{Corollary}
    \newtheorem{criterion}[theorem]{Criterion}
    \newtheorem{definition}[theorem]{Definition}
    \newtheorem{example}[theorem]{Example}
    \newtheorem{exercise}[theorem]{Exercise}
    \newtheorem{lemma}[theorem]{Lemma}
    \newtheorem{notation}[theorem]{Notation}
    \newtheorem{problem}[theorem]{Problem}
    \newtheorem{proposition}[theorem]{Proposition}
    \newtheorem{remark}[theorem]{Remark}
    \newtheorem{solution}[theorem]{Solution}
    \newtheorem{summary}[theorem]{Summary}
	
	\section{Introduction}

Reliability assessment plays a central role in the design, production, and maintenance of highly reliable products and systems. 
In many modern applications, products are designed to operate for long periods under normal conditions, and only a small number of failures may be observed within the time available for experimentation. This makes the direct inference of lifetime characteristics difficult or prohibitively expensive. Accelerated life testing (ALT) addresses this problem by exposing experimental units to stress levels higher than those encountered under normal operating conditions. The resulting failure information is then combined with a suitable lifetime-stress relationship to infer product reliability under normal-use conditions. 
Since the seminal work of \cite{nelson1990} and \cite{meeker1998} , a substantial body of literature has been developed on a wide variety of ALT designs.
Moreover, hypothesis testing is an essential component of ALT, as practitioners frequently need to determine whether particular model assumptions are supported by the observed data. For example, it may be necessary to assess whether the applied stress has a significant effect on lifetime or whether a simpler submodel provides an adequate description of the failure mechanism. Hypothesis testing is also important at the experimental design stage, since optimal ALT designs rely on parameter values estimated in advance using some preliminary data, and unreliable estimates can lead to a poorly chosen design.
In the context of step-stress testing, \cite{wang2009testing} proposed a test statistic for assessing the validity of the exponential lifetime distribution and the cumulative exposure assumption.
\cite{zhu2020exact} developed exact inference based on likelihood-ratio test for a simple step-stress ALT. 
\cite{Kateri2024scale} extended the likelihood-ratio testing framework of \cite{zhu2020exact} to multilevel step-stress models and propose scale-invariant test statistics with free parameter distribution for testing a specified value of the stress-effect parameter.
\cite{balakrishnan2025robusttesting, balakrishnan2026robusttesting} proposed robust divergence-based test statistics for step-stress tests.

Several stress-loading schemes have been considered in the ALT literature. Under constant-stress testing, each unit is subjected to a fixed stress level throughout the experiment. In step-stress testing, the stress level is increased at one or more predetermined times. Although these schemes are appropriate in many applications, some products are naturally exposed to stresses that vary repeatedly over time. Examples arise in devices affected by alternating thermal conditions, repeated voltage changes, pressure cycles, or recurrent variations between operating and standby modes. In cyclic-stress ALT (CyALT) the applied stress alternates between specified levels according to a predetermined schedule, thereby producing repeated periods of relatively low and high stress. \cite{kim2021optimal} first developed optimal CyALT plans under log-normal lifetimes, Type-I censoring, and the cumulative exposure model, determining the common floor level and unit allocation proportions to minimize the asymptotic variance of the MLE of a lifetime quantile at the use condition. \cite{kim2023optimal} later extended this work to proposing optimal and compromise CyALT plans with practical sample size guidelines under the complete data. More recently, \cite{zhang2024reliability} studied reliability analysis under cyclic ALTs for the log-location-scale family of distributions under censoring, and \cite{zhang2025reliability} considered reliability estimation under cyclic ALTs based on a scale family of distributions. 

Most statistical procedures for ALTs are developed under the assumption that failure times are observed exactly. Continuous monitoring, however, may be costly, technically difficult, or incompatible with the experimental protocol. In practice, these units are often inspected only at a finite set of predetermined times. Consequently, the exact failure time of a unit is not observed; it is only known that the failure occurred between two consecutive inspections. Units that remain operational at the final inspection are right-censored. The resulting data therefore consist of the numbers of failures recorded in the successive inspection intervals, together with the number of surviving units at the end of the experiment. Such interval-monitored data arise naturally when inspections require specialized equipment, manual intervention or substantial operational costs.
To mention a few, \cite{ding2013design} discussed  optimal ALT planning with random removals. 
Inference methods for step-stress ALTs with interval-censoring was studied in \cite{gouno2001}.
Optimal design under step-stress models was analyzed in \cite{bobotas2019optimal}. 
Bayesian inference for interval-inspected setups was developed in \cite{kohl2019bayesian}.
More recently, there has been growing interest in robust inference for interval-monitored step-stress ALTs, particularly following the contributions of \cite{balakrishnan2026robustinterval, baghel2024robust}. 
For CyALT models, robust point estimation  under interval monitoring and log-normal lifetimes was studied in \cite{jaenada2026robust}.
The interval monitoring setup introduces important statistical challenges. In particular, the likelihood must be formulated in terms of interval failure probabilities rather than exact lifetime densities. Moreover, under a cyclic-stress scheme, these probabilities depend on the entire stress history experienced by the units. Therefore, a suitable model must simultaneously account for the cumulative effect of the applied stresses, the repeated changes in the stress level, and the incomplete information caused by interval monitoring. 
In an interval-monitored cyclic accelerated life test, the testing problem is particularly challenging, as statistical tests must account simultaneously for the cyclic nature of the stress profile and the incomplete information caused by interval monitoring. 

In interval-monitored experiments, an atypical unit does not appear as an isolated extreme failure time but may generate an unexpectedly large failure count in a particular inspection interval. Such anomalous counts may result from recording errors, inspection inaccuracies, uncontrolled changes in the experimental environment, heterogeneity among the tested units, or the presence of a small subpopulation following a different failure mechanism. Since interval probabilities are estimated jointly, contamination in one or more intervals can affect the entire parameter vector. 

Classical inferential procedures, including hypothesis tests on the model parameters, are commonly based on the maximum likelihood estimator (MLE). Under standard regularity conditions, likelihood-based methods generally perform well when the assumed probabilistic model is correctly specified for all units under study. However, they may be highly sensitive to model contamination, and even a small number of outlying observations can substantially affect the resulting estimates and inferential conclusions.
This sensitivity is inherited by the likelihood-based test statistics, resulting in severe distortions in the  significance level and power of the tests, and thus leading to incorrect scientific or engineering conclusions.
Robust hypothesis-testing procedures are therefore needed to obtain inferential conclusions that remain reliable in the presence of moderate contamination \cite{basu2016generalized, basu2022robust}. 

\cite{jaenada2026robust} derived robust point estimation and confidence intervals for interval-monitored CyALTs based on the Density Power Divergence (DPD).  In this paper, we generalize the classical Wald and Rao-score tests under the DPD framework, deriving robust test statistics whose significance levels and power are not excessively affected by a small proportion of anomalous observations. 
The effect of the successive stress exposures is described through a cumulative exposure model. Lifetimes at a constant stress level are assumed to follow a log-normal distribution, with a common shape parameter and a stress-dependent scale parameter. A log-linear relationship is used to connect the scale parameter with the applied stress, allowing the lifetime distribution at normal operating conditions to be estimated from observations obtained under accelerated conditions. Since the units are inspected only at predetermined times, the observed interval counts follow a multinomial distribution whose cell probabilities are determined by the CyALT lifetime model.
The finite-sample performance of the proposed Wald-type and Rao-type test statistics is examined through an extensive Monte Carlo simulation study.  Both uncontaminated and contaminated scenarios are investigated to assess the trade-off between efficiency and robustness. The numerical results show how the DPD tuning parameter influences the estimators and illustrate the advantages of robust procedures when the observed interval counts depart from the assumed model.

The remainder of the paper is organized as follows. Section~\ref{sec:model} introduces the CyALT model under log-normal lifetimes and interval monitoring, and presents the weighted minimum density power divergence estimator together with its asymptotic distribution. Section~\ref{sec:robust_tests} formulates the simple and composite testing problems of interest throughout the paper. Sections~\ref{sec:wald} and \ref{sec:rao} develop the Wald-type and Rao-type test statistics, respectively, and establish their asymptotic null distribution and consistency. Section~\ref{sec:IF} derives the influence function of both test statistics, quantifying their robustness to point contamination. Section~\ref{sec:sim_study} reports an extensive Monte Carlo simulation study examining the empirical level and power of the proposed tests under increasing sample size and contamination. Section~\ref{sec:real_data} illustrates the tests on interval-monitored failure data from an automotive air-conditioner evaporator experiment. Section~\ref{sec:conclusion} concludes the paper.

    \section{Model, assumptions and parameter estimation}
    \label{sec:model}

    \subsection{Test scenario under CyALT}
    \label{subsec:test_cond}

    In a CyALT, test units are exposed repeatedly to a sequence of alternating stress levels according to a predetermined schedule. 
    We suppose that $K$ identical units are put on test under CyALT and exposed to $R$ cyclic-stress conditions $s_1, s_2, \dots, s_R$, each exceeding the use condition $s_0.$ 
    The test runs until a pre-specified censoring time $t_c$, and $K_i$ units, a positive integer, are allocated to condition $s_i$ with $\sum_{i=1}^R K_i = K$; the corresponding allocation proportion is $\pi_i = K_i/K$, so that $0 < \pi_i < 1$ and $\sum_{i=1}^R \pi_i = 1$. Each cyclic-stress condition $s_i$ is characterized by a floor level $s_{iF}$ and a ceiling level $s_{iC}$. 
    Under this study it is assumed 
    that the acceleration is implemented by increasing only the ceiling pressure so that $s_{1F} = s_{2F} = \dots = s_{RF} < s_{1C} < s_{2C} < \dots < s_{RC}$. For each condition $s_i$, the ceiling level $s_{iC}$ is maintained for $(100\tau)\%$ of each cycle and the floor level for the remaining $(100(1-\tau))\%$. 
    {This assumption is appropriate for many real-world applications, such as air-conditioning systems tested under CyALT (\cite{kim2021optimal}}). Under CyALT experiment, lifetimes of testing units are measured in cycles. 

     \begin{figure}[H]
         \centering
         \resizebox{1\textwidth}{6.25cm}{%
         \begin{tikzpicture}[x=2.5cm,y=1.2cm,>=stealth]

    \definecolor{sOneCol}{RGB}{170,195,225}   
    \definecolor{sTwoCol}{RGB}{70,110,170}    
    \definecolor{sRCol}{RGB}{10,30,90}        
    \definecolor{useShade}{RGB}{220,220,220}  

    \def\tauNum{0.375}
    \def\tc{5}
    \def\tauSym{\tau}
    \def\sZeroF{0.3}
    \def\sZeroC{0.9}
    \def\sF{1.2}
    \def\sOneC{2.2}
    \def\sTwoC{2.4}
    \def\sRC{3.0}

    \foreach \k in {0,1} {
        \fill[useShade,opacity=0.6]
            (\k,\sZeroF) rectangle (\k+1,\sZeroC);
    }
    \fill[useShade,opacity=0.6]
        (\tc-1,\sZeroF) rectangle (\tc,\sZeroC);

    \draw[->] (0,0) -- (\tc+0.5,0)
        node[below right] {$t$};
    \draw[->] (0,0) -- (0,3.5)
        node[above left] {$s$};

    \foreach \y/\lab in {
        \sZeroF/$s_{0F}$,
        \sZeroC/$s_{0C}$,
        \sF/\scriptsize{$s_{1F}=s_{2F}=\dots=s_{RF}$},
        \sOneC/$s_{1C}$,
        \sTwoC/$s_{2C}$,
        \sRC/$s_{RC}$
    } {
        \draw (0,\y) -- (0.05,\y);
        \node[left] at (0,\y) {\lab};
    }

    \foreach \k in {0,1} {
        \draw[dotted,very thick,color=sOneCol]
            (\k,\sOneC) -- (\k+\tauNum-0.02,\sOneC)
            -- (\k+\tauNum-0.02,\sF)
            -- (\k+1-0.02,\sF)
            -- (\k+1-0.02,\sOneC);
    }

    \foreach \k in {0,1} {
        \draw[dash dot,very thick,color=sTwoCol]
            (\k,\sTwoC) -- (\k+\tauNum-0.015,\sTwoC)
            -- (\k+\tauNum,\sF)
            -- (\k+1-0.015,\sF)
            -- (\k+1-0.015,\sTwoC);
    }

    \foreach \k in {0,1} {
        \draw[very thick,color=sRCol]
            (\k,\sRC) -- (\k+\tauNum,\sRC)
            -- (\k+\tauNum,\sF)
            -- (\k+1,\sF)
            -- (\k+1,\sRC);
    }

    \draw[dashed,very thick,color=sOneCol]
        (2,\sOneC) -- (\tc-1,\sOneC);
    \draw[dashed,very thick,color=sTwoCol]
        (2,\sTwoC) -- (\tc-1,\sTwoC);
    \draw[dashed,very thick,color=sRCol]
        (2,\sRC)   -- (\tc-1,\sRC);
    \draw[dashed,thick,black]
        (2,\sZeroC)-- (\tc-1,\sZeroC);

    \foreach \k in {0,1} {
        \draw[very thick,black]
            (\k,\sZeroC) -- (\k+\tauNum,\sZeroC)
            -- (\k+\tauNum,\sZeroF)
            -- (\k+1,\sZeroF)
            -- (\k+1,\sZeroC);
    }

    \foreach \x/\lbl in {
        1.2/$IT_1$,
        2.3/$IT_2$,
        3/$IT_3$,
        3.8/$\dots\dots$} {
        \draw[thick,black] (\x,-0.05) -- (\x,0.05);
        \node[black,below=2pt] at (\x,0) {\lbl};
    }
    \draw[thick,black] (\tc,-0.05) -- (\tc,3.2);
    \node[black,below=2pt] at (\tc,0) {$IT_L=t_c$};

    \draw[<->,thick] (0,\sZeroC-0.25) --
        (\tauNum,\sZeroC-0.25)
        node[midway,below=2pt] {$\tauSym$};
    \draw[<->,thick] (\tauNum,\sZeroF+0.25) --
        (1,\sZeroF+0.25)
        node[midway,above=2pt] {$1-\tauSym$};
    \draw[<->,thick] (0,\sZeroF-1.25) --
        (1,\sZeroF-1.25)
        node[midway,above=2pt] {1 Cycle};

    \draw[dotted,very thick,color=sOneCol]
        (\tc-1,\sOneC) -- (\tc-1+\tauNum-0.02,\sOneC)
        -- (\tc-1+\tauNum-0.02,\sF)
        -- (\tc-0.02,\sF)
        -- (\tc-0.02,\sOneC);
    \draw[dash dot,very thick,color=sTwoCol]
        (\tc-1,\sTwoC) -- (\tc-1+\tauNum-0.015,\sTwoC)
        -- (\tc-1+\tauNum,\sF)
        -- (\tc-0.015,\sF)
        -- (\tc-0.015,\sTwoC);
    \draw[very thick,color=sRCol]
        (\tc-1,\sRC) -- (\tc-1+\tauNum,\sRC)
        -- (\tc-1+\tauNum,\sF)
        -- (\tc,\sF)
        -- (\tc,\sRC);
    \draw[very thick,black]
        (\tc-1,\sZeroC) -- (\tc-1+\tauNum,\sZeroC)
        -- (\tc-1+\tauNum,\sZeroF)
        -- (\tc,\sZeroF)
        -- (\tc,\sZeroC);

    \node[color=sOneCol] at (1.5,\sOneC+0.15) {$s_1$};
    \node[color=sTwoCol] at (1.5,\sTwoC+0.15) {$s_2$};
    \node[color=sRCol]   at (1.5,\sRC+0.15)   {$s_R$};

\end{tikzpicture}}
         \caption{Cyclic-stress ALT with $R$ accelerated cyclic-stress conditions: dotted light blue = $s_1$ (floor $s_{1F}$, ceiling $s_{1C}$), dash-dotted medium blue = $s_2$ (floor $s_{2F}$, ceiling $s_{2C}$), solid dark navy = $s_R$ (floor $s_{RF}$, ceiling $s_{RC}$), and solid black line with grey background = use stress (floor $s_{0F}$, ceiling $s_{0C}$).}
         \label{fig:CyALT}
     \end{figure}
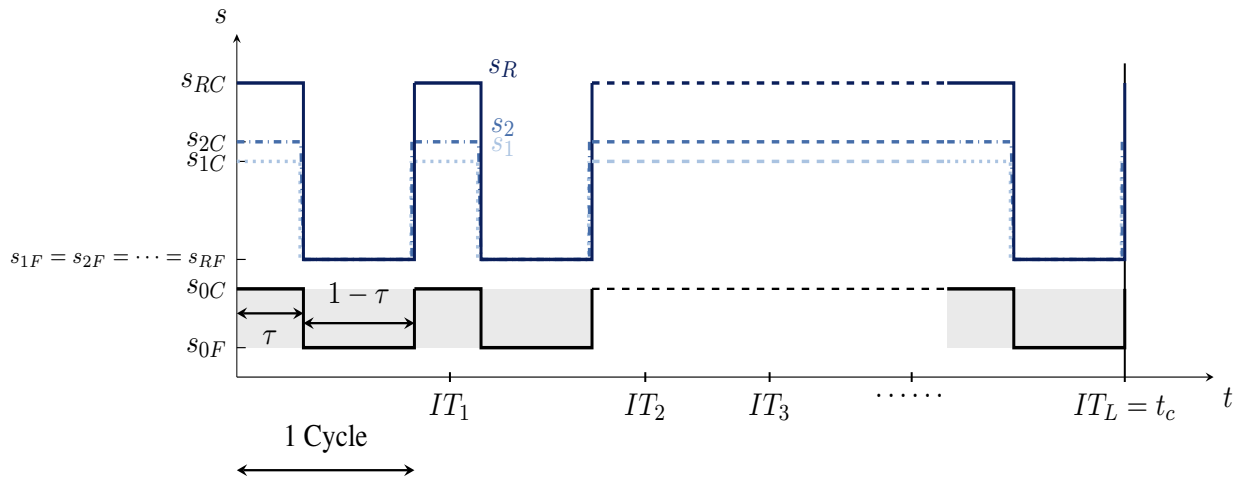
    
    Further, to perform the test under interval monitoring in Figure \ref{fig:interval_monitoring}, units are inspected at $L$ pre-specified time points $0 < IT_1 < IT_2 < \dots < IT_L = t_c$. The observed data consist of failure counts $n_{ij}$, the number of units failing in the $j$-th interval $(IT_{j-1}, IT_j]$ under $i$-th condition $s_i$, and $n_{i,L+1}$, the number of surviving units beyond $t_c$. The complete observed dataset is given by
    \begin{equation*}
        \mathcal{D} = \{ n_{ij},\; i = 1,\dots,R,\; 
        j = 1,\dots,L,\; 
        n_{1,L+1}, \dots, n_{R,L+1} \}.
    \end{equation*}
     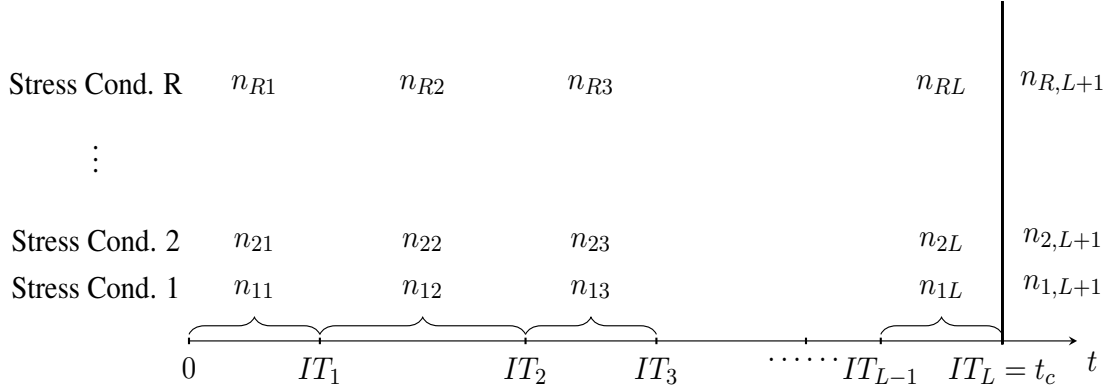
\begin{figure}[H]
         \centering
         \resizebox{0.9\textwidth}{5.25cm}{%
         \begin{tikzpicture}[x=2.5cm,y=0.8cm,>=stealth]
			\def\tc{4.85}  
			\def\ITs{{0.5,1.1,2.1,3}}  
			\def\L{3}    
			
			\draw[->] (0.5,0) -- (\tc+0.4,0) node[below right] {$t$};
			
			\foreach \x/\lbl in {0.5/$0$,1.2/$IT_1$,2.3/$IT_2$,3/$IT_3$,3.8/$\dots\dots$,4.2/$IT_{L-1}$,4.85/$IT_L=t_c$} {
				\draw[thick,black] (\x,-0.05) -- (\x,0.05);
				\node[black,below=2pt] at (\x,0) {\lbl};
				
				\draw[thick,black] (\tc,-0.05) -- (\tc,5.2);
				
			}
			\foreach \i/\xstart/\xend in {1/0.5/1.2,2/1.2/2.3,3/2.3/3,4/4.2/4.85} {
				\draw[decorate,decoration={brace,amplitude=6pt}] 
				(\xstart,0.1) -- (\xend,0.1);
			}
			\node[above=1em] at (0.001*\tc/\L,0) {Stress Cond. 1};
			\node[above=1em] at ($ (0.5,0)!0.5!(1.2,0) $) {$n_{11}$};
			\node[above=1em] at ($ (1.2,0)!0.5!(2.3,0) $) {$n_{12}$};
			\node[above=1em] at ($ (2.3,0)!0.5!(3,0) $) {$n_{13}$};
			\node[above=1em] at ($ (4.2,0)!0.5!(4.85,0) $) {$n_{1L}$};
			\node[above=1em] at ($ (4.85,0)!0.5!(5.5,0) $) {$n_{1,L+1}$};
			
			\node[above=2.5em] at (0.001*\tc/\L,0) {Stress Cond. 2};
			\node[above=2.5em] at ($ (0.5,0)!0.5!(1.2,0) $) {$n_{21}$};
			\node[above=2.5em] at ($ (1.2,0)!0.5!(2.3,0) $) {$n_{22}$};
			\node[above=2.5em] at ($ (2.3,0)!0.5!(3,0) $) {$n_{23}$};
			\node[above=2.5em] at ($ (4.2,0)!0.5!(4.85,0) $) {$n_{2L}$};
			\node[above=2.5em] at ($ (4.85,0)!0.5!(5.5,0) $) {$n_{2,L+1}$};
			
			\node[above=5em] at (0.001*\tc/\L,0) {\vdots};
			
			\node[above=7.5em] at (0.001*\tc/\L,0) {Stress Cond. R};
			\node[above=7.5em] at ($ (0.5,0)!0.5!(1.2,0) $) {$n_{R1}$};
			\node[above=7.5em] at ($ (1.2,0)!0.5!(2.3,0) $) {$n_{R2}$};
			\node[above=7.5em] at ($ (2.3,0)!0.5!(3,0) $) {$n_{R3}$};
			\node[above=7.5em] at ($ (4.2,0)!0.5!(4.85,0) $) {$n_{RL}$};
			\node[above=7.5em] at ($ (4.85,0)!0.5!(5.5,0) $) {$n_{R,L+1}$};
\end{tikzpicture}
        }
         \caption{Failure counts under interval monitoring of CyALT forming the observed data 
         $\mathcal{D}$.}
         \label{fig:interval_monitoring}
     \end{figure}
    Figure \ref{fig:CyALT} illustrates the cyclic-stress loading profile, and Figure \ref{fig:interval_monitoring} shows the interval-monitored data structure. We now describe the statistical model related to the lifetime characteristics underlying the observed data $\mathcal{D}$.   

    \subsection{Model assumptions}
    \label{subsec:model_assump}
    
    At any stress level $s_{im}$, the lifetime of a unit follows a log-normal distribution with scale parameter $\zeta(s_{im})$ and shape parameter $\sigma$ denoted by $\text{log-normal}(\zeta(s_{im}), \sigma)$, with PDF
    \begin{equation*}
        f(x, \zeta(s_{im}), \sigma) = 
        \frac{1}{\sqrt{2\pi}\,\sigma x}
        \exp \left(
        -\frac{1}{2\sigma^2}
        \big\{\ln(x) - \ln(\zeta(s_{im}))\big\}^2
        \right).
    \end{equation*}
    Note that the shape parameter plays no role in lifetime change when stress is changed, since  the lifetime of a unit depends on the stress variable only through its scale. 
    This homogeneity condition is commonly adopted in ALT studies, as it simplifies some mathematical derivations while still providing a realistic assumption for practical applications.
    The scale  parameter is linked to the standardized stress  $s_{im}$ through the log-linear relationship
    \begin{equation}
        \ln(\zeta(s_{im})) = \alpha_0 + \alpha_1 s_{im},
        \label{eq:stress_life}
    \end{equation}
    with intercept $\alpha_0 \in \mathbb{R}$ and 
    slope $\alpha_1 < 0$, since higher stress always 
    corresponds to a shorter lifetime. The 
    standardized stress is defined as
    \begin{equation}
        s_{im} = \frac{g(v_{im}) - g(v_{0F})}
        {g(v_{hC}) - g(v_{0F})},
        \label{eq:standardisation}
    \end{equation}
    where $v_{im}$ is the physical stress variable, 
    $v_{0F}$ is the floor stress at the use condition, 
    $v_{hC}$ is the ceiling stress at the highest 
    accelerated condition, and $g(\cdot)$ is a monotone 
    function chosen according to the nature of the 
    stress. Common choices include
    \begin{align*}
        g(v) = 
        \begin{cases}
            \ln(v) & \text{for the inverse power law, 
                      appropriate for mechanical stress,} \\
            1/v    & \text{for the Arrhenius model, 
                      appropriate for thermal stress,} \\
            v      & \text{for a linear model.}
        \end{cases}
    \end{align*}
Under the standardization, $s_{im} \in [0,1]$ for all conditions, with $s_{0F} = 0$ at the use floor and $s_{hC} = 1$ at the highest accelerated ceiling.
    
    We use the cumulative exposure (CE) model originally introduced by \cite{nelson1980} to account for the effect of changing stress levels on lifetime within each cycle. Under this model, the cumulative exposure at $t$ cycles under condition $s_i$ is
    \begin{equation*}
        \epsilon_i(t) = t\left\{
        \frac{\tau}{\zeta(s_{iC})} + 
        \frac{1-\tau}{\zeta(s_{iF})}
        \right\}.
    \end{equation*}
    That is, the CDF due to the cumulative exposure model at the time point $t$ is equivalent to the CDF with a unit scale at the cumulative exposure at $t$ the time. Therefore, the CDF at $t$ cycles under condition $s_i$ is therefore
    \begin{equation*}
        F_i(t) = \Phi \left( \frac{\ln(\epsilon_i(t))}{\sigma} \right) = \Phi \left(
        \frac{\ln(t) + 
        \ln(\tau\eta_{iC} + (1-\tau)\eta_{iF})}
        {\sigma}
        \right),
    \end{equation*}
    where $\eta_{im} = 1/\zeta(s_{im}) = \exp\{-\alpha_0 - \alpha_1 s_{im}\}$ for $m \in \{C, F\}$ and $\Phi$ is the standard normal CDF. Define $ \mu_i = -\ln \left( \tau\eta_{iC} + (1-\tau)\eta_{iF} \right)$. Let $T_{ik}$ be the lifetime of a $k$-th unit subjected to $i$-th cyclic stress loading, then $T_{ik} \sim \text{log-normal}(\mu_i, \sigma)$ for $ k=1,\ldots, K_i$ and $i=1,\ldots, R$. Hence, the CDF and PDF at $t$ cycles under $i$-th cyclic-stress loading are given by
    \begin{align*}
        F_i(t) =  \Phi \left( \frac{\ln(t) - \mu_i}{\sigma} \right), \qquad 
        f_i(t) = \frac{1}{\sigma t} \phi \left( \frac{\ln(t) -\mu_i} {\sigma} \right),
    \end{align*}
    respectively, where $\phi$ represents the PDF of a standard normal distribution. The lifetime on a $\log$ scale under $i$-th cyclic-stress condition, $s_i$ satisfies
    \begin{equation}
        X_{ik} = \ln(T_{ik}) \sim N(\mu_i, \sigma^2),
        \label{eq:log-normal}
    \end{equation}
    with location parameter $\mu_i$ and standard deviation $\sigma$. 
    
    Note that the parameter vector of interest is $\boldsymbol{\theta} = (\alpha_0, \alpha_1, \sigma)^\top \in \boldsymbol{\Theta}$ with the parameter space $ \boldsymbol{\Theta}$ that we represent by
    \begin{equation*}
        \boldsymbol{\Theta} = \left\{
        (\alpha_0, \alpha_1, \sigma) :
        \alpha_0 \in \mathbb{R},\;
        \alpha_1 < 0,\;
        \sigma > 0
        \right\}
        = \mathbb{R} \times \mathbb{R}^- \times 
        \mathbb{R}^+.
    \end{equation*}
    The probability of failure in the $j$-th 
    interval under condition $s_i$ is
    \begin{align}
        p_{ij}(\boldsymbol{\theta}) &= 
        \Phi \left(
        \frac{\ln(IT_j) - \mu_i}{\sigma}
        \right) -
        \Phi \left(
        \frac{\ln(IT_{j-1}) - \mu_i}{\sigma}
        \right),
        i = 1,\dots,R,\quad j = 1,\dots,L,
        \label{eq:p_ij}
    \end{align}
    and the survival probability at stress condition $s_i$ is
    \begin{equation}
        p_{i,L+1}(\boldsymbol{\theta}) = 
        1 - \Phi \left(
        \frac{\ln(t_c) - \mu_i}{\sigma}
        \right),
        \quad i = 1,\dots,R.
        \label{eq:p_iL1}
    \end{equation} 
    These defined probabilities are used to build the subsequent inference in all the upcoming sections. Now we present the robust estimators of $\boldsymbol{\theta}$ based on the DPD and its asymptotic distribution.
    
    \subsection{The Weighted Minimum Density Power Divergence Estimator}
    \label{sec:WMDPDE}

    In this section, we describe robust estimators for the CyALT model with 
	log-normal lifetimes using the DPD approach proposed 
	by \cite{jaenada2026robust}. 
    Let define the theoretical probability vector for the cycle $i$ as
	\begin{align*}
		\mathbf{p}_i(\boldsymbol{\theta}) = \left( p_{i1}(\boldsymbol{\theta}),\, p_{i2}(\boldsymbol{\theta}),\, \ldots,\, 
		p_{iL}(\boldsymbol{\theta}),\, p_{i,L+1}(\boldsymbol{\theta}) \right)^\top, 
		\quad i = 1, 2, \ldots, R,
	\end{align*}
	and the corresponding empirical probability vector as
	\begin{align*}
		\widehat{\mathbf{p}}_i = \left( \frac{n_{i1}}{K_i},\, \frac{n_{i2}}{K_i},\, \ldots,\, 
		\frac{n_{iL}}{K_i},\, \frac{n_{i,L+1}}{K_i} \right)^\top, 
		\quad i = 1, 2, \ldots, R,
	\end{align*}
	where $K_i$ is the number of units assigned to the $i$-th stress condition. 

	The DPD between the probability vectors under the $i$-th cyclic-stress condition, $\widehat{\mathbf{p}}_i$ and $\mathbf{p}_i(\boldsymbol{\theta})$, for tuning parameter $\beta > 0$, is given by
	\begin{align}
		d_\beta(\widehat{\mathbf{p}}_i, \mathbf{p}_i(\boldsymbol{\theta})) 
		&= \sum_{j=1}^{L+1}\left\{
		p_{ij}(\boldsymbol{\theta})^{1+\beta} 
		- \left(1+\frac{1}{\beta}\right) \frac{n_{ij}}{K_i}
		p_{ij}(\boldsymbol{\theta})^{\beta}
		+ \frac{1}{\beta}
		\left(\frac{n_{ij}}{K_i}\right)^{1+\beta}
		\right\}. 
		\label{eq:DPD_full}
	\end{align}

    Combining the information from the $R$ cyclic-stress conditions, the weighted DPD is given by
    \begin{align}
		\text{DPD}_\beta(\boldsymbol{\theta}) = & \sum_{i=1}^{R} \frac{K_i}{K} d_\beta(\widehat{\mathbf{p}}_i, \mathbf{p}_i(\boldsymbol{\theta})) , &
		\label{eq:DPD}
	\end{align}
    %
	and the weighted minimum density power divergence estimator (WMDPDE) with tuning parameter $\beta$ is then defined as
	\begin{align}
		\widehat{\boldsymbol{\theta}}_\beta = (\widehat{\alpha}_{0,\beta}, \widehat{\alpha}_{1,\beta}, \widehat{\sigma}_\beta)  =\argmin_{\boldsymbol{\theta} \in \boldsymbol{\Theta}}\; \text{DPD}_\beta(\boldsymbol{\theta}).
		\label{eq:WMDPDE}
	\end{align}
    
    \cite{jaenada2026robust} showed that the DPD family can be extended to include $\beta=0$ as a limiting case, for which the WMDPDE coincides with the maximum likelihood estimator (MLE).
	Note that the last term in Eq. \eqref{eq:DPD_full}
	does not depend on $\boldsymbol{\theta}$ and therefore has no role in the minimization of $d_\beta(\widehat{\mathbf{p}}_i, \mathbf{p}_i(\boldsymbol{\theta}))$ 
	with respect to $\boldsymbol{\theta}$. 

	Setting the partial derivatives of the DPD loss funcion in \eqref{eq:DPD} to zero gives the estimating equations
    in matrix form as
	\begin{align}
		\sum_{i=1}^{R} \frac{K_i}{K}\,
		\mathbf{W}_i(\boldsymbol{\theta})^\top\,
		\mathbf{D}_i^{(\beta-1)}(\boldsymbol{\theta})\,
		\big[\mathbf{p}_i(\boldsymbol{\theta}) - \widehat{\mathbf{p}}_i\big] 
		= \mathbf{0}_3,
		\label{eq:estimeq}
	\end{align}
	where $\mathbf{0}_3$ is the $3$-dimensional null vector, $\mathbf{D}_i(\boldsymbol{\theta}) = 
	\mathrm{diag}(p_{i1}(\boldsymbol{\theta}),\ldots,p_{iL}(\boldsymbol{\theta}),p_{i,L+1}(\boldsymbol{\theta}))$ 
	is the $(L+1)\times(L+1)$ diagonal matrix of fitted probabilities, 
	$\mathbf{D}_i^{(r)}(\boldsymbol{\theta})$ denotes the diagonal matrix with entries 
	$p_{ij}(\boldsymbol{\theta})^{r}$, and $\mathbf{W}_i(\boldsymbol{\theta})$ is the 
	$(L+1)\times 3$ matrix whose $j$-th row is the gradient vector
	\begin{align*}
		\mathbf{w}_{ij}(\boldsymbol{\theta}) = \nabla_{\boldsymbol{\theta}} p_{ij}(\boldsymbol{\theta}) = 
		\left(\frac{\partial p_{ij}(\boldsymbol{\theta})}{\partial\alpha_0},\;
		\frac{\partial p_{ij}(\boldsymbol{\theta})}{\partial\alpha_1},\;
		\frac{\partial p_{ij}(\boldsymbol{\theta})}{\partial\sigma}\right)^\top,
		\quad j = 1,\ldots,L+1,
	\end{align*}
    \begin{align}
		\frac{\partial p_{ij}(\boldsymbol{\theta})}{\partial\alpha_0} 
		&= -\frac{1}{\sigma}\big[\phi(a_{i,j})-\phi(a_{i,j-1})\big],  &
		\frac{\partial p_{i,L+1}(\boldsymbol{\theta})}{\partial\alpha_0} 
		&= \phantom{-}\frac{\phi(a_{i,L})}{\sigma}, 
		\label{eq:dp_a0}\\
		\frac{\partial p_{ij}(\boldsymbol{\theta})}{\partial\alpha_1} 
		&= -\frac{\bar{s}_i}{\sigma}\big[\phi(a_{i,j})-\phi(a_{i,j-1})\big], &
		\frac{\partial p_{i,L+1}(\boldsymbol{\theta})}{\partial\alpha_1} 
		&= \phantom{-}\frac{\bar{s}_i\,\phi(a_{i,L})}{\sigma},
		\label{eq:dp_a1}\\
		\frac{\partial p_{ij}(\boldsymbol{\theta})}{\partial\sigma} 
		&= -\frac{1}{\sigma}\big[a_{i,j}\phi(a_{i,j})-a_{i,j-1}\phi(a_{i,j-1})\big], \quad &
		\frac{\partial p_{i,L+1}(\boldsymbol{\theta})}{\partial\sigma} 
		&= \phantom{-}\frac{a_{i,L}\,\phi(a_{i,L})}{\sigma}.
		\label{eq:dp_sigma}
	\end{align}
    with the convention that $p_{i,L+1}(\boldsymbol{\theta})$ corresponds to row $j = L+1$.
    
    
    The following result is established in 
    \cite{jaenada2026robust}. It is stated here 
    without proof, as it is the foundation for the 
    test statistics in Section \ref{sec:robust_tests}.
    
    \begin{theorem}[\cite{jaenada2026robust}]
    \label{thm:asymp}
    Under the CyALT model of 
    Section \ref{subsec:model_assump}, suppose 
    $\mathbf{J}_\beta(\boldsymbol{\theta}_0)$ is 
    non-singular and $p_{ij}(\boldsymbol{\theta})$ 
    is twice continuously differentiable in a 
    neighbourhood of $\boldsymbol{\theta}_0$. 
    Then $\widehat{\boldsymbol{\theta}}_\beta$ is 
    consistent for 
    $\boldsymbol{\theta}_0 = (\alpha_{0,0}, 
    \alpha_{1,0}, \sigma_0)^\top$ and
    \begin{equation}
        \sqrt{K}\,\big(
        \widehat{\boldsymbol{\theta}}_\beta - 
        \boldsymbol{\theta}_0
        \big)
        \;\xrightarrow[K \to \infty]{\mathcal{L}}\;
        N_3 \Big(\mathbf{0}_3,\;
        \mathbf{J}_\beta(\boldsymbol{\theta}_0)^{-1}
        \mathbf{K}_\beta(\boldsymbol{\theta}_0)
        \mathbf{J}_\beta(\boldsymbol{\theta}_0)^{-1}
        \Big),
        \label{eq:asymp}
    \end{equation}
    where
    \begin{align}
        \mathbf{J}_\beta(\boldsymbol{\theta}) &= 
        \sum_{i=1}^{R} \frac{K_i}{K}\,
        \mathbf{W}_i(\boldsymbol{\theta})^\top\,
        \mathbf{D}_i^{(\beta-1)}(\boldsymbol{\theta})\,
        \mathbf{W}_i(\boldsymbol{\theta}),
        \label{eq:Jbeta}\\[4pt]
        \mathbf{K}_\beta(\boldsymbol{\theta}) &= 
        \sum_{i=1}^{R} \frac{K_i}{K}\,
        \mathbf{W}_i(\boldsymbol{\theta})^\top
        \Big[
        \mathbf{D}_i^{(2\beta-1)}(\boldsymbol{\theta})
        - \mathbf{D}_i^{(\beta)}(\boldsymbol{\theta})\,
        \mathbf{1}_{L+1}
        \mathbf{1}_{L+1}^\top\,
        \mathbf{D}_i^{(\beta)}(\boldsymbol{\theta})
        \Big]
        \mathbf{W}_i(\boldsymbol{\theta}).
        \label{eq:Kbeta}
    \end{align}
    with $\mathbf{D}_i^{(r)}(\boldsymbol{\theta})$ denoting the $(L+1)\times (L+1)$ diagonal matrix with $j$-th diagonal entry $p_{ij}(\boldsymbol{\theta})^{r}$, and $\mathbf{1}_{L+1}$ the (L+1)-dimensional vector of ones.
    \end{theorem}
    
    We denote the components of $\widehat{\boldsymbol{\theta}}_\beta$ 
    corresponding to $\alpha_0$, $\alpha_1$, and $\sigma$ by 
    $\widehat{\alpha}_{0,\beta}$, $\widehat{\alpha}_{1,\beta}$, and 
    $\widehat{\sigma}_\beta$, respectively. Since the true parameter 
    vector $\boldsymbol{\theta}_0 = (\alpha_{0,0}, \alpha_{0,1}, \sigma_{0})^\top$ is unknown, the covariance matrix is 
    estimated by replacing $\boldsymbol{\theta}_0$ with 
    $\widehat{\boldsymbol{\theta}}_\beta$ in its mathematical form. 
    Let us denote the covariance estimate by 
    $\widehat{\boldsymbol{\Sigma}}_\beta$, then
    \begin{equation}
        \widehat{\boldsymbol{\Sigma}}_\beta = \frac{1}{K}\,
        \mathbf{J}_\beta(\widehat{\boldsymbol{\theta}}_\beta)^{-1}\,
        \mathbf{K}_\beta(\widehat{\boldsymbol{\theta}}_\beta)\,
        \mathbf{J}_\beta(\widehat{\boldsymbol{\theta}}_\beta)^{-1}.
        \label{eq:Sigma_hat}
    \end{equation}
    
    \section{Testing Hypotheses in CyALT}
    \label{sec:robust_tests}

   The main objective of this work is to develop robust hypothesis-testing procedures based on the WMDPDE. Our approach provides novel robust tests for hypotheses concerning the CyALT model parameters. 
   More specifically, we develop the theoretical foundations of two families of robust test statistics: Wald-type and Rao-type tests in sections \ref{sec:wald} and \ref{sec:rao}, respectively.
   Since the observations are independent but not identically distributed, the proposed procedures are formulated within the weighted DPD framework introduced in Section \ref{sec:WMDPDE}. Their theoretical properties are likewise established under this non-identically distributed setting.
    As noted in Section \ref{subsec:model_assump}, the parameter space of the CyALT model under log-normal lifetimes is
    \[
    \Theta
    =
    \left\{
    \boldsymbol{\theta}
    =
    (\alpha_0,\alpha_1,\sigma)^{T}
    :
    \alpha_0\in\mathbb{R},\;
    \alpha_1<0,\;
    \sigma>0
    \right\}
    =
    \mathbb{R}\times(-\infty,0)\times(0,\infty).
    \]
    
    We first consider simple null hypotheses in which the parameter vector is
    completely specified,   such as
    \begin{equation}
    H_0:
    (\alpha_0, \alpha_1, \sigma) =(\alpha_{0,0}, \alpha_{1,0}, \sigma_{0}) \quad \text{vs} \quad
    H_1:
    (\alpha_0, \alpha_1, \sigma) \neq (\alpha_{0,0}, \alpha_{1,0}, \sigma_{0}).
    \label{W1}
    \end{equation}
    
    Other simple testing problems arise when some of the model parameters are
    assumed to be known. For example, we may consider
    \begin{subequations}\label{W2}
    \begin{align}
    H_{0} &: \alpha_0=\alpha_{0,0} \quad \text{vs} \quad
    H_{1} : \alpha_0\neq\alpha_{0,0},
    \qquad
    \text{with $\alpha_1$ and $\sigma$ known},& 
    \label{W2a}\\
    H_0 &: \alpha_1=\alpha_{1,0} \quad \text{vs} \quad
    H_{1} : \alpha_1\neq\alpha_{1,0},
    \qquad
    \text{with $\alpha_0$ and $\sigma$ known}.&
    \label{W2b}
    \end{align}
    \end{subequations}
    Similarly, one may test
    \begin{equation}
    H_0:
    (\alpha_0, \alpha_1)=(\alpha_{0,0},\alpha_{1,0}) \quad \text{vs} \quad
    H_1: (\alpha_0, \alpha_1) \neq (\alpha_{0,0},\alpha_{1,0}),
    \qquad
    \text{with $\sigma$ known}.
    \label{W3}
    \end{equation}
    
    Secondly, we present the theory for composite null hypotheses in which the remaining
    parameters are unspecified, such as
    \begin{subequations}\label{W4}
    \begin{align}
    H_0 &: \alpha_0=\alpha_{0,0} \quad \text{vs} \quad
    H_1: \alpha_0\neq\alpha_{0,0},
    \qquad
    \text{with $\alpha_1$ and $\sigma$ unknown},&
    \label{W4a}\\
    H_0 &: \alpha_1=\alpha_{1,0} \quad \text{vs} \quad
    H_1: \alpha_1\neq\alpha_{1,0},
    \qquad
    \text{with $\alpha_0$ and $\sigma$ unknown}.&
    \label{W4b}\\
    H_0 &: \sigma=\sigma_{0} \quad \text{vs} \quad
    H_1: \sigma \neq \sigma_{0},
    \qquad
    \text{with $\alpha_0$ and $\alpha_1$ unknown}.&
    \label{W4c}
    \end{align}
    \end{subequations}
    More generally, we consider composite null hypotheses of the form
    \begin{equation}
    H_0:
    \boldsymbol{h}(\boldsymbol{\theta})
    =
    \boldsymbol{0}_{r} \quad \text{vs} \quad
     H_1: \boldsymbol{h}(\boldsymbol{\theta})
    \neq    \boldsymbol{0}_{r}
     \label{W5}
    \end{equation}
    where
    \[
    \boldsymbol{h}:\Theta\longrightarrow\mathbb{R}^{r},
    \qquad
    r\leq 3,
    \]
    is a sufficiently smooth function satisfying the regularity conditions
    specified below.
    The composite hypotheses in \eqref{W4a}–\eqref{W4c} are particularly useful for testing whether an individual model parameter takes a specified value. Of particular interest is the hypothesis $H_0:\alpha_1=0$, which assesses whether the lifetime distribution depends on the stress level.

    For brevity, throughout the remainder of the paper we refer to the null hypotheses given in \eqref{W1}, \eqref{W2a}, \eqref{W2b} and 
    \eqref{W3} as $H_{0}^{(1)}$, $H_0^{(2a)}$, $H_0^{(2b)}$, and $H_0^{(3)}$, respectively. That is, $H_0^{(1)}$ denotes the full simple null hypothesis 
    on $(\alpha_0,\alpha_1,\sigma)$; $H_0^{(2a)}$ and $H_0^{(2b)}$ denote the simple null hypotheses on $\alpha_0$ alone and $\alpha_1$ alone, respectively, with the remaining two parameters known; and $H_0^{(3)}$ denotes the simple null hypothesis on 
    $(\alpha_0,\alpha_1)$ jointly, with $\sigma$ known.
    
    %

    \section{Wald-type tests \label{sec:wald}}

    The classical Wald test provides an approach for testing hypotheses on model parameters by measuring the distance between an unrestricted estimator and the parameter values specified under the null hypothesis, scaled by its estimated asymptotic variance-covariance matrix. For a simple null hypothesis, the null parameter value is completely specified, so the Wald statistic directly compares the unrestricted estimator with this fixed value. In contrast, under a composite null hypothesis, the restrictions define a set of admissible parameter values rather than a single point, and the test is formulated in terms of the restriction functions evaluated at the unrestricted estimator. In this section, we present a robust generalization of the Wald test for testing CyALT parameters based on the DPD. We derive explicit formulas of the test under both simple and composite null hypotheses.
    
    \subsection{Simple null hypotheses}

    We first consider the case of simple null hypotheses. For brevity, we derive the corresponding expressions for hypotheses of the form \eqref{W1}. However, the theoretical framework developed here can be straightforwardly extended to the hypotheses in \eqref{W2} and \eqref{W3}.

    The classical Wald test, introduced by \cite{wald1943tests}, for testing the hypothesis given in (\ref{W1}) is given by
    \begin{equation*}
    W_{K}^{0}\left( \alpha _{0},\text{ }\alpha _{1},\sigma \right) = K \left[
    \left( \widehat{\alpha }_{0},\text{ }\widehat{\alpha }_{1},\widehat{\sigma }%
    \right) -\left( \alpha_{0,0},\text{ }\alpha_{1,0},\sigma_{0}\right) %
    \right] ^{T}\boldsymbol{I}_{F}\left( \alpha_{0,0},\text{ }\alpha_{1,0},
    \sigma_{0}\right) \left[ \left( \widehat{\alpha }_{0},\text{ }%
    \widehat{\alpha }_{1},\widehat{\sigma }\right) -\left( \alpha_{0,0},\text{
    }\alpha_{1,0},\sigma_{0}\right) \right] 
    \end{equation*}
    being $\left( \widehat{\alpha }_{0},\text{ }\widehat{\alpha }_{1},\widehat{%
    \sigma }\right)$ the MLE of $\left( \alpha _{0},\alpha _{1},\sigma \right) $
    and $\boldsymbol{I}_{F}\left( \alpha _{0},\text{ }\alpha _{1},\sigma \right) 
    $ the Fisher information matrix given by
    \begin{equation*}
    \boldsymbol{I}_{F}\left( \alpha _{0},\text{ }\alpha _{1},\sigma \right)
    =\tsum\limits_{i=1}^{R}\frac{K_{i}}{K}
    \boldsymbol{W}_{i}\left( \alpha _{0},\alpha _{1},\sigma \right)^{T}
    \boldsymbol{D}_{i}^{-1}
    \boldsymbol{W}_{i}\left( \alpha _{0},\alpha _{1},\sigma \right)
    \label{eq:J0}
    \end{equation*}
    The Fisher information matrix 
    is obtained from matrix $\boldsymbol{J}_{\beta }\left(\alpha_{0},\alpha_{1},\sigma \right)$ defined in \eqref{eq:Jbeta} by setting $\beta =0.$ 
    %
    
    It is well known that the asymptotic distribution of  $W_{K}^{0}\left( \alpha _{0},\alpha _{1},\sigma \right)$ is chi-square with $3$ degrees of 
    freedom. Therefore, the classical Wald test rejects the null hypothesis in (\ref{W1}) if
    \begin{equation}
    W_{K}^{0}\left( \alpha _{0},\alpha _{1},\sigma \right) > \chi_{3,\gamma}^{2},
    \label{W6}
    \end{equation}
    where $\chi _{3,\gamma }^{2}$ is the $100\left( 1-\gamma \right) $ upper
    quantile of the chi-square distribution with 3 degrees of freedom. In a
    similar manner Wald tests for testing (\ref{W2}) and (\ref{W3}) can be defined.
    For more details see, for instance, \cite{sen1994large}.
    
    The classical Wald test may be highly sensitive to the presence of outliers in the data. This limitation motivates the development of robust Wald-type tests based on the WMDPDE. This family of test statistics was originally proposed 
 in \cite{basu2016generalized} and \cite{ghosh} for general statistical models.
    
\begin{definition}
	The Wald-type test statistics, based on the WMDPDE, with tuning parameter $%
	\beta \geq 0,$ for testing the null hypothesis given in (\ref{W1}),
	are defined by%
    \begin{equation}
    W_{K}^{\beta }\left( \boldsymbol{\theta }_{0}\right) = K\left( \widehat{%
    \boldsymbol{\theta }}_{\beta }-\boldsymbol{\theta }_{0}\right) ^{T}\left[ 
    \boldsymbol{J}_{\beta }\left( \boldsymbol{\theta }_{0}\right) ^{-1}%
    \boldsymbol{K}_{\beta }\left( \boldsymbol{\theta }_{0}\right) \boldsymbol{J}%
    _{\beta }\left( \boldsymbol{\theta }_{0}\right) ^{-1}\right] ^{-1}\left( 
    \widehat{\boldsymbol{\theta }}_{\beta }-\boldsymbol{\theta }_{0}\right),
    \label{eq:waldtest}
    \end{equation}
     where $\widehat{\boldsymbol{\theta}}_{\beta} = 
(\widehat{\alpha}_{0,\beta},\, \widehat{\alpha}_{1,\beta},\, 
\widehat{\sigma}_\beta)^\top$ and $\boldsymbol{\theta}_{0} = 
(\alpha_{0,0},\, \alpha_{1,0},\, \sigma_0)^\top$.
    \end{definition}

Note the correspondence of the above formula with the classical Wald test. The central matrix
$
\boldsymbol{J}_{\beta}\left(\boldsymbol{\theta}_{0}\right)^{-1}$ $
\boldsymbol{K}_{\beta}\left(\boldsymbol{\theta}_{0}\right)
\boldsymbol{J}_{\beta}\left(\boldsymbol{\theta}_{0}\right)^{-1}
$
is the asymptotic variance–covariance matrix of the WMDPDE, as established in Theorem \ref{thm:asymp}. 
Building on this result and following the theoretical developments presented in \cite{basu2016generalized}, it can be proved that 
    the asymptotic distribution of the Wald-type test statistics, $W_{K}^{\beta
    }\left( \boldsymbol{\theta }_{0}\right),$ is a chi-square distribution with 
    $3$ degrees of freedom.
   %
Therefore, the Wald-type test statistics based on the WMDPDE with tuning parameter $\beta$ reject the null hypothesis in (\ref{W1}) if
    \begin{equation}
    W_{K}^{\beta
    }\left( \boldsymbol{\theta }_{0}\right) > \chi_{3,\gamma}^{2},
    \end{equation}
  where $\chi _{3,\gamma }^{2}$ is the $100\left( 1-\gamma \right) $ upper
    quantile of the chi-square distribution with 3 degrees of freedom.   
    
    \begin{remark}
        If we denote 
        \begin{equation*}
        \boldsymbol{J}_{\beta }\left( \boldsymbol{\theta }_{0}\right) =\left(
        J_{\beta }^{ij}\left( \boldsymbol{\theta }_{0}\right) _{i,j=1,2,3}\right) 
        \text{ and }\boldsymbol{K}_{\beta }\left( \boldsymbol{\theta }_{0}\right)
        =\left( K_{\beta }^{ij}\left( \boldsymbol{\theta }_{0}\right)
        _{i,j=1,2,3}\right)
        \end{equation*}%
        the Wald-type test statistics for testing the hypotheses presented in (\ref%
        {W2}) can be straightforwardly obtained as
        \begin{equation*}
        W_{K}^{\beta }\left( \alpha _{0}\right) = K \left( \widehat{\alpha }%
        _{0,\beta }-\alpha _{0,0}\right) ^{2}\frac{J_{\beta }^{11}\left( 
        \boldsymbol{\theta }_{0}\right)^2} {K_{\beta }^{11}\left( \boldsymbol{\theta }%
        _{0}\right) }\text{ and }W_{K}^{\beta }\left( \alpha _{1}\right) = K \left( 
        \widehat{\alpha }_{1,\beta }-\alpha _{1,0}\right)  ^{2}\frac{J_{\beta
        }^{22}\left( \boldsymbol{\theta }_{0}\right)^2} {K_{\beta }^{22}\left( 
        \boldsymbol{\theta }_{0}\right) },
        \end{equation*}%
        respectively, and they converge to a chi-square distribution with $1$ degree of freedom. Here, $\widehat{\alpha}_{0,\beta}$ and $\widehat{\alpha}_{1,\beta}$ are obtained by minimizing the DPD loss \eqref{eq:DPD} with respect to $\alpha_0$ and $\alpha_1$, respectively, each obtained with the remaining two parameters fixed at their known values. 
        \end{remark}

    \begin{remark}
    Let $\boldsymbol{J}_{\beta}^{1:2,1:2}(\boldsymbol{\theta}_0)$ and $\boldsymbol{K}_{\beta}^{1:2,1:2}(\boldsymbol{\theta}_0)$
    denote the $2\times2$ upper-left submatrices of $\boldsymbol{J}_{\beta}(\boldsymbol{\theta}_0)$ and
    $\boldsymbol{K}_{\beta}(\boldsymbol{\theta}_0)$ corresponding to $(\alpha_0,\alpha_1)$. The Wald-type test statistic for testing the hypothesis 
    presented in (\ref{W3}) is given by
    \begin{equation*}W_{K}^{\beta}\left(\alpha_0,\alpha_1\right) = K
    \left(\begin{array}{c}
    \widehat{\alpha}_{0,\beta}-\alpha_{0,0}\\
    \widehat{\alpha}_{1,\beta}-\alpha_{1,0}
    \end{array}\right)^{T}
    \left[\Big[\boldsymbol{J}_{\beta}^{1:2,1:2}(\boldsymbol{\theta}_0)\Big]^{-1}
    \boldsymbol{K}_{\beta}^{1:2,1:2}(\boldsymbol{\theta}_0)
    \Big[\boldsymbol{J}_{\beta}^{1:2,1:2}(\boldsymbol{\theta}_0)\Big]^{-1}\right]^{-1}
    \left(\begin{array}{c}
    \widehat{\alpha}_{0,\beta}-\alpha_{0,0}\\
    \widehat{\alpha}_{1,\beta}-\alpha_{1,0}
    \end{array}\right),
    \end{equation*}
    and it converges to a chi-square distribution with 2 degrees of freedom. Here, $\widehat{\alpha}_{0,\beta}$ and $\widehat{\alpha}_{1,\beta}$ denote the joint minimizers of the DPD loss \eqref{eq:DPD} with respect to $(\alpha_0,\alpha_1)$, obtained with $\sigma$ fixed at its known value.
    \end{remark}

    Moreover, this generalized test is consistent, meaning that its power function converges to 1 under the alternative hypothesis.
    We denote by $P_{W_{K}^{\beta }\left( \boldsymbol{\theta }_{0}\right)
    }\left( \boldsymbol{\theta }\right) $ the power of the Wald-type
    test given in (\ref{W1}) at the point $\boldsymbol{\theta }=\left(
    \alpha _{0}, \alpha _{1},\sigma\right) \in \Theta.$ The following theorem establishes the consistency of the Wald-type test
    given in (\ref{W1}) in the sense of \cite{Fraser}.

    \begin{theorem}
    \label{thm:wald_consistency}
    Let $ \boldsymbol{\theta }^{\ast}$ be the true value of the CyALT model parameter, which does not satisfy the simple null hypothesis in (\ref{W1}) with null parameter vector $\boldsymbol{\theta}_0 = (\alpha_{0,0}, \alpha_{1,0}, \sigma_o)^\top$.
    The Wald-type test statistic based on the WMDPDE with tuning parameter $\beta$ given in (\ref{eq:waldtest}) satisfies
    \begin{equation*}
    \lim_{K\rightarrow \infty }P_{W_{K}^{\beta }\left( \boldsymbol{\theta }%
    _{0}\right) }\left( \boldsymbol{\theta }^{\ast }\right) =1.
    \end{equation*}
    \end{theorem}
    \begin{proof}
    See Appendix~\ref{app:proofs}.
    \end{proof}

    \begin{remark}
    Based on the previous result we can get the minimum sample size for a
    fix power, $\pi ^{\ast }.$ It is a simple exercise to establish that the
    necessary sample size, $K,$ is given by 
    \begin{equation*}
    K=\left[ K^{\ast }\right] +1
    \end{equation*}%
    being%
    \begin{equation*}
    K^{\ast }=\frac{a+b+\sqrt{a(a+2b)}}{2L\left( \boldsymbol{\theta }^{\ast
    }\right) ^{2}}
    \end{equation*}%
    with%
    \begin{equation*}
    a=\sigma _{W_{K}^{\beta }\left( \boldsymbol{\theta }_{0}\right) }^{2}\left(
    \Phi ^{-1}\left( 1-\pi ^{\ast }\right) \right) ^{2}\text{ and }b=\frac{1}{2}%
    \chi _{3,\gamma }^{2}L\left( \boldsymbol{\theta }^{\ast }\right) .
    \end{equation*}
    \end{remark}
    
    \subsection{Composite null hypotheses}
    
    In this section, we introduce generalized Wald tests for testing composite null hypotheses. We denote by 
    \begin{equation*}
    \Theta =\left\{ \left( \alpha _{0},\alpha _{1},\sigma \right) ~|~ \alpha
    _{0}\in \mathbb{R}\text{, }\alpha _{1}<0,\sigma >0\right\} =\mathbb{R}\times \mathbb{R%
    }^{-}\times \mathbb{R}^{+}
    \end{equation*}%
    the parameter space associated to the log-normal model and we consider the
    restricted parameter space $\Theta _{0}\subset \Theta $ defined by 
    \begin{equation}
    \Theta _{0}=\left\{ \left( \alpha _{0},\alpha _{1},\sigma \right) \in \Theta
    ~|~ \boldsymbol{h}\left( \alpha _{0},\alpha _{1},\sigma \right) =\boldsymbol{0}%
    _{r}\right\}  \label{W.Null}
    \end{equation}%
    where 
    \begin{equation}
    \boldsymbol{h}\boldsymbol{:}\text{ }\mathbb{R\times R}^{-}\mathbb{\times R}%
    ^{+}\rightarrow \mathbb{R}^{r}\text{ }(r\leq 3)  \label{W7}
    \end{equation}%
    and $\boldsymbol{0}_{r}$ denotes the null vector of dimension $r.$ We also assume
    that the matrix
    \begin{equation}
    \boldsymbol{H}\left( \alpha _{0},\alpha _{1},\sigma \right) =\frac{\partial 
    \boldsymbol{h}\left( \alpha _{0},\alpha _{1},\sigma \right) ^{T}}{%
    \boldsymbol{\partial }\left( \alpha _{0},\alpha _{1},\sigma \right) }
    \label{W8}
    \end{equation}%
    exists and is continuous in $\left( \alpha _{0},\alpha _{1},\sigma \right) $
    with $\text{rank}\left( \boldsymbol{H}\left( \alpha _{0},\alpha _{1},\sigma \right)
    \right) =r.$ 
    
    Our interest is in testing 
    \begin{equation}
    H_{0}:\left( \alpha _{0},\alpha _{1},\sigma \right) \in \Theta _{0} \quad \text{vs} \quad H_{1}:\left( \alpha _{0},\alpha _{1},\sigma \right) \notin \Theta
    _{0}  \label{W1.5}
    \end{equation}%
    on the basis of a random sample of size $K$.
    
    Remark that, if we define $ \boldsymbol{h} \left( \alpha _{0}, \alpha _{1}, \sigma
    \right) = \left( \alpha _{0}-\alpha_{0,0}, \alpha _{1}-\alpha
    _{1,0}, \sigma - \sigma_{0}\right) ^{T} $ ($r=3$), we get the simple null hypothesis
    considered in (\ref{W1}). For $\boldsymbol{h}\left( \alpha _{0},\alpha
    _{1},\sigma \right) =\left( \alpha _{0}-\alpha_{0,0},\alpha _{1}-\alpha
    _{1,0}\right) ^{T}$ ($r=2$), we get a composite null hypothesis for the test \eqref{W3} but with unknown $\sigma.$
    For $\boldsymbol{h}\left( \alpha_0, \alpha_1, \sigma \right) = \alpha_0-\alpha_{0,0}$ ($r=1$), we get the composite null hypothesis considered in (\ref{W4a}).
    
    \begin{definition}
    \label{def:wald_composite}
    The family of Wald-type statistics based on the WMDPDE for testing (\ref{W1.5}) is given by 
    \begin{equation}
    W_{K}^{\beta }\left( \widehat{\boldsymbol{\theta }}_{\beta }\right) =K%
    \boldsymbol{h}\left( \widehat{\boldsymbol{\theta }}_{\beta }\right) ^{T}%
    \boldsymbol{\Sigma }_{\boldsymbol{h},\beta}^{-1} \left(\widehat{\boldsymbol{\theta }}_{\beta }\right)%
    ~\boldsymbol{h}\left( \widehat{\boldsymbol{\theta }}_{\beta }\right), \qquad \beta \geq 0,
    \label{W1.6}
    \end{equation}%
    being 
    \begin{equation*}
    \boldsymbol{\Sigma }_{\boldsymbol{h},\beta} (\widehat{\boldsymbol{\theta }}_{\beta })=\boldsymbol{H}%
    ^{T}(\widehat{\boldsymbol{\theta }}_{\beta })\boldsymbol{J}_{ \beta }^{-1}(%
    \widehat{\boldsymbol{\theta }}_{\beta })\boldsymbol{K}_{\beta }(\widehat{%
    \boldsymbol{\theta }}_{\beta })\boldsymbol{J}_{\beta}^{-1}(\widehat{%
    \boldsymbol{\theta }}_{\beta }) \boldsymbol{H} (\widehat{\boldsymbol{\theta }}%
    _{\beta })
    \end{equation*}%
    so that $K^{-1}\boldsymbol{\Sigma }_{\boldsymbol{h},\beta}(\widehat{\boldsymbol{\theta }}_{\beta })$ is the covariance matrix of $\boldsymbol{h} \left( \widehat{\boldsymbol{\theta }}_{\beta }\right)$ evaluated at $\widehat{\boldsymbol{\theta }}_{\beta }$, and the matrices $\boldsymbol{H}(\boldsymbol{\theta }),\boldsymbol{J}%
        _{\beta }\left( \boldsymbol{\theta }\right) $ and $\boldsymbol{K}_{\beta
        }\left( \boldsymbol{\theta }\right) $ were defined in \eqref{W8}, \eqref{eq:Jbeta} and
        \eqref{eq:Kbeta}, respectively and the function $\boldsymbol{h}$ in (\ref{W7}).
    \end{definition}
    
    \begin{theorem}
    The asymptotic distribution of the Wald-type statistics given in (\ref{W1.6}%
    ) is a chi-square distribution with $r$ degrees of freedom.
    \end{theorem}
    
    \begin{proof}
    See \cite{basu2016generalized}.
    \end{proof}
    
    \begin{remark}
    Based on the previous theorem, the Wald-type tests, with significance level $%
    \gamma $, reject the null hypothesis given in (\ref{W1.5}) if 
    \begin{equation}
    W_{K}^\beta\left( \widehat{\boldsymbol{\theta }}_{\beta }\right) >\chi _{r,\gamma
    }^{2}.  \label{W22}
    \end{equation}
    It can be established that the Wald-type test given in (\ref{W22}) is also
    consistent, i.e.,%
    \begin{equation*}
    \lim_{K\rightarrow \infty }P_{W_{K}^{\beta }\left( \widehat{\boldsymbol{%
    \theta }}_{\beta }\right) }\left( \boldsymbol{\theta }^{\ast }\right) =1.
    \end{equation*}
    \end{remark}

    \section{Rao-type tests \label{sec:rao}}

    The classical Rao score test provides a classical approach for testing hypotheses on model parameters based on the score function, $\boldsymbol{U}(\boldsymbol{\theta})$ defined as the gradient of the log-likelihood of the CyALT model with respect to the parameter vector $\boldsymbol{\theta} = (\alpha_0,\alpha_1,\sigma)^\top$. For a simple null hypothesis as in \eqref{W1}-\eqref{W3}, the parameter value is completely specified under the null, so the Rao statistic can be evaluated directly at this fixed value. In contrast, under a composite null hypothesis, some parameters remain unspecified and therefore the test requires their estimation subject to the restrictions imposed by the null hypothesis. Consequently, the score function and the corresponding information matrix are evaluated in the restricted version of the MLE obtained in the null. In this section, a robust generalization of the Rao score test for the case of simple and composite null hypothesis in CyALT test, based on the WMDPDE is presented. The asymptotic properties of the Rao-type test statistics is also demonstrated.

    \subsection{Simple null hypothesis}
    
    The classical Rao test, for testing the null
    hypothesis given in \eqref{W1} has the expression
    \begin{equation*}
    R\left( \boldsymbol{\theta }_{0}\right) =K\text{ }\boldsymbol{U}\left( 
    \boldsymbol{\theta }_{0}\right) I_{F}(\boldsymbol{\theta }_{0})^{-1}%
    \boldsymbol{U}\left( \boldsymbol{\theta }_{0}\right) ^{T}
    \end{equation*}%
    with
    \begin{equation*}
    \boldsymbol{U}\left( \boldsymbol{\theta }\right) =\tsum \limits_{i=1}^{R}%
    \frac{K_{i}}{K}\boldsymbol{W}_{i}\left( \boldsymbol{\theta }\right) 
    \boldsymbol{D}_{i}^{-1 }\left( \boldsymbol{\theta }\right) \left(\mathbf{p}_i(\boldsymbol{\theta}) - \widehat{\mathbf{p}}_i \right).
    \end{equation*}
    The function $\boldsymbol{U}\left( \boldsymbol{\theta }\right)$ is known as the score of the model likelihood, defined by the estimating equations of the MLE. The score function measures the local sensitivity of the log-likelihood to changes in the model parameters, indicating the direction and magnitude in which the likelihood increases from a given parameter value.
    Besides, note that under simple null hypotheses, the parameter value $\boldsymbol{\theta}_0$ is completely specified under the null, so the Rao statistic does not need an estimate of the CyALT model parameters. Besides,
    it is well-known that, under the null hypothesis given in \eqref{W1}, we have 
    \begin{equation*}
    R\left( \boldsymbol{\theta }_{0}\right) \underset{K\rightarrow \infty }{%
    \overset{L}{\rightarrow }}\chi_{3}^{2}.
    \end{equation*}
   
   Therefore, the classical Rao test rejects the null hypothesis in (\ref{W1}) if
    \begin{equation}
    R\left( \boldsymbol{\theta }_{0}\right) > \chi_{3,\gamma}^{2},
    \end{equation}
    where $\chi_{3,\gamma }^{2}$ is the $100\left( 1-\gamma \right) $ upper
    quantile of the chi-square distribution with 3 degrees of freedom. 
        The classical Rao test for the simple null hypothesis in \eqref{W2} and \eqref{W3} can be defined in a similar manner, and its asymptotic distributions are a chi-square distribution with 1 and 2 degrees of freedom, respectively.

    Although the Rao test for simple null hypothesis does not require a parameter estimate, it relies on the score function, which, like the likelihood function itself, can be sensitive to data contamination. In particular, large discrepancies between the expected and observed count of failures may substantially affect the score and, consequently, the value of the Rao test statistic. To mitigate this sensitivity to data contamination, we introduce a robust generalization of the Rao test based on the DPD score function.
    
    Following \cite{basu2022robust, basu2025statistical}, we define the Rao-type test based on WMDPDE for testing the hypothesis in \eqref{W1} as follows.
    \begin{definition}
    The Rao-type tests based on the WMDPDE for testing the null hypothesis given in \eqref{W1} are given by%
    \begin{equation*}
    R_{K}^{\beta }\left( \boldsymbol{\theta }_{0}\right) = K\boldsymbol{U}_{\beta
    }\left( \boldsymbol{\theta }_{0}\right) \boldsymbol{K}_{\beta }\left( 
    \boldsymbol{\theta }_{0}\right) ^{-1}\boldsymbol{U}_{\beta }\left( 
    \boldsymbol{\theta }_{0}\right) ^{T}
    \end{equation*}%
    with 
    \begin{equation*}
    \boldsymbol{U}_{\beta }\left( \boldsymbol{\theta }_{0}\right)
    =\tsum \limits_{i=1}^{R}\frac{K_{i}}{K}\boldsymbol{W}_{i}\left( \boldsymbol{%
    \theta }_{0}\right) \boldsymbol{D}_{i}^{\left( \beta -1\right) }\left( 
    \boldsymbol{\theta }_{0}\right) \left(\mathbf{p}_i(\boldsymbol{\theta}_0) - \widehat{\mathbf{p}}_i \right) .
    \end{equation*}
    \label{def:rao_simple}
    \end{definition}
    
    The function $\boldsymbol{U}_{\beta}(\boldsymbol{\theta}_0)$ is known as the $\beta$-score vector of the model. It is a three-dimensional random vector defined in terms of the first derivatives of the DPD loss in \eqref{eq:DPD}. In the following theorem, we derive the first moments of the $\beta$-score vector $\boldsymbol{U}_{\beta}(\boldsymbol{\theta}_0).$
    
    \begin{theorem}
    \label{thm:u_beta_mean_var}
    Under the null hypothesis given in \eqref{W1} we have, 
    \begin{equation*}
    E\left[ \boldsymbol{U}_{\beta }\left( \boldsymbol{\theta }_{0}\right) \right]
    =\boldsymbol{0}_{3}\text{ and }Cov\left[ \boldsymbol{U}_{\beta }\left( 
    \boldsymbol{\theta }_{0}\right) \right] =\frac{1}{K}\boldsymbol{K}_{\beta
    }\left( \boldsymbol{\theta }_{0}\right) .
    \end{equation*}
    \end{theorem}

    \begin{proof}
        See Appendix \ref{app:proofs}.
    \end{proof}
    
    \begin{remark}
    Applying the previous result, and the Limit Central Theorem we can stablish that
    \begin{equation*}
    \sqrt{K}\boldsymbol{U}_{\beta }\left( \boldsymbol{\theta }_{0}\right) 
    \underset{K\rightarrow \infty }{\overset{L}{\rightarrow }}N(\boldsymbol{0}%
    _{3},\boldsymbol{K}_{\beta }\left( \boldsymbol{\theta }_{0}\right) ).
    \end{equation*}
    \end{remark}

    Once the $\beta$-score distribution is stablished, we derived the asymptotic distribution of the generalized Rao-test based on the WMDPDE with tuning parameter $\beta.$
    \begin{theorem}
    \label{thm:rao_simple_chisq}
    Under the null hypothesis given in \eqref{W1}, we have%
    \begin{equation*}
    R_{K}^{\beta }\left( \boldsymbol{\theta }_{0}\right) \underset{K\rightarrow
    \infty }{\overset{L}{\rightarrow }}\chi _{3}^{2}
    \end{equation*}
    \end{theorem}
    
    \begin{proof}
        See Appendix \ref{app:proofs} 
    \end{proof}

    Therefore, the Rao-type test statistics with tuning parameter $\beta$ reject the null hypothesis in (\ref{W1}) if
    \begin{equation}
    R_{K}^{\beta
    }\left( \boldsymbol{\theta }_{0}\right) > \chi_{3,\gamma}^{2},
    \end{equation}
  where $\chi _{3,\gamma }^{2}$ is the $100\left( 1-\gamma \right) $ upper
    quantile of the chi-square distribution with 3 degrees of freedom.    
    
    \begin{remark}
        Denoting the components of $\boldsymbol{U}_{\beta}(\boldsymbol{\theta}_0)$ 
        by $U_{\beta}^{1}(\boldsymbol{\theta}_0)$, $U_{\beta}^{2}(\boldsymbol{\theta}_0)$, 
        and $U_{\beta}^{3}(\boldsymbol{\theta}_0)$, the Rao-type test statistics for 
        testing the hypotheses presented in (\ref{W2}) can be straightforwardly 
        obtained as
        \begin{equation*}
        R_{K}^{\beta}\left(\alpha_0\right) = 
        \frac{K\,\left(U_{\beta}^{1}(\boldsymbol{\theta}_0)\right)^2}
        {K_{\beta}^{11}(\boldsymbol{\theta}_0)}
        \text{ and }
        R_{K}^{\beta}\left(\alpha_1\right) = 
        \frac{K\,\left(U_{\beta}^{2}(\boldsymbol{\theta}_0)\right)^2}
        {K_{\beta}^{22}(\boldsymbol{\theta}_0)},
        \end{equation*}
        respectively, and they converge to a chi-square distribution 
        with 1 degree of freedom.
    \end{remark}

    \begin{remark}
    Let $\boldsymbol{U}_{\beta}^{1:2}(\boldsymbol{\theta}_0) = 
    \left(U_{\beta}^{1}(\boldsymbol{\theta}_0), 
    U_{\beta}^{2}(\boldsymbol{\theta}_0)\right)^\top$, and let $\boldsymbol{K}_{\beta}^{1:2,1:2}(\boldsymbol{\theta}_0)$
    denote the $2\times2$ upper-left submatrix of $\boldsymbol{K}_{\beta}(\boldsymbol{\theta}_0)$. The Rao-type 
    test statistics for testing the hypothesis presented in 
    (\ref{W3}) are given by
    \begin{equation*}
    R_{K}^{\beta}\left(\alpha_0,\alpha_1\right) = K\,
    \boldsymbol{U}_{\beta}^{1:2}(\boldsymbol{\theta}_0)^\top\,
    \Big[\boldsymbol{K}_{\beta}^{1:2,\,1:2}(\boldsymbol{\theta}_0)\Big]^{-1}\,
    \boldsymbol{U}_{\beta}^{1:2}(\boldsymbol{\theta}_0),
    \end{equation*}
    and it converges to a chi-square distribution with 2 degrees of freedom.
    \end{remark}

    \subsection{Composite null hypothesis}
    
    To define the Rao-type test for composite null hypothesis in \eqref{W5}, it is
    necessary to introduce the WMDPDE restricted to the null hypothesis of the
    form 
    \begin{equation*}
    \Theta _{0}=\left \{ \boldsymbol{\theta }\in \text{ }\Theta ~ | ~ \boldsymbol{h(\theta )}=\boldsymbol{0}_{r}\right \} .
    \end{equation*}
    As before, we denote 
    \begin{equation*}
    \boldsymbol{H}(\boldsymbol{\theta })\boldsymbol{=}\frac{\partial \boldsymbol{%
    h(\theta )}^{T}}{\partial \boldsymbol{\theta }}.
    \end{equation*}%
    The restricted WMDPDE (RWMDPDE), denoted by $\widetilde{\boldsymbol{\theta }}_{\beta }$, is defined as the minimizer of the DPD loss in \eqref{eq:DPD} over the parameter space $\Theta_0$ restricted by the null hypothesis,
    \begin{equation*}
    \widetilde{\boldsymbol{\theta }}_{\beta }=\arg \min_{\boldsymbol{\theta }%
    \text{ }\in \text{ }\Theta _{0}}\boldsymbol{H}_{\beta }(\boldsymbol{\theta }%
    ).
    \end{equation*}%
    In the next theorem presents the asymptotic distribution of the RWMDPDE, $\widetilde{\boldsymbol{\theta }}_{\beta }.$
    
    \begin{theorem}
    \label{thm:rmdpde_asymp}
    Let $\boldsymbol{\theta }_{0}$ denote the true value of the parameter, satisfying the null constraints 
    $\boldsymbol{h(\theta }_{0}\boldsymbol{)=0}_{r}$. The RMDPDE\ of $\boldsymbol{\theta ,}$ $\widetilde{\boldsymbol{\theta }}%
    _{\beta },$ obtained under the null hypothesis through the constraints $ \boldsymbol{h(\theta )=0}_{r}$ has the following asymptotic distribution
    \begin{equation*}
    \sqrt{K}\left( \boldsymbol{\widetilde{\boldsymbol{\theta }}_{\beta }-\theta }%
    _{0}\right) \underset{K\rightarrow \infty }{\overset{L}{\rightarrow }}N\left(
    \boldsymbol{0,\Sigma }\left( \boldsymbol{\theta }_{0}\right) \right)
    \end{equation*}%
    with 
    \begin{equation*}
    \boldsymbol{\Sigma }\left( \boldsymbol{\theta }_{0}\right) =P\boldsymbol{%
    \boldsymbol{(\boldsymbol{\theta }}}_{0}\boldsymbol{\boldsymbol{)}J}_{\beta
    }\left( \boldsymbol{\theta }_{0}\right) ^{-1}\boldsymbol{K}_{\beta }\left( 
    \boldsymbol{\theta }_{0}\right) P\boldsymbol{\boldsymbol{(\boldsymbol{\theta 
    }}}_{0}\boldsymbol{\boldsymbol{)}J}_{\beta }\left( \boldsymbol{\theta }%
    _{0}\right) ^{-1}.
    \end{equation*}
    \end{theorem}
    
    \begin{proof}
    See Appendix \ref{app:proofs}.
    \end{proof}
    
    \noindent  Now we are going to develop Rao-type tests for testing the composite null hypothesis $\Theta _{0}=\left \{ 
    \boldsymbol{\theta }\in \text{ }\Theta \text{ } |\text{ }\boldsymbol{h(\theta
    )}=\boldsymbol{0}_{r}\right \} $ given in \eqref{W5}.
    
    \begin{definition}
    For testing the null hypothesis given in \eqref{W5}, the Rao-type test statistics are given by 
    \begin{equation}
    R_{K}^{\beta }\left( \boldsymbol{\widetilde{\boldsymbol{\theta }}_{\beta }}%
    \right) =K\boldsymbol{U}_{\beta }\left( \boldsymbol{\widetilde{\boldsymbol{%
    \theta }}_{\beta }}\right) ^{T}\boldsymbol{Q}^{\beta }\left( \boldsymbol{%
    \widetilde{\boldsymbol{\theta }}_{\beta }}\right) \left[ \boldsymbol{Q}%
    ^{\beta }\left( \boldsymbol{\widetilde{\boldsymbol{\theta }}_{\beta }}%
    \right) ^{T}\boldsymbol{K}^{\beta }\left( \boldsymbol{\widetilde{\boldsymbol{%
    \theta }}_{\beta }}\right) \boldsymbol{Q}^{\beta }\left( \boldsymbol{%
    \widetilde{\boldsymbol{\theta }}_{\beta }}\right) \right] ^{-1}\boldsymbol{Q}%
    ^{\beta }\left( \boldsymbol{\widetilde{\boldsymbol{\theta }}_{\beta }}%
    \right) ^{T}\boldsymbol{U}_{\beta }\left( \boldsymbol{\widetilde{\boldsymbol{%
    \theta }}_{\beta }}\right)  \label{5}
    \end{equation}%
    with 
    \begin{equation}
    \boldsymbol{Q}^{\beta } \left( \boldsymbol{\theta }
    _{0}\right)  =\boldsymbol{J}_{\beta }\left( \boldsymbol{\theta }%
    _{0}\right) ^{-1}\boldsymbol{H\boldsymbol{(\boldsymbol{\theta }}}_{0}%
    \boldsymbol{\boldsymbol{)}}\left( \boldsymbol{H\boldsymbol{(\boldsymbol{%
    \theta }}}_{0}\boldsymbol{\boldsymbol{)}}^{T}\boldsymbol{\boldsymbol{J}%
    _{\beta }\left( \boldsymbol{\theta }_{0}\right) ^{-1}H\boldsymbol{(%
    \boldsymbol{\theta }}}_{0}\boldsymbol{\boldsymbol{)}}\right) ^{-1}  \label{6}
    \end{equation}
    \label{def:rao_composite}
    \end{definition}
    
    In the following theorem, we present the asymptotic distribution of Rao-type
    test statistic, $R_{K}^{\beta }\left( \boldsymbol{\widetilde{\boldsymbol{%
    \theta }}_{\beta }}\right) .$
    
    \begin{theorem}
    \label{thm:rao_composite_chisq}
    Under the null hypothesis, $H_{0}:\boldsymbol{h\boldsymbol{(\boldsymbol{%
    \theta })=0}}_{r},$ the Rao-type test statistic, $R_{K}^{\beta }\left( 
    \boldsymbol{\widetilde{\boldsymbol{\theta }}_{\beta }}\right) ~$converges in
    law $\ $to a chi-square distribution with $r$ degrees of freedom, i.e., 
    \begin{equation*}
    R_{K}^{\beta }\left( \boldsymbol{\widetilde{\boldsymbol{\theta }}_{\beta }}%
    \right) \underset{K\rightarrow \infty }{\overset{L}{\rightarrow }}\chi
    _{r}^{2}
    \end{equation*}%
    where $r$ is the number of independent constraints imposed by $\boldsymbol{h%
    \boldsymbol{(\boldsymbol{\theta })=0}}_{r}.$
    \end{theorem}

    \begin{proof}
    See Appendix \ref{app:proofs}.
    \end{proof}

    \section{Influence function for the Wald and Rao-type test statistics}
    \label{sec:IF}
    Let us denote $F_{\boldsymbol{\theta }}^{l}$ for the assumed distribution of 
    $l$-th stress condition with mass function $\mathbf{p}_{l}\left(\boldsymbol{\theta}\right) =\left( p_{l1}\left(\boldsymbol{\theta}\right) \right.$ $\left.,...,p_{lL+1}\left( \boldsymbol{\theta }\right) \right) ^{T},$ $l=1,...,R,$
    under the cyclic-stress ALT and interval monitoring with log-normal
    lifetimes and we also denote $\boldsymbol{F}_{\boldsymbol{\theta }}=\left(
    F_{\boldsymbol{\theta }}^{1},...,F_{\boldsymbol{\theta }}^{R}\right) .$ Let $%
    G^{l}$ denote the true distribution underlying the data with mass function $%
    \boldsymbol{g}^{l}=\left( g_{l1},...,g_{lR}\right) ^{T},l=1,...,R$ and we
    write, 
    \begin{equation*}
    \boldsymbol{G=}\left( G^{1},...,G^{R}\right) \text{ and }\boldsymbol{g}%
    =\left( \boldsymbol{g}^{1},...,\boldsymbol{g}^{R}\right) .
    \end{equation*}%
    The minimum density power divergence statistical functional,$\boldsymbol{\ T}%
    _{\beta }\left( \boldsymbol{G}\right) ,$ is defined as the minimizer of the
    weighted DPD between the probability vectors 
    \begin{equation*}
    \mathbf{p}_{l}\left( \boldsymbol{\theta }\right) =\left( p_{l1}\left( 
    \boldsymbol{\theta }\right) ,...,p_{lL+1}\left( \boldsymbol{\theta }\right)
    \right) ^{T}\text{ and }\boldsymbol{g}^{l}=\left( g_{l1},...,g_{lR}\right)
    ^{T},
    \end{equation*}%
    $l=1,...,R,$ i.e.$,$%
    \begin{equation*}
    \tsum\limits_{l=1}^{R}\frac{K_{l}}{K}d_{\beta }\left( \boldsymbol{g}^{l},%
    \boldsymbol{T}_{\beta }\left( \boldsymbol{G}\right) \right) =\min_{%
    \boldsymbol{\theta }\text{ }\in \text{ }\Theta }\tsum\limits_{l=1}^{R}\frac{%
    K_{l}}{K}d_{\beta }\left( \boldsymbol{g}^{l},\mathbf{p}_{l}\left( 
    \boldsymbol{\theta }\right) \right) ,
    \end{equation*}%
    Let $\boldsymbol{\theta }^{T}\boldsymbol{=T}_{\beta }\left(
    G^{1},...,G^{R}\right) $ be the minimum weighted density power divergence
    functional with contamination in all $R$ stress. Consider a contaminated
    version of \ the true lifetime distribution $G^{l}$ by 
    \begin{equation*}
    G_{\varepsilon }^{l}=(1-\varepsilon )G^{l}+\varepsilon \Delta _{t_{0}^{l}}
    \end{equation*}%
    with $\varepsilon $ the contamination proportion and $\Delta _{t_{0}^{l}}$
    being the degenerate distribution at the contamination point $t_{0}^{l}$. In the
    model under consideration we shall consider only a cell contamination in the 
    $l$-th stress condition, and so the contamination for $t_{0}^{l}$ should have
    all elements equal to zero except for only one component. In the following, we
    shall assume that%
    \begin{equation}
    G_{\varepsilon }^{l}=(1-\varepsilon )F_{\boldsymbol{\theta }%
    _{0}}^{l}+\varepsilon \Delta _{t_{0}^{l}}  \label{q}
    \end{equation}%
    being $\boldsymbol{\theta }_{0}$ the true value of the unknown parameter.
    
    We shall represent by 
    \begin{equation*}
    \boldsymbol{g}_{\varepsilon }^{l}=(1-\varepsilon )\boldsymbol{g}%
    ^{l}+\varepsilon \boldsymbol{e}_{t_{0}^{l}}
    \end{equation*}%
    the resulting multinomial vector associated to $G^{l}$ in which $\boldsymbol{%
    e}_{t_{0}^{l}}$ has all elements zero except for the interval containing the
    point perturbation $t_{0}^{l}.$ In accordance with (\ref{q}), we have 
    \begin{equation*}
    \boldsymbol{g}_{\varepsilon }^{l}=(1-\varepsilon )\mathbf{p}_{l}\left( 
    \boldsymbol{\theta }_{0}\right) +\varepsilon \boldsymbol{e}_{t_{0}^{l}}.
    \end{equation*}
    
    Let $\boldsymbol{\theta }_{l,\varepsilon }^{T}=\boldsymbol{T}_{\beta }\left(
    G^{1},.,G^{l-1},G_{\varepsilon }^{l},G^{l+1},.,G^{R}\right) $ be the minimum
    weighted density power divergence functional with contamination only in the $%
    l-$th stress condition, the influence function (IF) associated to the WMDPDE
    (see for more details \cite{jaenada2026robust})  is defined by%
    \begin{equation*}
    \text{IF}(t_{0}^{l},\boldsymbol{T}_{\beta },\boldsymbol{F}_{\boldsymbol{\theta }%
    _{0}})=\boldsymbol{J}_{\beta }\left( \boldsymbol{\theta }_{0}\right) ^{-1}%
    \frac{K_{l}}{K}\tsum\limits_{j=1}^{L+1}p_{lj}\left( \boldsymbol{\theta }%
    _{0}\right) ^{\beta -1}\left( \frac{\partial p_{lj}\left( \boldsymbol{\theta 
    }\right) }{\partial \boldsymbol{\theta }}\right) _{\boldsymbol{\theta }=%
    \boldsymbol{\theta }_{0}}\left( -p_{lj}\left( \boldsymbol{\theta }%
    _{0}\right) +\boldsymbol{e}_{t_{0}^{l}}\right) 
    \end{equation*}%
    and the IF for the WMDPDE in all the $R$ stress conditions is given by,%
    \begin{equation*}
    \text{IF}(t_{0}^{1},...,t_{0}^{R},\boldsymbol{T}_{\beta },\boldsymbol{F}_{%
    \boldsymbol{\theta }_{0}}))=\boldsymbol{J}_{\beta }\left( \boldsymbol{\theta 
    }_{0}\right) ^{-1}\tsum\limits_{l=1}^{R}\frac{K_{l}}{K}\tsum%
    \limits_{j=1}^{L+1}p_{lj}\left( \boldsymbol{\theta }_{0}\right) ^{\beta
    -1}\left( \frac{\partial p_{lj}\left( \boldsymbol{\theta }\right) }{\partial 
    \boldsymbol{\theta }}\right) _{\boldsymbol{\theta }=\boldsymbol{\theta }%
    _{0}}\left( -p_{lj}\left( \boldsymbol{\theta }_{0}\right) +\boldsymbol{e}%
    _{t_{0}^{l}}\right) .
    \end{equation*}
    
    \subsection{Influence function for the Wald-type test statistics}
    
    \subsubsection{Simple null hypothesis}
    
    The statistical functional corresponding to the Wald-type test statistics, $%
    W_{K,0}^{\beta }(\boldsymbol{\theta }_{0})$ , \ under the cyclic stept ALT, 
    intervals monitoring and log-normal lifetime, \ \ for testing simple and
    composite null hypothesis, respectively, are given by (ignoring multipliers
    K) by 
    \begin{equation*}
    W_{K,0}^{\beta }(\boldsymbol{G})=\left( \boldsymbol{T}_{\beta }(\boldsymbol{G%
    })-\boldsymbol{\theta }_{0}\right) ^{T}\boldsymbol{\Sigma }_{\beta }\left( 
    \boldsymbol{\theta }_{0}\right) ^{-1}\left( \boldsymbol{T}_{\beta }(%
    \boldsymbol{G})-\boldsymbol{\theta }_{0}\right) 
    \end{equation*}
    
    being $\boldsymbol{\Sigma }_{\beta }\left( \boldsymbol{\theta }_{0}\right) =%
    \boldsymbol{J}\left( \boldsymbol{\theta }_{0}\right) ^{-1}\boldsymbol{K}%
    \left( \boldsymbol{\theta }_{0}\right) \boldsymbol{J}\left( \boldsymbol{%
    \theta }_{0}\right) ^{-1}$. If we denote, 
    \begin{equation*}
    \boldsymbol{G}^{\varepsilon }\boldsymbol{=}\left(
    G_{1},...,G_{l-1},G_{l}^{\varepsilon },G_{l+1},..,G_{R}\right) 
    \end{equation*}%
    the first order influence function, with contamination only in the $l$-th
    stress condition, is defined by,%
    \begin{eqnarray*}
    \text{IF}(t_{0}^{l},\boldsymbol{W}_{K,0}^{\beta },\boldsymbol{F}_{\boldsymbol{%
    \theta }_{0}}) &=&\left( \frac{\partial W_{K,0}^{\beta }(\boldsymbol{G}%
    ^{\varepsilon })}{\partial \varepsilon }\right) _{\varepsilon =0}
    \end{eqnarray*}%
    \begin{eqnarray*}
    &=&2\left( \boldsymbol{T}_{\beta }(\boldsymbol{G}^{\varepsilon })-%
    \boldsymbol{\theta }_{0}\right) ^{T}\boldsymbol{\Sigma }_{\beta }\left( 
    \boldsymbol{\theta }_{0}\right) ^{-1}
    \end{eqnarray*}%
    which, when evaluated at the null hypothesis and contamination in the $l$-th
    direction, we get $\boldsymbol{T}_{\beta }(\boldsymbol{F}_{\boldsymbol{%
    \theta }_{0}})=\boldsymbol{\theta }_{0}.$ Therefore, $\text{IF}(t_{0}^{l},%
    \boldsymbol{W}_{K,0}^{\beta },\boldsymbol{F}_{\boldsymbol{\theta }_{0}})=0$
    and the IF analysis based on the first derivative of $W_{K,0}^{\beta }(%
    \boldsymbol{G}^{\varepsilon }$) is not adequate to quantify the robustness
    of this statistics. Then we are going to consider the second order IF of $%
    W_{K,0}^{\beta },$ i.e.,%
    \begin{equation*}
    \text{IF}_{2}(t_{0}^{l},\boldsymbol{W}_{K,0}^{\beta },\boldsymbol{F}_{\boldsymbol{%
    \theta }_{0}}))=\left( \frac{\partial ^{2}W_{K,0}^{\beta }(\boldsymbol{G}%
    ^{\varepsilon })}{\partial \varepsilon ^{2}}\right) _{\varepsilon =0}.
    \end{equation*}%
    But, 
    \begin{eqnarray*}
    \frac{\partial ^{2}W_{K,0}^{\beta }(\boldsymbol{G}_{\varepsilon })}{\partial
    \varepsilon ^{2}} &=&2\text{IF}(,\boldsymbol{T}_{\beta }(\boldsymbol{G}))%
    \boldsymbol{\Sigma }_{\beta }\left( \boldsymbol{\theta }_{0}\right)
    ^{-1}\text{IF}(t_{0}^{l},\boldsymbol{T}_{\beta }(\boldsymbol{G})) \\
    &&+2\left( \boldsymbol{T}_{\beta }(\boldsymbol{G})-\boldsymbol{\theta }%
    _{0}\right) ^{T}\boldsymbol{\Sigma }_{\beta }^{\ast }\left( \boldsymbol{%
    \theta }_{0}\right) ^{-1}\text{IF}(t_{0}^{l},\boldsymbol{T}_{\beta }(\boldsymbol{G})
    \end{eqnarray*}%
    and $\boldsymbol{T}_{\beta }(\boldsymbol{F}_{\boldsymbol{\theta }_{0}})=%
    \boldsymbol{\theta }_{0},$ therefore, 
    \begin{eqnarray*}
    \text{IF}_{2}(t_{0}^{l},\boldsymbol{W}_{K,0}^{\beta },\boldsymbol{F}_{\boldsymbol{%
    \theta }_{0}}) &=&\left( \frac{\partial ^{2}W_{K,0}^{\beta }(\boldsymbol{G}%
    _{\varepsilon })}{\partial \varepsilon ^{2}}\right) _{\varepsilon =0}. \\
    &=&2\text{IF}(t_{0}^{l},\boldsymbol{T}_{\beta },\boldsymbol{F}_{\boldsymbol{\theta }%
    _{0}})\boldsymbol{\Sigma }_{\beta }\left( \boldsymbol{\theta }_{0}\right)
    ^{-1}\text{IF}(t_{0}^{l},\text{IF}(t_{0}^{l},\boldsymbol{T}_{\beta },\boldsymbol{F}_{%
    \boldsymbol{\theta }_{0}}).
    \end{eqnarray*}%
    In a similar way we can derive the first and second order IF of $%
    W_{K,0}^{\beta }(\boldsymbol{G})$ for contamination in all directions $%
    \left( t_{0}^{1},...,t_{0}^{R}\right) $ at $\boldsymbol{G}=\boldsymbol{F}_{%
    \boldsymbol{\theta }_{0}}$ obtaining%
    \begin{eqnarray*}
    \text{IF}_{2}(t_{0}^{1},...,t_{0}^{R},\boldsymbol{W}_{K}^{\beta },,\boldsymbol{F}_{%
    \boldsymbol{\theta }_{0}}) &=&\left( \frac{\partial ^{2}W_{K}^{\beta }(%
    \boldsymbol{G}_{\varepsilon })}{\partial \varepsilon ^{2}}\right)
    _{\varepsilon =0} \\
    &=&2\text{IF}(t_{0}^{1},...,t_{0}^{R},\boldsymbol{T}_{\beta },\boldsymbol{F}_{%
    \boldsymbol{\theta }_{0}})\boldsymbol{\Sigma }_{\beta }\left( \boldsymbol{%
    \theta }_{0}\right) ^{-1}\text{IF}(t_{0}^{1},...,t_{0}^{R},\boldsymbol{T}_{\beta },%
    \boldsymbol{F}_{\boldsymbol{\theta }_{0}}).
    \end{eqnarray*}%
    
    We can derive in a very similar way the first and second order IFs of the
    Wlad-type functional $W_{K}^{\beta }(\boldsymbol{G})$ for composite null
    hypotheses. In the same way that previously the first order IF for
    contamination in either or all the directions are both identically zero,
    i.e., 
    \begin{equation*}
    \text{IF}(t_{0}^{l},\boldsymbol{W}_{K}^{\beta },\boldsymbol{G}))=0\text{ and }%
    \text{IF}(t_{0}^{1},...,t_{0}^{R},\boldsymbol{W}_{K}^{\beta },\boldsymbol{F}_{%
    \boldsymbol{\theta }_{0}}))=0.
    \end{equation*}%
    Therefore it is necessary to obtain the second order influence functions
    that are given by, 
    \begin{equation*}
    \text{IF}_{2}(t_{0}^{l},\boldsymbol{W}_{K}^{\beta },\boldsymbol{G}))=\left( \frac{%
    \partial ^{2}W_{K}^{\beta }(\boldsymbol{G}_{\varepsilon })}{\partial
    \varepsilon ^{2}}\right) _{\varepsilon =0}=2\text{IF}(t_{0}^{l},\boldsymbol{T}%
    _{\beta },\boldsymbol{F}_{\boldsymbol{\theta }_{0}}))\boldsymbol{\Sigma }%
    _{\beta }^{\ast }\left( \boldsymbol{\theta }_{0}\right) ^{-1}\text{IF}(t_{0}^{l},%
    \boldsymbol{T}_{\beta },\boldsymbol{F}_{\boldsymbol{\theta }_{0}}).
    \end{equation*}%
    being $\boldsymbol{\Sigma }_{\beta }^{\ast }\left( \boldsymbol{\theta }%
    _{0}\right) =\boldsymbol{H}\left( \boldsymbol{\theta }_{0}\right) ^{T}%
    \boldsymbol{J}\left( \boldsymbol{\theta }_{0}\right) ^{-1}\boldsymbol{K}%
    \left( \boldsymbol{\theta }_{0}\right) \boldsymbol{J}\left( \boldsymbol{%
    \theta }_{0}\right) ^{-1}\boldsymbol{H}\left( \boldsymbol{\theta }%
    _{0}\right)$,    
    and%
    \begin{equation*}
    \text{IF}_{2}(t_{0}^{1},...,t_{0}^{R},\boldsymbol{W}_{K}^{\beta },\boldsymbol{G}%
    )=\left( \frac{\partial ^{2}W_{K}^{\beta }(\boldsymbol{G}_{\varepsilon })}{%
    \partial \varepsilon ^{2}}\right) _{\varepsilon
    =0}=2\text{IF}((t_{0}^{1},...,t_{0}^{R},\boldsymbol{T}_{\beta },\boldsymbol{F}_{%
    \boldsymbol{\theta }_{0}}))\boldsymbol{\Sigma }_{\beta }^{\ast}\left( \boldsymbol{%
    \theta }_{0}\right) ^{-1}\text{IF}((t_{0}^{1},...,t_{0}^{R},\boldsymbol{T}_{\beta },%
    \boldsymbol{F}_{\boldsymbol{\theta }_{0}}).
    \end{equation*}%
    
    \subsubsection{Composite null hypothesis}
    
    Let $\boldsymbol{T}_{\beta }(\boldsymbol{G)}$ the functional associated to
    the WMDPDE and we consider the functional associate to the Wald-type tests
    for composite null hypothesis, 
    \begin{equation*}
    W_{\beta }(\boldsymbol{G}))=\boldsymbol{h}(\boldsymbol{T}_{\beta }(%
    \boldsymbol{G))}^{T}\left( \boldsymbol{H}(\boldsymbol{T}_{\beta }(%
    \boldsymbol{G))\Sigma }_{\beta }(\boldsymbol{T}_{\beta }(\boldsymbol{G)H}(%
    \boldsymbol{T}_{\beta }(\boldsymbol{G)}\right) ^{-1}\boldsymbol{h}(%
    \boldsymbol{T}_{\beta }(\boldsymbol{G)),}
    \end{equation*}%
    with  $\boldsymbol{\Sigma }_{\beta }\left( \boldsymbol{\theta }\right) =%
    \boldsymbol{J}\left( \boldsymbol{\theta }\right) ^{-1}\boldsymbol{K}\left( 
    \boldsymbol{\theta }\right) \boldsymbol{J}\left( \boldsymbol{\theta }\right)
    ^{-1}\boldsymbol{.}$
    
    Therefore the IF of the Wald-type test statistics under cyclic step ALT
    interval monitoring and log-normal lifetime with contamination only in the $l$%
    -th direction is given by,  
    \begin{equation*}
    \text{IF}(t_{0}^{l},\boldsymbol{W}_{K}^{\beta },\boldsymbol{F}_{\boldsymbol{\theta }%
    _{0}}))=\left( \frac{\partial W_{\beta }(\boldsymbol{G}_{\epsilon }))}{%
    \partial \varepsilon }\right) _{\varepsilon =0}.
    \end{equation*}%
    But, 
    \begin{eqnarray*}
    \frac{\partial W_{\beta }(\boldsymbol{G}_{\epsilon }))}{\partial \varepsilon 
    } &=&2\frac{\partial \boldsymbol{h}(\boldsymbol{T}_{\beta }(\boldsymbol{G}%
    _{\epsilon }\boldsymbol{))}^{T}}{\partial \varepsilon }\left( \boldsymbol{H}(%
    \boldsymbol{T}_{\beta }(\boldsymbol{G}_{\epsilon }\boldsymbol{))\Sigma }%
    _{\beta }(\boldsymbol{T}_{\beta }(\boldsymbol{G}_{\epsilon }\boldsymbol{)H}(%
    \boldsymbol{T}_{\beta }(\boldsymbol{G}_{\epsilon }\boldsymbol{)}\right) ^{-1}%
    \boldsymbol{h}(\boldsymbol{T}_{\beta }(\boldsymbol{G}_{\epsilon }\boldsymbol{%
    ))} \\
    &&+\boldsymbol{h}(\boldsymbol{T}_{\beta }(\boldsymbol{G}_{\epsilon }%
    \boldsymbol{))}^{T}\frac{\partial \left( \boldsymbol{H}(\boldsymbol{T}%
    _{\beta }(\boldsymbol{G}_{\epsilon }\boldsymbol{))\Sigma }_{\beta }(%
    \boldsymbol{T}_{\beta }(\boldsymbol{G}_{\epsilon }\boldsymbol{)H}(%
    \boldsymbol{T}_{\beta }(\boldsymbol{G}_{\epsilon }\boldsymbol{)}\right) ^{-1}%
    }{\partial \varepsilon }\boldsymbol{h}(\boldsymbol{T}_{\beta }(\boldsymbol{G}%
    _{\epsilon }\boldsymbol{))}
    \end{eqnarray*}%
    Now we have,%
    \begin{equation*}
    \left( \boldsymbol{h}(\boldsymbol{T}_{\beta }(\boldsymbol{G}_{\epsilon }%
    \boldsymbol{))}\right) _{\varepsilon =0}=\boldsymbol{h}(\boldsymbol{T}%
    _{\beta }(\boldsymbol{F}_{\theta _{0}}\boldsymbol{))=0.}
    \end{equation*}%
    Therefore, 
    \begin{equation*}
    \text{IF}(t_{0}^{l},\boldsymbol{W}_{K}^{\beta },\boldsymbol{F}_{\boldsymbol{\theta }%
    _{0}})=0
    \end{equation*}%
    and we must calculate the IF of order two, i.e., 
    \begin{equation*}
    \text{IF}_{2}(t_{0}^{l},\boldsymbol{W}_{K}^{\beta },\boldsymbol{F}_{\boldsymbol{%
    \theta }_{0}})=\left( \frac{\partial ^{2}W_{\beta }(\boldsymbol{G}_{\epsilon
    }))}{\partial \varepsilon ^{2}}\right) _{\varepsilon =0}.
    \end{equation*}%
    After some algebra we get 
    \begin{equation*}
    \text{IF}_{2}((t_{0}^{l},\boldsymbol{W}_{K}^{\beta },\boldsymbol{F}_{\boldsymbol{%
    \theta }_{0}}))=2\text{IF}(t_{0}^{l},\boldsymbol{T}_{\beta },\boldsymbol{F}_{%
    \boldsymbol{\theta }_{0}})\left( \boldsymbol{H}(\boldsymbol{\theta }_{0}%
    \boldsymbol{)\Sigma }_{\beta }(\boldsymbol{\theta }_{0}\boldsymbol{)H}(%
    \boldsymbol{\theta }_{0}\boldsymbol{)}\right) ^{-1}\text{IF}(t_{0}^{l},\boldsymbol{T%
    }_{\beta },\boldsymbol{F}_{\boldsymbol{\theta }_{0}}).
    \end{equation*}%
    In a similar way we can get the IF in all the directions, 
    \begin{equation*}
    \text{IF}_{2}(t_{0}^{1},...,t_{0}^{R},\boldsymbol{W}_{K}^{\beta },\boldsymbol{G}%
    )=\left( \frac{\partial ^{2}W_{K}^{\beta }(\boldsymbol{G})}{\partial
    \varepsilon ^{2}}\right) _{\varepsilon =0}=2\text{IF}((t_{0}^{1},...,t_{0}^{R},%
    \boldsymbol{T}_{\beta },\boldsymbol{F}_{\boldsymbol{\theta }_{0}}))%
    \boldsymbol{\Sigma }_{\beta }\left( \boldsymbol{\theta }_{0}\right)
    ^{-1}\text{IF}((t_{0}^{1},...,t_{0}^{R},\boldsymbol{T}_{\beta },\boldsymbol{F}_{%
    \boldsymbol{\theta }_{0}}).
    \end{equation*}
    
    \subsection{Influence function for the Rao-type test statistics}
    
    \subsubsection{Simple null hypothesis}
    
    We consider the functional%
    \begin{equation*}
    \boldsymbol{U}_{\beta }\left( \boldsymbol{G,\theta }_{0}\right) =\frac{1}{K}%
    \tsum\limits_{i=1}^{R}K_{i}\boldsymbol{W}_{i}\left( \boldsymbol{\theta }%
    _{0}\right) \boldsymbol{D}\left( \boldsymbol{\theta }_{0}\right) ^{\left(
    \beta -1\right) }\left( \boldsymbol{g}_{i}-\mathbf{p}_{l}\left( 
    \boldsymbol{\theta }\right) \right) 
    \end{equation*}%
    and we can observe that 
    \begin{equation*}
    E_{\boldsymbol{G}}\left[ \boldsymbol{U}_{\beta }\left( \boldsymbol{\theta }%
    _{0}\right) \right] =\boldsymbol{U}_{\beta }(\boldsymbol{G,\theta }_{0})
    \end{equation*}%
    and 
    \begin{equation*}
    E_{\boldsymbol{F}_{\boldsymbol{\theta }_{0}}}\left[ \boldsymbol{U}_{\beta
    }\left( \boldsymbol{\theta }_{0}\right) \right] =\boldsymbol{U}_{\beta }(%
    \boldsymbol{F}_{\boldsymbol{\theta }_{0}}\boldsymbol{,\theta }_{0})=%
    \boldsymbol{0}.
    \end{equation*}%
    The IF of the functional $\boldsymbol{U}_{\beta }(\boldsymbol{G,\theta }_{0})
    $ in the $l$-th condition is given by,%
    \begin{eqnarray*}
    \text{IF}\left( t_{0}^{l},\boldsymbol{U}_{\beta },\boldsymbol{F}_{\boldsymbol{%
    \theta }_{0}}\right)  &=&\left( \frac{\partial \boldsymbol{U}_{\beta }\left( 
    \boldsymbol{G}_{\varepsilon }\boldsymbol{,\theta }_{0}\right) }{\partial
    \varepsilon }\right) _{\varepsilon =0}=\frac{K_{l}}{K}\boldsymbol{W}%
    _{i}\left( \boldsymbol{\theta }_{0}\right) \boldsymbol{D}\left( \boldsymbol{%
    \theta }_{0}\right) ^{\left( \beta -1\right) }\left( \boldsymbol{e}%
    _{t_{0}^{l}}-\mathbf{p}_{l}\left( \boldsymbol{\theta }_{0}\right)
    \right)  \\
    &=&\boldsymbol{J}_{\beta }\left( \boldsymbol{\theta }_{0}\right) \text{IF}\left(
    t_{0}^{l},\boldsymbol{T}_{\beta },\boldsymbol{F}_{\boldsymbol{\theta }%
    _{0}}\right) .
    \end{eqnarray*}%
    On the basis of functional $\boldsymbol{U}_{\beta }(\boldsymbol{G,\theta }%
    _{0})$ $\ $we can define the Rao-type statistical functional by 
    \begin{equation*}
    R_{\beta }^{0}\left( \boldsymbol{G}\right) =\boldsymbol{U}_{\beta }(%
    \boldsymbol{G,\theta }_{0})^{T}\boldsymbol{K}_{\beta }\left( \boldsymbol{%
    \theta }_{0}\right) ^{-1}\boldsymbol{U}_{\beta }\boldsymbol{G,\theta }_{0}).
    \end{equation*}%
    Now the IF of the Rao-type test statistics, for testing the simple null hypothesis in the direction $l$-th is given by 
    \begin{equation*}
    \text{IF}(t_{0}^{l},R_{\beta }^{0},\boldsymbol{F}_{\boldsymbol{\theta }%
    _{0}})=\left( \frac{\partial R_{\beta }\left( \boldsymbol{G}_{\varepsilon
    }\right) }{\partial \varepsilon }\right) _{\varepsilon =0}=2\text{IF}\left(
    t_{0}^{l},\boldsymbol{U}_{\beta },\boldsymbol{F}_{\boldsymbol{\theta }%
    _{0}}\right) ^{T}\boldsymbol{K}_{\beta }\left( \boldsymbol{\theta }%
    _{0}\right) ^{-1}\boldsymbol{U}_{\beta }\left( \boldsymbol{F}_{\boldsymbol{%
    \theta }_{0}}\boldsymbol{,\theta }_{0}\right) =\boldsymbol{0}
    \end{equation*}%
    because $\boldsymbol{U}_{\beta }\left( \boldsymbol{F}_{\boldsymbol{\theta }%
    _{0}}\boldsymbol{,\theta }_{0}\right) =0.$ Therefore we must get the second
    order influence function, i.e,%
    \begin{eqnarray*}
    \text{IF}_{2}(t_{0}^{l},R_{\beta }^{0},\boldsymbol{F}_{\boldsymbol{\theta }_{0}})
    &=&\left( \frac{\partial ^{2}R_{\beta }\left( \boldsymbol{G}_{\varepsilon
    }\right) }{\partial \varepsilon ^{2}}\right) _{\varepsilon =0}=2\text{IF}\left(
    t_{0}^{l},\boldsymbol{U}_{\beta },\boldsymbol{F}_{\boldsymbol{\theta }%
    _{0}}\right) ^{T}\boldsymbol{K}_{\beta }\left( \boldsymbol{\theta }%
    _{0}\right) ^{-1}\text{IF}\left( t_{0}^{l},\boldsymbol{U}_{\beta },\boldsymbol{F}_{%
    \boldsymbol{\theta }_{0}}\right)  \\
    &=&2\text{IF}\left( t_{0}^{l},\boldsymbol{T}_{\beta },\boldsymbol{F}_{\boldsymbol{%
    \theta }_{0}}\right) ^{T}\boldsymbol{J}_{\beta }\left( \boldsymbol{\theta }%
    _{0}\right) \boldsymbol{K}_{\beta }\left( \boldsymbol{\theta }_{0}\right)
    ^{-1}\boldsymbol{J}_{\beta }\left( \boldsymbol{\theta }_{0}\right) \text{IF}\left(
    t_{0}^{l},\boldsymbol{T}_{\beta },\boldsymbol{F}_{\boldsymbol{\theta }%
    _{0}}\right) .
    \end{eqnarray*}%
    The last equality follows because,%
    \begin{equation*}
    \text{IF}\left( t_{0}^{l},\boldsymbol{U}_{\beta },\boldsymbol{F}_{\boldsymbol{%
    \theta }_{0}}\right) =\boldsymbol{J}_{\beta }\left( \boldsymbol{\theta }%
    _{0}\right) \text{IF}\left( t_{0}^{l},\boldsymbol{T}_{\beta },\boldsymbol{F}_{%
    \boldsymbol{\theta }_{0}}\right) .
    \end{equation*}
    \subsubsection{Composite null hypothesis}
    Before to get the IF of the Rao-type test for composite null hypothesis it
    is necessary to get the IF associated to the RMDPDE, $\widetilde{\boldsymbol{%
    \theta }}_{\beta }.$ The RMDPDE functional, $\widetilde{\boldsymbol{T}}%
    _{\beta },$ verifies, 
    
    \begin{equation*}
    \frac{\partial H_{\beta }(\widetilde{\boldsymbol{T}}_{\beta }\left( 
    \boldsymbol{G}_{\varepsilon }\right) }{\partial \boldsymbol{\theta }}+H(%
    \widetilde{\boldsymbol{T}}_{\beta }\left( \boldsymbol{G}_{\varepsilon
    }\right) )^{T}\boldsymbol{\lambda }_{\beta }\left( \boldsymbol{G}%
    _{\varepsilon }\right) =\boldsymbol{0}
    \end{equation*}%
    and 
    \begin{equation*}
    \boldsymbol{h}\left( \widetilde{\boldsymbol{T}}_{\beta }\left( \boldsymbol{G}%
    _{\varepsilon }\right) \right) =\boldsymbol{0.}
    \end{equation*}%
    We are going to get the derivative of the function  
    \begin{equation*}
    f_{1}\boldsymbol{(}\varepsilon \boldsymbol{)=H}(\widetilde{\boldsymbol{T}}%
    _{\beta }\left( \boldsymbol{G}_{\varepsilon }\right) )^{T}\boldsymbol{%
    \lambda }_{\beta }\left( \boldsymbol{G}_{\varepsilon }\right) 
    \end{equation*}%
    with respect to $\varepsilon $ and evaluate it at $\varepsilon =0,$%
    \begin{eqnarray*}
    \left( \frac{\partial f\boldsymbol{(\varepsilon )}}{\partial \varepsilon }%
    \right) _{\varepsilon =0} &=&\left( \frac{\partial \boldsymbol{H}(\widetilde{%
    \boldsymbol{T}}_{\beta }\left( \boldsymbol{G}_{\varepsilon }\right) )^{T}}{%
    \partial \varepsilon }\right) _{\varepsilon =0}\boldsymbol{\lambda }_{\beta
    }\left( \boldsymbol{F}_{\boldsymbol{\theta }_{0}}\right) +\boldsymbol{H}(%
    \boldsymbol{\theta }_{0})^{T}\left( \frac{\partial \boldsymbol{\lambda }%
    _{\beta }\left( \boldsymbol{G}_{\varepsilon }\right) }{\partial \varepsilon }%
    \right) _{\varepsilon =0} \\
    &=&\boldsymbol{H}(\boldsymbol{\theta }_{0})^{T}\text{IF}\left( t_{0}^{l},%
    \boldsymbol{\lambda }_{\beta },\boldsymbol{F}_{\boldsymbol{\theta }%
    _{0}}\right) .
    \end{eqnarray*}%
    Now we are going to get the derivative of the function, 
    \begin{equation*}
    f_{2}\left( \varepsilon \right) =\frac{\partial \boldsymbol{H}_{\beta }(%
    \widetilde{\boldsymbol{T}}_{\beta }\left( \boldsymbol{G}_{\varepsilon
    }\right) )}{\partial \boldsymbol{\theta }}
    \end{equation*}%
    evaluated at $\varepsilon =0.$ We have, 
    \begin{eqnarray*}
    \left( \frac{\partial f_{2}\left( \varepsilon \right) }{\partial \varepsilon 
    }\right) _{\varepsilon =0} &=&\left( \frac{\partial }{\partial \varepsilon }%
    \left( \frac{\partial \boldsymbol{H}_{\beta }(\widetilde{\boldsymbol{T}}%
    _{\beta }\left( \boldsymbol{G}_{\varepsilon }\right) }{\partial \boldsymbol{%
    \theta }}\right) \right) _{\varepsilon =0}+\left( \frac{\partial ^{2}%
    \boldsymbol{H}_{\beta }(\widetilde{\boldsymbol{T}}_{\beta }\left( 
    \boldsymbol{G}_{\varepsilon }\right) }{\partial \boldsymbol{\theta \theta }%
    ^{T}}\right) _{\boldsymbol{\theta =\theta }_{0}}\left( \frac{\partial 
    \widetilde{\boldsymbol{T}}_{\beta }\left( \boldsymbol{G}_{\varepsilon
    }\right) }{\partial \varepsilon }\right) _{\varepsilon =0} \\
    &=&-\boldsymbol{J}_{\beta }\left( \boldsymbol{\theta }_{0}\right)
    \text{IF}((t_{0}^{l},\boldsymbol{T}_{\beta },\boldsymbol{F}_{\boldsymbol{\theta }%
    _{0}})+\boldsymbol{J}_{\beta }\left( \boldsymbol{\theta }_{0}\right)
    \text{IF}((t_{0}^{l},\widetilde{\boldsymbol{T}}_{\beta },\boldsymbol{F}_{%
    \boldsymbol{\theta }_{0}}).
    \end{eqnarray*}%
    Therefore, 
    \begin{equation*}
    -\boldsymbol{J}_{\beta }\left( \boldsymbol{\theta }_{0}\right) \text{IF}(t_{0}^{l},%
    \boldsymbol{T}_{\beta },\boldsymbol{F}_{\boldsymbol{\theta }_{0}})+%
    \boldsymbol{J}_{\beta }\left( \boldsymbol{\theta }_{0}\right) \text{IF}((t_{0}^{l},%
    \widetilde{\boldsymbol{T}}_{\beta },\boldsymbol{F}_{\boldsymbol{\theta }%
    _{0}})+\boldsymbol{H}(\boldsymbol{\theta }_{0})^{T}\text{IF}\left( t_{0}^{l},%
    \boldsymbol{\lambda }_{\beta },\boldsymbol{F}_{\boldsymbol{\theta }%
    _{0}}\right) =\boldsymbol{0}
    \end{equation*}%
    and 
    \begin{equation}
    \text{IF}(t_{0}^{l},\widetilde{\boldsymbol{T}}_{\beta },\boldsymbol{F}_{\boldsymbol{%
    \theta }_{0}})=\text{IF}(t_{0}^{l},\boldsymbol{T}_{\beta },\boldsymbol{F}_{%
    \boldsymbol{\theta }_{0}})-\boldsymbol{J}_{\beta }\left( \boldsymbol{\theta }%
    _{0}\right) ^{-1}\boldsymbol{H}(\boldsymbol{\theta }_{0})^{T}\text{IF}\left(
    t_{0}^{l},\boldsymbol{\lambda }_{\beta },\boldsymbol{F}_{\boldsymbol{\theta }%
    _{0}}\right)   \label{A}
    \end{equation}%
    On the other hand, 
    \begin{equation}
    \boldsymbol{H}(\boldsymbol{\theta }_{0})\text{IF}(t_{0}^{l},\widetilde{\boldsymbol{T%
    }}_{\beta },\boldsymbol{F}_{\boldsymbol{\theta }_{0}})=\boldsymbol{0.}
    \label{B}
    \end{equation}%
    Therefore,%
    \begin{equation}
    \boldsymbol{H}(\boldsymbol{\theta }_{0})\text{IF}(t_{0}^{l},\widetilde{\boldsymbol{T%
    }}_{\beta },\boldsymbol{F}_{\boldsymbol{\theta }_{0}})=\boldsymbol{H}(%
    \boldsymbol{\theta }_{0})\text{IF}(t_{0}^{l},\boldsymbol{T}_{\beta },\boldsymbol{F}%
    _{\boldsymbol{\theta }_{0}})-\boldsymbol{H}(\boldsymbol{\theta }_{0})%
    \boldsymbol{J}_{\beta }\left( \boldsymbol{\theta }_{0}\right) ^{-1}%
    \boldsymbol{H}(\boldsymbol{\theta }_{0})^{T}\text{IF}\left( t_{0}^{l},\boldsymbol{%
    \lambda }_{\beta },\boldsymbol{F}_{\boldsymbol{\theta }_{0}}\right) 
    \label{C}
    \end{equation}%
    and 
    \begin{equation*}
    \boldsymbol{H}(\boldsymbol{\theta }_{0})\text{IF}(t_{0}^{l},\boldsymbol{T}_{\beta },%
    \boldsymbol{F}_{\boldsymbol{\theta }_{0}})-\boldsymbol{H}(\boldsymbol{\theta 
    }_{0})\boldsymbol{J}_{\beta }\left( \boldsymbol{\theta }_{0}\right) ^{-1}%
    \boldsymbol{H}(\boldsymbol{\theta }_{0})^{T}\text{IF}\left( t_{0}^{l},\boldsymbol{%
    \lambda }_{\beta },\boldsymbol{F}_{\boldsymbol{\theta }_{0}}\right) =%
    \boldsymbol{0}
    \end{equation*}%
    Then, 
    \begin{equation*}
    \text{IF}\left( t_{0}^{l},\boldsymbol{\lambda }_{\beta },\boldsymbol{F}_{%
    \boldsymbol{\theta }_{0}}\right) =\left( \boldsymbol{H}(\boldsymbol{\theta }%
    _{0})\boldsymbol{J}_{\beta }\left( \boldsymbol{\theta }_{0}\right) ^{-1}%
    \boldsymbol{H}(\boldsymbol{\theta }_{0})^{T}\right) ^{-1}\boldsymbol{H}(%
    \boldsymbol{\theta }_{0})\text{IF}((t_{0}^{l},\boldsymbol{T}_{\beta },\boldsymbol{F}%
    _{\boldsymbol{\theta }_{0}})
    \end{equation*}%
    and including the expression of $\text{IF}\left( t_{0}^{l},\boldsymbol{\lambda }%
    _{\beta },\boldsymbol{F}_{\boldsymbol{\theta }_{0}}\right) $ in (\ref{C}) we
    get%
    \begin{equation*}
    \text{IF}(t_{0}^{l},\widetilde{\boldsymbol{T}}_{\beta },\boldsymbol{F}_{\boldsymbol{%
    \theta }_{0}})=\left( \boldsymbol{I-J}_{\beta }\left( \boldsymbol{\theta }%
    _{0}\right) ^{-1}\boldsymbol{H}(\boldsymbol{\theta }_{0})^{T}\left( 
    \boldsymbol{H}(\boldsymbol{\theta }_{0})\boldsymbol{J}_{\beta }\left( 
    \boldsymbol{\theta }_{0}\right) ^{-1}\boldsymbol{H}(\boldsymbol{\theta }%
    _{0})^{T}\right) ^{-1}\boldsymbol{H}(\boldsymbol{\theta }_{0})\right)
    \text{IF}(t_{0}^{l},\boldsymbol{T}_{\beta },\boldsymbol{F}_{\boldsymbol{\theta }%
    _{0}}).
    \end{equation*}%
    Now we are going to get the IF of the Rao-type test statistics for composite
    null hypothesis. The Rao type test statistical functional re given by,%
    \begin{equation*}
    R_{\beta }(\boldsymbol{G)=U}_{\beta }(\boldsymbol{G,}\widetilde{\boldsymbol{T%
    }}_{\beta }\left( \boldsymbol{G}\right) )^{T}\boldsymbol{M}_{\beta }\left( 
    \widetilde{\boldsymbol{T}}_{\beta }\left( \boldsymbol{G}\right) \right) 
    \boldsymbol{U}_{\beta }(\boldsymbol{G,}\widetilde{\boldsymbol{T}}_{\beta
    }\left( \boldsymbol{G}\right) )
    \end{equation*}%
    with 
    \begin{equation*}
    \boldsymbol{M}_{\beta }\left( \boldsymbol{\theta }\right) =\boldsymbol{J}%
    _{\beta }\left( \boldsymbol{\theta }_{0}\right) ^{-1}\boldsymbol{H}(%
    \boldsymbol{\theta }_{0})^{T}\left[ \boldsymbol{H}(\boldsymbol{\theta }_{0})%
    \boldsymbol{J}_{\beta }\left( \boldsymbol{\theta }_{0}\right) ^{-1}%
    \boldsymbol{K}(\boldsymbol{\theta }_{0})\boldsymbol{J}_{\beta }\left( 
    \boldsymbol{\theta }_{0}\right) ^{-1}\boldsymbol{H}(\boldsymbol{\theta }%
    _{0})^{T}\right] ^{-1}\boldsymbol{H}(\boldsymbol{\theta }_{0})\boldsymbol{J}%
    _{\beta }\left( \boldsymbol{\theta }_{0}\right) ^{-1}.
    \end{equation*}%
    Now we have, 
    \begin{eqnarray*}
    \frac{\partial R_{\beta }(\boldsymbol{G}_{\varepsilon }\boldsymbol{)}}{%
    \partial \varepsilon } &=&\frac{\partial \boldsymbol{U}_{\beta }(\boldsymbol{%
    G}_{\varepsilon }\boldsymbol{,}\widetilde{\boldsymbol{T}}_{\beta }\left( 
    \boldsymbol{G}_{\varepsilon }\right) )^{T}}{\partial \varepsilon }%
    \boldsymbol{M}_{\beta }\left( \widetilde{\boldsymbol{T}}_{\beta }\left( 
    \boldsymbol{G}_{\varepsilon }\right) \right) \boldsymbol{U}_{\beta }(%
    \boldsymbol{G}_{\varepsilon }\boldsymbol{,}\widetilde{\boldsymbol{T}}_{\beta
    }\left( \boldsymbol{G}_{\varepsilon }\right) ) \\
    &&+\boldsymbol{U}_{\beta }(\boldsymbol{G}_{\varepsilon }\boldsymbol{,}%
    \widetilde{\boldsymbol{T}}_{\beta }\left( \boldsymbol{G}_{\varepsilon
    }\right) )^{T}\frac{\partial \boldsymbol{M}_{\beta }(\boldsymbol{G}%
    _{\varepsilon }\boldsymbol{,}\widetilde{\boldsymbol{T}}_{\beta }\left( 
    \boldsymbol{G}_{\varepsilon }\right) )}{\partial \varepsilon }\boldsymbol{U}%
    _{\beta }(\boldsymbol{G}_{\varepsilon }\boldsymbol{,}\widetilde{\boldsymbol{T%
    }}_{\beta }\left( \boldsymbol{G}_{\varepsilon }\right) ) \\
    &&+\boldsymbol{U}_{\beta }(\boldsymbol{G}_{\varepsilon }\boldsymbol{,}%
    \widetilde{\boldsymbol{T}}_{\beta }\left( \boldsymbol{G}_{\varepsilon
    }\right) )\boldsymbol{M}_{\beta }\left( \widetilde{\boldsymbol{T}}_{\beta
    }\left( \boldsymbol{G}_{\varepsilon }\right) \right) \frac{\partial 
    \boldsymbol{U}_{\beta }(\boldsymbol{G}_{\varepsilon }\boldsymbol{,}%
    \widetilde{\boldsymbol{T}}_{\beta }\left( \boldsymbol{G}_{\varepsilon
    }\right) )}{\partial \varepsilon } \\
    &=&2\boldsymbol{U}_{\beta }(\boldsymbol{G}_{\varepsilon }\boldsymbol{,}%
    \widetilde{\boldsymbol{T}}_{\beta }\left( \boldsymbol{G}_{\varepsilon
    }\right) )^{T}\boldsymbol{M}_{\beta }\left( \widetilde{\boldsymbol{T}}%
    _{\beta }\left( \boldsymbol{G}_{\varepsilon }\right) \right) \frac{\partial 
    \boldsymbol{U}_{\beta }(\boldsymbol{G}_{\varepsilon }\boldsymbol{,}%
    \widetilde{\boldsymbol{T}}_{\beta }\left( \boldsymbol{G}_{\varepsilon
    }\right) )}{\partial \varepsilon } \\
    &&+\boldsymbol{U}_{\beta }(\boldsymbol{G}_{\varepsilon }\boldsymbol{,}%
    \widetilde{\boldsymbol{T}}_{\beta }\left( \boldsymbol{G}_{\varepsilon
    }\right) )^{T}\frac{\partial \boldsymbol{M}_{\beta }(\boldsymbol{G}%
    _{\varepsilon }\boldsymbol{,}\widetilde{\boldsymbol{T}}_{\beta }\left( 
    \boldsymbol{G}_{\varepsilon }\right) )}{\partial \varepsilon }\boldsymbol{U}%
    _{\beta }(\boldsymbol{G}_{\varepsilon }\boldsymbol{,}\widetilde{\boldsymbol{T%
    }}_{\beta }\left( \boldsymbol{G}_{\varepsilon }\right) ).
    \end{eqnarray*}%
    But 
    \begin{equation*}
    \left( \boldsymbol{U}_{\beta }(\boldsymbol{G}_{\varepsilon }\boldsymbol{,}%
    \widetilde{\boldsymbol{T}}_{\beta }\left( \boldsymbol{G}\right) )\right)
    _{\varepsilon =0}=\boldsymbol{U}_{\beta }(\boldsymbol{F}_{\boldsymbol{\theta 
    }_{0}}\boldsymbol{,}\widetilde{\boldsymbol{T}}_{\beta }\left( \boldsymbol{F}%
    _{\boldsymbol{\theta }_{0}}\right) )=\boldsymbol{U}_{\beta }(\boldsymbol{F}_{%
    \boldsymbol{\theta }_{0}}\boldsymbol{,\theta }_{0})=\boldsymbol{0.}
    \end{equation*}%
    Therefore, 
    \begin{equation*}
    \text{IF}\left( t_{0}^{l},R_{\beta },\boldsymbol{F}_{\boldsymbol{\theta }%
    _{0}}\right) =\left( \frac{\partial R_{\beta }(\boldsymbol{G}_{\varepsilon }%
    \boldsymbol{)}}{\partial \varepsilon }\right) _{\varepsilon =0}=\boldsymbol{%
    0.}
    \end{equation*}%
    Then we are going to get the second order influence function. First we are
    going to get 
    \begin{equation*}
    \frac{\partial ^{2}R_{\beta }(\boldsymbol{G}_{\varepsilon }\boldsymbol{)}}{%
    \partial \varepsilon ^{2}}.
    \end{equation*}%
    We have,%
    \begin{eqnarray*}
    \frac{\partial ^{2}R_{\beta }(\boldsymbol{G}_{\varepsilon }\boldsymbol{)}}{%
    \partial \varepsilon ^{2}} &=&2\frac{\partial \boldsymbol{U}_{\beta }(%
    \boldsymbol{G}_{\varepsilon }\boldsymbol{,}\widetilde{\boldsymbol{T}}_{\beta
    }\left( \boldsymbol{G}_{\varepsilon }\right) )^{T}}{\partial \varepsilon }%
    \boldsymbol{M}_{\beta }\left( \widetilde{\boldsymbol{T}}_{\beta }\left( 
    \boldsymbol{G}_{\varepsilon }\right) \right) \frac{\partial \boldsymbol{U}%
    _{\beta }(\boldsymbol{G}_{\varepsilon }\boldsymbol{,}\widetilde{\boldsymbol{T%
    }}_{\beta }\left( \boldsymbol{G}_{\varepsilon }\right) )}{\partial
    \varepsilon } \\
    &&+2\boldsymbol{U}_{\beta }(\boldsymbol{G}_{\varepsilon }\boldsymbol{,}%
    \widetilde{\boldsymbol{T}}_{\beta }\left( \boldsymbol{G}_{\varepsilon
    }\right) )\frac{\partial \boldsymbol{M}_{\beta }(\boldsymbol{G}_{\varepsilon
    }\boldsymbol{,}\widetilde{\boldsymbol{T}}_{\beta }\left( \boldsymbol{G}%
    _{\varepsilon }\right) )}{\partial \varepsilon }\frac{\partial \boldsymbol{U}%
    _{\beta }(\boldsymbol{G}_{\varepsilon }\boldsymbol{,}\widetilde{\boldsymbol{T%
    }}_{\beta }\left( \boldsymbol{G}_{\varepsilon }\right) )}{\partial
    \varepsilon } 
	\end{eqnarray*}%
	\begin{eqnarray*}
    &&+2\boldsymbol{U}_{\beta }(\boldsymbol{G}_{\varepsilon }\boldsymbol{,}%
    \widetilde{\boldsymbol{T}}_{\beta }\left( \boldsymbol{G}_{\varepsilon
    }\right) )^{T}\boldsymbol{M}_{\beta }\left( \widetilde{\boldsymbol{T}}%
    _{\beta }\left( \boldsymbol{G}_{\varepsilon }\right) \right) \frac{\partial
    ^{2}\boldsymbol{U}_{\beta }(\boldsymbol{G}_{\varepsilon }\boldsymbol{,}%
    \widetilde{\boldsymbol{T}}_{\beta }\left( \boldsymbol{G}_{\varepsilon
    }\right) )}{\partial \varepsilon ^{2}} \\
    &&+\frac{\partial \boldsymbol{U}_{\beta }(\boldsymbol{G}_{\varepsilon }%
    \boldsymbol{,}\widetilde{\boldsymbol{T}}_{\beta }\left( \boldsymbol{G}%
    _{\varepsilon }\right) )^{T}}{\partial \varepsilon }\frac{\partial 
    \boldsymbol{M}_{\beta }(\boldsymbol{G}_{\varepsilon }\boldsymbol{,}%
    \widetilde{\boldsymbol{T}}_{\beta }\left( \boldsymbol{G}_{\varepsilon
    }\right) )}{\partial \varepsilon }\boldsymbol{U}_{\beta }(\boldsymbol{G}%
    _{\varepsilon }\boldsymbol{,}\widetilde{\boldsymbol{T}}_{\beta }\left( 
    \boldsymbol{G}_{\varepsilon }\right) )\\
    &&+\boldsymbol{U}_{\beta }(\boldsymbol{G}_{\varepsilon }\boldsymbol{,}%
    \widetilde{\boldsymbol{T}}_{\beta }\left( \boldsymbol{G}_{\varepsilon
    }\right) )^{T}\frac{\partial ^{2}\boldsymbol{M}_{\beta }(\boldsymbol{G}%
    _{\varepsilon }\boldsymbol{,}\widetilde{\boldsymbol{T}}_{\beta }\left( 
    \boldsymbol{G}_{\varepsilon }\right) )}{\partial \varepsilon ^{2}}%
    \boldsymbol{U}_{\beta }(\boldsymbol{G}_{\varepsilon }\boldsymbol{,}%
    \widetilde{\boldsymbol{T}}_{\beta }\left( \boldsymbol{G}_{\varepsilon
    }\right) ) \\
    &&+\boldsymbol{U}_{\beta }(\boldsymbol{G}_{\varepsilon }\boldsymbol{,}%
    \widetilde{\boldsymbol{T}}_{\beta }\left( \boldsymbol{G}_{\varepsilon
    }\right) )^{T}\frac{\partial \boldsymbol{M}_{\beta }(\boldsymbol{G}%
    _{\varepsilon }\boldsymbol{,}\widetilde{\boldsymbol{T}}_{\beta }\left( 
    \boldsymbol{G}_{\varepsilon }\right) )}{\partial \varepsilon }\frac{\partial 
    \boldsymbol{U}_{\beta }(\boldsymbol{G}_{\varepsilon }\boldsymbol{,}%
    \widetilde{\boldsymbol{T}}_{\beta }\left( \boldsymbol{G}_{\varepsilon
    }\right) )}{\partial \varepsilon }.
    \end{eqnarray*}%
    But 
    \begin{equation*}
    \left( \boldsymbol{U}_{\beta }(\boldsymbol{G}_{\varepsilon }\boldsymbol{,}%
    \widetilde{\boldsymbol{T}}_{\beta }\left( \boldsymbol{G}\right) )\right)
    _{\varepsilon =0}=\boldsymbol{U}_{\beta }(\boldsymbol{F}_{\boldsymbol{\theta 
    }_{0}}\boldsymbol{,}\widetilde{\boldsymbol{T}}_{\beta }\left( \boldsymbol{F}%
    _{\boldsymbol{\theta }_{0}}\right) )=\boldsymbol{U}_{\beta }(\boldsymbol{F}_{%
    \boldsymbol{\theta }_{0}}\boldsymbol{,\theta }_{0})=\boldsymbol{0.}
    \end{equation*}%
    Therefore, 
    \begin{equation*}
    \left( \frac{\partial ^{2}R_{\beta }(\boldsymbol{G}_{\varepsilon }%
    \boldsymbol{)}}{\partial \varepsilon ^{2}}\right) _{\epsilon =0}=2\left( 
    \frac{\partial \boldsymbol{U}_{\beta }(\boldsymbol{G}_{\varepsilon }%
    \boldsymbol{,}\widetilde{\boldsymbol{T}}_{\beta }\left( \boldsymbol{G}%
    _{\varepsilon }\right) )^{T}}{\partial \varepsilon }\boldsymbol{M}_{\beta
    }\left( \widetilde{\boldsymbol{T}}_{\beta }\left( \boldsymbol{G}%
    _{\varepsilon }\right) \right) \frac{\partial \boldsymbol{U}_{\beta }(%
    \boldsymbol{G}_{\varepsilon }\boldsymbol{,}\widetilde{\boldsymbol{T}}_{\beta
    }\left( \boldsymbol{G}_{\varepsilon }\right) )}{\partial \varepsilon }%
    \right) _{\varepsilon =0}.
    \end{equation*}%
    Now we denote 
    \begin{equation*}
    g(\varepsilon )=\boldsymbol{U}_{\beta }(\boldsymbol{G}_{\varepsilon }%
    \boldsymbol{,}\widetilde{\boldsymbol{T}}_{\beta }\left( \boldsymbol{G}%
    _{\varepsilon }\right) ),\\ \qquad
    g^{\prime }(\varepsilon )=\frac{\partial \boldsymbol{U}_{\beta }(\boldsymbol{%
    G}_{\varepsilon }\boldsymbol{,}\widetilde{\boldsymbol{T}}_{\beta }\left( 
    \boldsymbol{G}_{\varepsilon }\right) )}{\partial \varepsilon }+\frac{%
    \partial \boldsymbol{U}_{\beta }(\boldsymbol{G}_{\varepsilon }\boldsymbol{,}%
    \widetilde{\boldsymbol{T}}_{\beta }\left( \boldsymbol{G}_{\varepsilon
    }\right) )}{\partial \boldsymbol{\theta }}\frac{\partial \widetilde{%
    \boldsymbol{T}}_{\beta }\left( \boldsymbol{G}_{\varepsilon }\right) }{%
    \partial \varepsilon },
    \end{equation*}%
    and 
    \begin{eqnarray*}
    &\left( \frac{\partial \boldsymbol{U}_{\beta }(\boldsymbol{G}_{\varepsilon }%
    \boldsymbol{,}\widetilde{\boldsymbol{T}}_{\beta }\left( \boldsymbol{G}%
    _{\varepsilon }\right) )}{\partial \varepsilon }\right) _{\varepsilon =0}=%
    \boldsymbol{J}_{\beta }\left( \boldsymbol{\theta }_{0}\right) \text{IF}\left(
    t_{0}^{l},\boldsymbol{T}_{\beta },\boldsymbol{F}_{\boldsymbol{\theta }%
    _{0}}\right) & \\
    &\left( \frac{\partial \boldsymbol{U}_{\beta }(\boldsymbol{G}_{\varepsilon }%
    \boldsymbol{,}\widetilde{\boldsymbol{T}}_{\beta }\left( \boldsymbol{G}%
    _{\varepsilon }\right) )}{\partial \boldsymbol{\theta }}\right) _{%
    \boldsymbol{\theta =\theta }_{0}}=-\boldsymbol{J}_{\beta }\left( \boldsymbol{%
    \theta }_{0}\right) & \\
    &\frac{\partial \widetilde{\boldsymbol{T}}_{\beta }\left( \boldsymbol{G}%
    _{\varepsilon }\right) }{\partial \varepsilon }=\text{IF}(t_{0}^{l},\widetilde{%
    \boldsymbol{T}}_{\beta },\boldsymbol{F}_{\boldsymbol{\theta }_{0}}).&
    \end{eqnarray*}%
    Then 
    \begin{eqnarray}
    \text{IF}_{2}\left( t_{0}^{l},R_{\beta },\boldsymbol{F}_{\boldsymbol{\theta }%
    _{0}}\right)  &=&\left( \frac{\partial ^{2}R_{\beta }(\boldsymbol{G}%
    _{\varepsilon }\boldsymbol{)}}{\partial \varepsilon ^{2}}\right)
    _{\varepsilon =0}  \label{D} \\
    &=&\left[ \text{IF}\left( t_{0}^{l},\boldsymbol{T}_{\beta },\boldsymbol{F}_{%
    \boldsymbol{\theta }_{0}}\right) -F(t_{0}^{l},\widetilde{\boldsymbol{T}}%
    _{\beta },\boldsymbol{F}_{\boldsymbol{\theta }_{0}})\right] ^{T}\boldsymbol{J%
    }_{\beta }\left( \boldsymbol{\theta }_{0}\right) \boldsymbol{M}_{\beta
    }\left( \boldsymbol{\theta }_{0}\right) \boldsymbol{J}_{\beta }\left( 
    \boldsymbol{\theta }_{0}\right)   \notag \\
    &&\left[ \text{IF}\left( t_{0}^{l},\boldsymbol{T}_{\beta },\boldsymbol{F}_{%
    \boldsymbol{\theta }_{0}}\right) -F(t_{0}^{l},\widetilde{\boldsymbol{T}}%
    _{\beta },\boldsymbol{F}_{\boldsymbol{\theta }_{0}})\right] .  \notag
    \end{eqnarray}
    We established 
    \begin{equation*}
    \text{IF}_{2}(t_{0}^{l},\widetilde{\boldsymbol{T}}_{\beta },\boldsymbol{F}_{%
    \boldsymbol{\theta }_{0}})=\boldsymbol{J}_{\beta }\left( \boldsymbol{\theta }%
    _{0}\right) \left( \boldsymbol{I-J}_{\beta }\left( \boldsymbol{\theta }%
    _{0}\right) ^{-1}\boldsymbol{H}(\boldsymbol{\theta }_{0})^{T}\left( 
    \boldsymbol{H}(\boldsymbol{\theta }_{0})\boldsymbol{J}_{\beta }\left( 
    \boldsymbol{\theta }_{0}\right) ^{-1}\boldsymbol{H}(\boldsymbol{\theta }%
    _{0})^{T}\right) ^{-1}\boldsymbol{H}(\boldsymbol{\theta }_{0})\right)
    \text{IF}(t_{0}^{l},\boldsymbol{T}_{\beta },\boldsymbol{F}_{\boldsymbol{\theta }%
    _{0}}).
    \end{equation*}%
    Therefore, 
    \begin{eqnarray}
      & &\text{IF}\left( t_{0}^{l},\boldsymbol{T}_{\beta },\boldsymbol{F}_{%
    \boldsymbol{\theta }_{0}}\right) -F(t_{0}^{l},\widetilde{\boldsymbol{T}}%
    _{\beta },\boldsymbol{F}_{\boldsymbol{\theta }_{0}})  \nonumber
	\end{eqnarray}
	\begin{eqnarray}    
    &=&\text{IF}\left(
    t_{0}^{l},\boldsymbol{T}_{\beta },\boldsymbol{F}_{\boldsymbol{\theta }%
    _{0}}\right)   \nonumber \\
    &&-\left( \boldsymbol{I-J}_{\beta }\left( \boldsymbol{\theta }_{0}\right)
    ^{-1}\boldsymbol{H}(\boldsymbol{\theta }_{0})^{T}\left( \boldsymbol{H}(%
    \boldsymbol{\theta }_{0})\boldsymbol{J}_{\beta }\left( \boldsymbol{\theta }%
    _{0}\right) ^{-1}\boldsymbol{H}(\boldsymbol{\theta }_{0})^{T}\right) ^{-1}%
    \boldsymbol{H}(\boldsymbol{\theta }_{0})\right) 
    \times \text{IF}(t_{0}^{l},\boldsymbol{T}_{\beta },\boldsymbol{F}_{\boldsymbol{\theta }%
    _{0}}) \\
    &=&\boldsymbol{J}_{\beta }\left( \boldsymbol{\theta }_{0}\right) ^{-1}%
    \boldsymbol{H}(\boldsymbol{\theta }_{0})^{T}\left( \boldsymbol{H}(%
    \boldsymbol{\theta }_{0})\boldsymbol{J}_{\beta }\left( \boldsymbol{\theta }%
    _{0}\right) ^{-1}\boldsymbol{H}(\boldsymbol{\theta }_{0})^{T}\right) ^{-1}%
    \boldsymbol{H}(\boldsymbol{\theta }_{0})  \times \text{IF}(t_{0}^{l},\boldsymbol{T}_{\beta },\boldsymbol{F}_{\boldsymbol{\theta }%
    _{0}})
    \end{eqnarray}
    and 
    \begin{align*}
    \boldsymbol{J}_{\beta }&\left( \boldsymbol{\theta }_{0}\right) \left(
    \text{IF}\left( t_{0}^{l},\boldsymbol{T}_{\beta },\boldsymbol{F}_{\boldsymbol{%
    \theta }_{0}}\right) -F(t_{0}^{l},\widetilde{\boldsymbol{T}}_{\beta },%
    \boldsymbol{F}_{\boldsymbol{\theta }_{0}})\right) &\\
    &=\boldsymbol{H}(%
    \boldsymbol{\theta }_{0})^{T}\left( \boldsymbol{H}(\boldsymbol{\theta }_{0})%
    \boldsymbol{J}_{\beta }\left( \boldsymbol{\theta }_{0}\right) ^{-1}%
    \boldsymbol{H}(\boldsymbol{\theta }_{0})^{T}\right) ^{-1}\boldsymbol{H}(%
    \boldsymbol{\theta }_{0}) \times \text{IF}(t_{0}^{l},\boldsymbol{T}_{\beta },\boldsymbol{F}%
    _{\boldsymbol{\theta }_{0}}).&
    \end{align*}
    Then,    
    \begin{eqnarray*}
    \text{IF}_{2}\left( t_{0}^{l},R_{\beta },\boldsymbol{F}_{\boldsymbol{\theta }%
    _{0}}\right)  &=&\text{IF}(t_{0}^{l},\boldsymbol{T}_{\beta },\boldsymbol{F}_{%
    \boldsymbol{\theta }_{0}})\boldsymbol{H}(\boldsymbol{\theta }_{0})\left( 
    \boldsymbol{H}(\boldsymbol{\theta }_{0})\boldsymbol{J}_{\beta }\left( 
    \boldsymbol{\theta }_{0}\right) ^{-1}\boldsymbol{H}(\boldsymbol{\theta }%
    _{0})^{T}\right) ^{-1}\boldsymbol{H}(\boldsymbol{\theta }_{0})^{T}%
    \boldsymbol{M}_{\beta }\left( \boldsymbol{\theta }_{0}\right)  \\
    &&\times \boldsymbol{H}(\boldsymbol{\theta }_{0})^{T}\left( \boldsymbol{H}(%
    \boldsymbol{\theta }_{0})\boldsymbol{J}_{\beta }\left( \boldsymbol{\theta }%
    _{0}\right) ^{-1}\boldsymbol{H}(\boldsymbol{\theta }_{0})^{T}\right) ^{-1}%
    \boldsymbol{H}(\boldsymbol{\theta }_{0})\times \text{IF}(t_{0}^{l},\boldsymbol{T}_{\beta },%
    \boldsymbol{F}_{\boldsymbol{\theta }_{0}}).
    \end{eqnarray*}%
    But 
    \begin{equation*}
    \boldsymbol{M}_{\beta }\left( \boldsymbol{\theta }\right) =\boldsymbol{J}%
    _{\beta }\left( \boldsymbol{\theta }_{0}\right) ^{-1}\boldsymbol{H}(%
    \boldsymbol{\theta }_{0})^{T}\left[ \boldsymbol{H}(\boldsymbol{\theta }_{0})%
    \boldsymbol{J}_{\beta }\left( \boldsymbol{\theta }_{0}\right) ^{-1}%
    \boldsymbol{K}(\boldsymbol{\theta }_{0})\boldsymbol{J}_{\beta }\left( 
    \boldsymbol{\theta }_{0}\right) ^{-1}\boldsymbol{H}(\boldsymbol{\theta }%
    _{0})^{T}\right] ^{-1}\boldsymbol{H}(\boldsymbol{\theta }_{0})\boldsymbol{J}%
    _{\beta }\left( \boldsymbol{\theta }_{0}\right) ^{-1},
    \end{equation*}%
    and 
    \begin{eqnarray*}
    \text{IF}_{2}\left( t_{0}^{l},R_{\beta },\boldsymbol{F}_{\boldsymbol{\theta }%
    _{0}}\right)  &=&\text{IF}(t_{0}^{l},\boldsymbol{T}_{\beta },\boldsymbol{F}_{%
    \boldsymbol{\theta }_{0}})\boldsymbol{H}(\boldsymbol{\theta }_{0})\left( 
    \boldsymbol{H}(\boldsymbol{\theta }_{0})\boldsymbol{J}_{\beta }\left( 
    \boldsymbol{\theta }_{0}\right) ^{-1}\boldsymbol{H}(\boldsymbol{\theta }%
    _{0})^{T}\right) ^{-1}\boldsymbol{H}(\boldsymbol{\theta }_{0})^{T} \\
    &&\times\boldsymbol{J}_{\beta }\left( \boldsymbol{\theta }_{0}\right) ^{-1}%
    \boldsymbol{H}(\boldsymbol{\theta }_{0})^{T}\left[ \boldsymbol{H}(%
    \boldsymbol{\theta }_{0})\boldsymbol{J}_{\beta }\left( \boldsymbol{\theta }%
    _{0}\right) ^{-1}\boldsymbol{K}(\boldsymbol{\theta }_{0})\boldsymbol{J}%
    _{\beta }\left( \boldsymbol{\theta }_{0}\right) ^{-1}\boldsymbol{H}(%
    \boldsymbol{\theta }_{0})^{T}\right] ^{-1}\boldsymbol{H}(\boldsymbol{\theta }%
    _{0})\boldsymbol{J}_{\beta }\left( \boldsymbol{\theta }_{0}\right) ^{-1} \\
    &&\boldsymbol{H}(\boldsymbol{\theta }_{0})^{T} \times \left( \boldsymbol{H}(%
    \boldsymbol{\theta }_{0})\boldsymbol{J}_{\beta }\left( \boldsymbol{\theta }%
    _{0}\right) ^{-1}\boldsymbol{H}(\boldsymbol{\theta }_{0})^{T}\right) ^{-1}%
    \boldsymbol{H}(\boldsymbol{\theta }_{0}) \times \text{IF}(t_{0}^{l},\boldsymbol{T}_{\beta },%
    \boldsymbol{F}_{\boldsymbol{\theta }_{0}}) \\
    &=&\text{IF}(t_{0}^{l},\boldsymbol{T}_{\beta },\boldsymbol{F}_{\boldsymbol{\theta }%
    _{0}})\boldsymbol{H}(\boldsymbol{\theta }_{0})^{T}\left[ \boldsymbol{H}(%
    \boldsymbol{\theta }_{0})\boldsymbol{J}_{\beta }\left( \boldsymbol{\theta }%
    _{0}\right) ^{-1}\boldsymbol{K}(\boldsymbol{\theta }_{0})\boldsymbol{J}%
    _{\beta }\left( \boldsymbol{\theta }_{0}\right) ^{-1}\boldsymbol{H}(%
    \boldsymbol{\theta }_{0})^{T}\right] ^{-1} \\
    &&\boldsymbol{H}(\boldsymbol{\theta }_{0}) \times \text{IF}(t_{0}^{l},\boldsymbol{T}_{\beta
    },\boldsymbol{F}_{\boldsymbol{\theta }_{0}}).
    \end{eqnarray*}%
    In a similar way can be obtained $\text{IF}_{2}\left(
    t_{0}^{1},...,t_{0}^{R},R_{\beta },\boldsymbol{F}_{\boldsymbol{\theta }%
    _{0}}\right).$ 
    
    We can observe that all the IF associated to Wald and Rao-type tests are
    quadratic forms and assuming that the matrices of the different quadratic
    forms are bounded the corresponding IF will be bounded if the IF of the
    WMDPDE are bounded.
    
    \section{Simulation study}
    \label{sec:sim_study}
    \subsection{Experimental design}
    \label{subsec:sim_design}
    
    We examine the performance of the Wald-type and Rao-type tests developed in Sections \ref{sec:wald} and \ref{sec:rao} through a Monte Carlo simulation study. A CyALT design similar to that of \cite{jaenada2026robust} is considered. As shown in Figure \ref{fig:sim_design}, the design has two cyclic-stress conditions ($R=2$), with $K_1=120$ and $K_2=80$ units in each group. There are $L=6$ inspection times, $IT=(15,25,35,50,65,80)$, with $\tau=0.40$. The stress levels are $s_{1F}=s_{2F}=0.40$, $s_{1C}=0.65$, and $s_{2C}=1.00$.
    
     \begin{figure}[H]
         \centering
         \resizebox{1\textwidth}{7.5cm}{%
         \begin{tikzpicture}[>=stealth, scale=1.0]

%
%

\fill[gray!13]  (0,0.0) rectangle (12,1.2);   

\foreach \y in {0.0, 1.2, 2.4, 3.9, 6.0}
  \draw[gray!35, thin, dashed] (0,\y) -- (12,\y);

\draw[gray!55, thick]
  (0,0)--(0,1.2)--(0.8,1.2)--(0.8,0)--(2,0)
  --(2,1.2)--(2.8,1.2)--(2.8,0)--(4,0)
  --(4,1.2)--(4.8,1.2)--(4.8,0)--(6,0)
  --(6,1.2)--(6.8,1.2)--(6.8,0)--(8,0)
  --(8,1.2)--(8.8,1.2)--(8.8,0)--(10,0)
  --(10,1.2)--(10.8,1.2)--(10.8,0)--(12,0);

\draw[blue!90!black, very thick]
  (-0.06,2.4)--(-0.06,3.9)--(0.74,3.9)--(0.74,2.4)--(1.94,2.4)
  --(1.94,3.9)--(2.74,3.9)--(2.74,2.4)--(3.94,2.4)
  --(3.94,3.9)--(4.74,3.9)--(4.74,2.4)--(5.94,2.4)
  --(5.94,3.9)--(6.74,3.9)--(6.74,2.4)--(7.94,2.4)
  --(7.94,3.9)--(8.74,3.9)--(8.74,2.4)--(9.94,2.4)
  --(9.94,3.9)--(10.74,3.9)--(10.74,2.4)--(11.94,2.4);

\draw[red!90!black, very thick]
  (0,2.4)--(0,6.0)--(0.8,6.0)--(0.8,2.4)--(2,2.4)
  --(2,6.0)--(2.8,6.0)--(2.8,2.4)--(4,2.4)
  --(4,6.0)--(4.8,6.0)--(4.8,2.4)--(6,2.4)
  --(6,6.0)--(6.8,6.0)--(6.8,2.4)--(8,2.4)
  --(8,6.0)--(8.8,6.0)--(8.8,2.4)--(10,2.4)
  --(10,6.0)--(10.8,6.0)--(10.8,2.4)--(12,2.4);

\node[font=\large, gray!65] at (12.65, 0.6)  {$\cdots$};
\node[font=\large, gray!65] at (12.65, 3.15) {$\cdots$};
\node[font=\large, gray!65] at (12.65, 4.50) {$\cdots$};

\draw[thick, ->] (-0.3,0) -- (13.4,0)
  node[right, font=\small] {Cycle};
\draw[thick, ->] (-0.3,0) -- (-0.3,7.0)
  node[above, font=\small] {Stress $s$};

\foreach \y/\lab in {
  0.0/{$0.00$},
  1.2/{$0.20$},
  2.4/{$0.40$},
  3.9/{$0.65$},
  6.0/{$1.00$}}
{
  \draw (-0.45,\y)--(-0.3,\y);
  \node[left, font=\small] at (-0.48,\y) {\lab};
}

\draw[<->, gray!60, thin] (0,-0.55)--(0.8,-0.55)
  node[midway, below, font=\footnotesize] {$\tau T$};
\draw[<->, gray!60, thin] (0.8,-0.55)--(2,-0.55)
  node[midway, below, font=\footnotesize] {$(1{-}\tau)T$};
\draw[gray!40, thin, dotted] (0.8,-0.80)--(0.8,6.5);  
\draw[decorate,
  decoration={brace, mirror, amplitude=5pt, raise=3pt}]
  (0,-0.85)--(2,-0.85)
  node[midway, below=7pt, font=\footnotesize] {One cycle of $T$ units};

\foreach \xx in {2.25, 3.75, 5.25, 7.50, 9.75, 12.0}
  \draw[gray!45, thin, dashed] (\xx,0)--(\xx,6.5);

\draw[thick, ->] (0,-2.0)--(13.5,-2.0)
  node[right, font=\small] {$t$ (cycles)};
\node[below, font=\footnotesize] at (0,-2.0) {$0$};

\foreach \xx/\lbl/\t in {
  2.25/{$IT_1$}/{15},
  3.75/{$IT_2$}/{25},
  5.25/{$IT_3$}/{35},
  7.50/{$IT_4$}/{50},
  9.75/{$IT_5$}/{65},
  12.0/{$IT_6$}/{$80=t_c$}}
{
  \draw[thick] (\xx,-1.88)--(\xx,-2.12);
  \node[above, font=\footnotesize] at (\xx,-1.88) {\lbl};
  \node[below, font=\footnotesize] at (\xx,-2.12) {\t};
}

\node[right, font=\footnotesize, gray!70]   at (12.0,0.0)  {$s_{0F}$};
\node[right, font=\footnotesize, gray!70]   at (12.0,1.2)  {$s_{0C}$};
\node[right, font=\footnotesize, black!60]  at (12.0,2.4)  {$s_F$};
\node[right, font=\footnotesize, blue!90!black] at (12.0,3.9) {$s_{1C} \rightarrow Cond. 1:\;K_1{=}40,60,\textbf{80},100,120,140,160 $};
\node[right, font=\footnotesize, red!90!black]  at (12.0,6.0) {$s_{2C}\rightarrow Cond. 2:\;K_2{=}60,90,\textbf{120},150,180,210,240$};


\end{tikzpicture}}
         \caption{}
         \label{fig:sim_design}
     \end{figure}
    
    The true parameter vector is $\boldsymbol{\theta}_0=(5.0,-2.0,0.5)^\top$, and the tuning parameter takes values $\beta \in \{0, 0.2, 0.4, 0.6, 0.8, 1.0\}$, where $\beta=0$ yields the classical Wald and Rao tests. All results are based on $B=1000$ Monte Carlo replications at significance level $\alpha=0.05$. We study the four simple null hypotheses given in \eqref{W1}-\eqref{W3}, denoted $H_0^{(1)}$, $H_0^{(2a)}$, $H_0^{(2b)}$, and $H_0^{(3)}$, as defined in Section \ref{sec:robust_tests}. The empirical level and power are computed as
    \begin{equation}
        \widehat{\alpha} = \frac{1}{B}\sum_{b=1}^{B} 
        \mathbf{1}\!\left[T_K^{\beta,(b)} > 
        \chi^2_{r,\,0.05}\right]\bigg|_{\boldsymbol{\theta} 
        = \boldsymbol{\theta}_0},
        \qquad
        \widehat{\pi} = \frac{1}{B}\sum_{b=1}^{B} 
        \mathbf{1}\!\left[T_K^{\beta,(b)} > 
        \chi^2_{r,\,0.05}\right]\bigg|_{\boldsymbol{\theta} 
        = \boldsymbol{\theta}^*},
        \label{eq:emp_level_power}
    \end{equation}
    where $T_K^\beta$ denotes the relevant test statistic (Wald-type or Rao-type) and $r$ is the degrees of freedom for the corresponding hypothesis.
    
    We index $\boldsymbol\theta^*$ by parameter: $\boldsymbol{\theta}^*_1$ for $H_0^{(1)}$, $\boldsymbol{\theta}^*_{2a}$ for the hypotheses on $\alpha_0$, $\boldsymbol{\theta}^*_{2b}$ for the hypotheses on $\alpha_1$, and $\boldsymbol{\theta}^*_3$ for $H_0^{(3)}$. The values used are    
    \begin{align*}
        & \boldsymbol{\theta}^*_1 = (5.2, -2.2, 0.48)^\top, &    \\
        &\boldsymbol{\theta}^*_{2a} = (5.2, -2.0, 0.5)^\top,  \quad
        \boldsymbol{\theta}^*_{2b} = (5.0, -2.2, 0.5)^\top, &    \\
        &\boldsymbol{\theta}^*_3 = (5.2, -2.2, 0.5)^\top. &
    \end{align*}
    Since these alternative values are chosen for power calculation and not on the basis of whether the null hypothesis is simple or composite, the same vectors $\boldsymbol{\theta}^*_{2a}$ and $\boldsymbol{\theta}^*_{2b}$, and the same simulated datasets generated from them, are reused for the composite hypotheses $H_0^{(4a)}$ and $H_0^{(4b)}$ studied in subsubsections (c)--(d) below.
    
    We assess the behavior of the proposed tests as sample size increases and as contamination increases, 
    with the results discussed in detail in Sections \ref{subsec:sim_wald_results} and 
    \ref{subsec:sim_rao_results}. For this purpose we consider sample sizes $K \in \{ 100,150,200,250,300,350,400 \}$ and contamination proportion $\varepsilon$ is varied over the range from $0$ to $0.10$, following the same contamination scheme used by \cite{jaenada2026robust} for parameter estimation in the CyALT model. Contamination is introduced in the first three inspection intervals: for 
    cyclic-stress condition 1, a proportion $(1-\varepsilon)$ of units follows the true multinomial model 
    with probability vector $\mathbf{p}_1(\boldsymbol{\theta})$, while the remaining proportion $\varepsilon$ 
    is drawn uniformly from the first three inspection intervals; failure counts for cyclic-stress condition 
    2 are always generated from the true model. The contamination proportion is capped at $\varepsilon = 
    0.10$, since larger values would change the underlying model itself rather than testing robustness to a 
    small fraction of outlying observations.

    \subsection{Wald-type tests}
    \label{subsec:sim_wald_results}

    For each hypothesis, the Wald-type test statistic is computed as described in Section \ref{sec:wald}. For $H_1$ \eqref{W1}, $\widehat{\boldsymbol{\Sigma}}_\beta$ is evaluated at the null $\boldsymbol{\theta}_0$, and the test statistic is
    \begin{equation}
        W_K^\beta(\boldsymbol{\theta}_0) = 
        (\widehat{\boldsymbol{\theta}}_\beta - \boldsymbol{\theta}_0)^\top 
        \widehat{\boldsymbol{\Sigma}}_\beta(\boldsymbol{\theta}_0)^{-1} 
        (\widehat{\boldsymbol{\theta}}_\beta - \boldsymbol{\theta}_0),
        \label{eq:wald_H1_sim}
    \end{equation}
    compared to $\chi^2_{3,\,0.05}$. 
    If we denote 
        \begin{equation*}
        \boldsymbol{J}_{\beta}(\boldsymbol{\theta}_0) = \left(
        J_{\beta}^{ij}(\boldsymbol{\theta}_0) \right)_{i,j=1,2,3}
        \text{ and }
        \boldsymbol{K}_{\beta}(\boldsymbol{\theta}_0) = \left(
        K_{\beta}^{ij}(\boldsymbol{\theta}_0) \right)_{i,j=1,2,3},
        \end{equation*} for $H_{2a}$ and $H_{2b}$ (\eqref{W2a}-\eqref{W2b}), the Wald test statistics reduce to
    \begin{equation}
        W_K^\beta(\alpha_0) = 
        K(\widehat{\alpha}_{0,\beta}-\alpha_{0,0})^2
        \frac{\big(J_\beta^{11}(\boldsymbol{\theta}_0)\big)^2}{K_\beta^{11}(\boldsymbol{\theta}_0)},
        \qquad
        W_K^\beta(\alpha_1) = 
        K(\widehat{\alpha}_{1,\beta}-\alpha_{1,0})^2
        \frac{\big(J_\beta^{22}(\boldsymbol{\theta}_0)\big)^2}{K_\beta^{22}(\boldsymbol{\theta}_0)},
        \label{eq:wald_H2_sim}
    \end{equation}
    respectively, where $\widehat{\alpha}_{0,\beta}$ and $\widehat{\alpha}_{1,\beta}$ are evaluated by minimizing the DPD loss \eqref{eq:DPD} with respect to $\alpha_0$ and $\alpha_1$, respectively, each obtained with the remaining two parameters fixed at their known values. 
    Each statistic is compared to $\chi^2_{1,\,0.05}$.  Finally, for $H_3$ \eqref{W3}, the test statistic is
    \begin{equation}
        W_K^\beta(\alpha_0, \alpha_1) = 
        K\Big(\widehat{\boldsymbol{\theta}}_{\beta}^{1:2} - \boldsymbol{\theta}_{0}^{1:2}\Big)^\top 
        \left(\Big[\boldsymbol{J}_\beta^{1:2,\,1:2}(\boldsymbol{\theta}_0)\Big]^{-1}
        \boldsymbol{K}_\beta^{1:2,\,1:2}(\boldsymbol{\theta}_0)
        \Big[\boldsymbol{J}_\beta^{1:2,\,1:2}(\boldsymbol{\theta}_0)\Big]^{-1}\right)^{-1}
        \Big(\widehat{\boldsymbol{\theta}}_{\beta}^{1:2} - \boldsymbol{\theta}_{0}^{1:2}\Big),
        \label{eq:wald_H3_sim}
    \end{equation}
    compared to $\chi^2_{2,\,0.05}$. Here $\widehat{\boldsymbol{\theta}}_{\beta}^{1:2} =         (\widehat{\alpha}_{0,\beta}, \widehat{\alpha}_{1,\beta})^\top$        and $\boldsymbol{\theta}_{0}^{1:2} =         (\alpha_{0,0}, \alpha_{1,0})^\top$, where $\widehat{\alpha}_{0,\beta}$ and $\widehat{\alpha}_{1,\beta}$ denote the joint minimizers of the DPD loss \eqref{eq:DPD} with respect to $(\alpha_0,\alpha_1)$, obtained with $\sigma$ plugged in at its known value.

    We also study the composite null hypotheses $H_0^{(4a)}$ and $H_0^{(4b)}$ given in \eqref{W4a}-\eqref{W4b}, in which the parameters not being tested are treated as unknown. The corresponding Wald-type statistics follow from Definition \ref{def:wald_composite} (eq. \eqref{W1.6}), evaluated at the WMDPDE $\widehat{\boldsymbol{\theta}}_\beta$ rather than at the fixed null $\boldsymbol{\theta}_0$. Each of $H_0^{(4a)}$, 
    $H_0^{(4b)}$, and $H_0^{(4c)}$ imposes a single restriction ($r=1$ in the general Definition), so $\boldsymbol h(\boldsymbol\theta)$ reduces to a scalar and $\boldsymbol H(\boldsymbol\theta)$ to a $3\times1$ vector in each case. For $H_0^{(4a)}$, the constraint and its gradient are $\boldsymbol{h}(\boldsymbol{\theta}) =  \alpha_0 - \alpha_{0,0}$ and $\boldsymbol{H}(\boldsymbol{\theta})  = (1,0,0)^\top$, respectively, giving
    \begin{equation}
        W_K^\beta(\alpha_0) = 
        \frac{(\widehat{\alpha}_{0,\beta}-\alpha_{0,0})^2}
        {\widehat{\boldsymbol{\Sigma}}^{11}_\beta
        (\widehat{\boldsymbol{\theta}}_\beta)},
        \label{eq:wald_H4a_sim}
    \end{equation}
    and for $H_0^{(4b)}$, $\boldsymbol{h}(\boldsymbol{\theta}) = \alpha_1 - \alpha_{1,0}$ and $\boldsymbol{H}(\boldsymbol{\theta}) = (0,1,0)^\top$, respectively, giving
    \begin{equation}
        W_K^\beta(\alpha_1) = 
        \frac{(\widehat{\alpha}_{1,\beta}-\alpha_{1,0})^2}
        {\widehat{\boldsymbol{\Sigma}}^{22}_\beta
        (\widehat{\boldsymbol{\theta}}_\beta)},
        \label{eq:wald_H4b_sim}
    \end{equation}
    and, although not included in the simulation study, the analogous constraint and statistic for $H_0^{(4c)}$ \eqref{W4c} on the shape parameter are $\boldsymbol{h}(\boldsymbol{\theta}) = \sigma - \sigma_{0,0}$ and $\boldsymbol{H}(\boldsymbol{\theta}) = (0,0,1)^\top$, respectively, giving
    \begin{equation}
        W_K^\beta(\sigma) = 
        \frac{(\widehat{\sigma}_{\beta}-\sigma_{0})^2}
        {\widehat{\boldsymbol{\Sigma}}^{33}_\beta
        (\widehat{\boldsymbol{\theta}}_\beta)},
        \label{eq:wald_H4c_sim}
    \end{equation}
    each compared to $\chi^2_{1,\,0.05}$, with \eqref{eq:wald_H4c_sim} used later in the real data illustration of Section 
    \ref{sec:real_data}.

    \subsubsection*{(a) Empirical level and power while varying sample size (simple hypotheses)}
    
    Figure \ref{fig:sim_samplesize} shows the empirical level and power against sample size $K$ with no contamination. It can be seen that the empirical level remains close to the nominal level of $0.05$ for all four simple hypotheses and all $\beta$ values. A natural variation due to Monte Carlo decreases as the sample size $K$ increases. The classical Wald test and robust Wald-type tests for any $\beta$ behave similarly in terms of level under clean data.
    \begin{figure}[H]
        \centering
        \includegraphics[width=\textwidth]{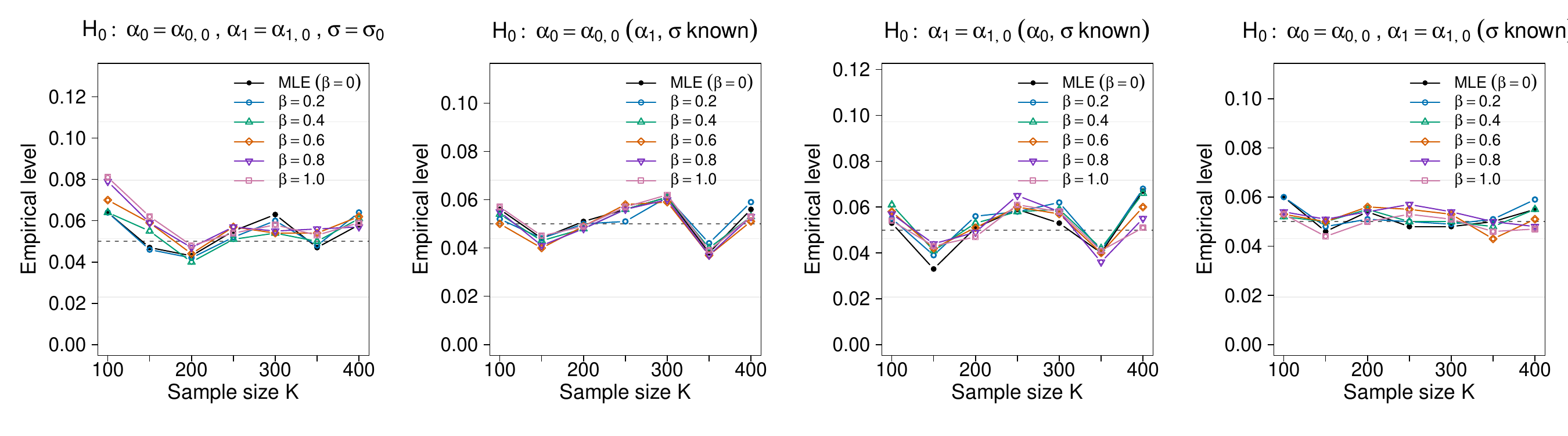}\\[6pt]
        \includegraphics[width=\textwidth]{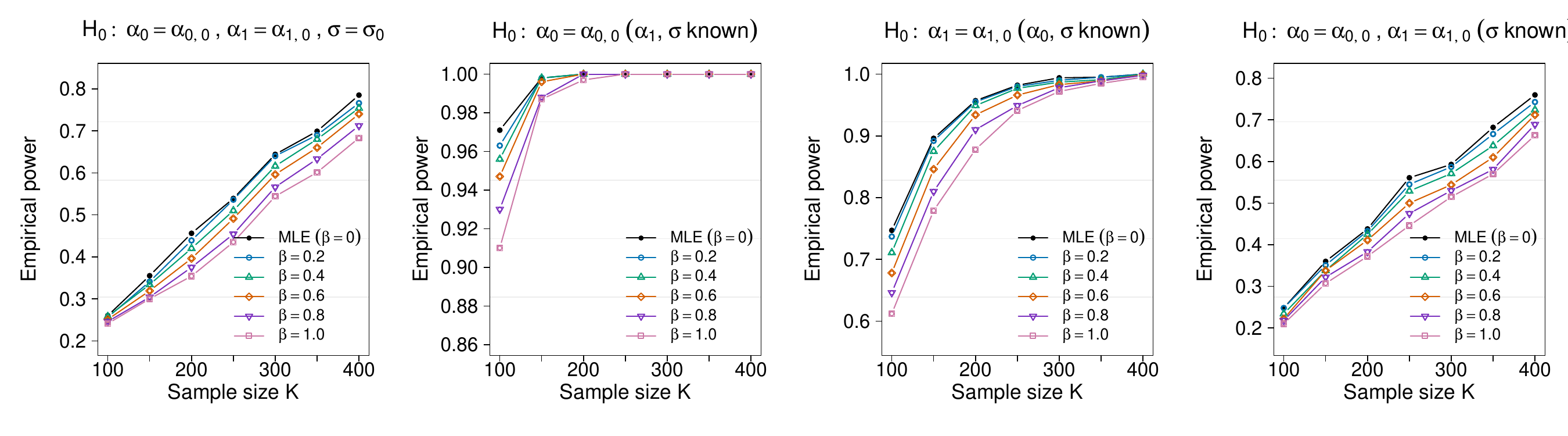}
        \caption{Empirical level (top) and empirical power (bottom) 
        of the Wald-type tests against sample size $K$ 
        ($\varepsilon = 0$). The dashed line represents the level of significance
        $\alpha = 0.05$.}
        \label{fig:sim_samplesize}
    \end{figure}
    
    The empirical power increases with $K$ for all four hypotheses, where the classical Wald test achieves the highest power for any given sample size, consistent with its optimality under uncontaminated observations. Robust Wald-type tests with larger $\beta$ show a modest reduction in power, reflecting the well-known trade-off between efficiency and robustness.

    \subsubsection*{(b) Empirical level and power over increasing contamination (simple hypotheses)}
    Figure \ref{fig:sim_contamination} shows the empirical level and power against contamination proportion $\varepsilon$ at fixed $K=200$. The level results clearly show the robustness advantage of the Wald-type tests. At $\varepsilon = 0$, all tests maintain the nominal level. As $\varepsilon$ increases, the level of the classical Wald test inflates rapidly toward one for all four hypotheses, while robust estimators with larger $\beta$ maintain level control for longer, with $\beta=1.0$ showing the slowest inflation. 
 
	The power results present a more varied picture. For $H_0^{(1)}$, power first decreases. As 
	contamination grows further, Figure \ref{fig:bias_check} shows that the mean estimate under the 
	alternative moves toward the null before crossing it and moving away on the other side; once it crosses, 
	the statistic grows large enough that the test rejects almost regardless of the truth, which makes power 
	look better without the test actually working better. $H_0^{(3)}$ shows the same pattern for the same 
	reason. For $H_0^{(2a)}$ and $H_0^{(2b)}$, only one component drives the statistic, and it matters 
	which side of the null the true alternative sits on. For $H_0^{(2a)}$, the estimate moves toward the 
	null but does not reach it, so power falls throughout. For $H_0^{(2b)}$, the estimate moves away from 
	the null instead, so power rises throughout. If $H_0^{(2b)}$'s true alternative lay on the other side of 
	the null, the power plot would likely show the same robustness pattern as the estimates. Power for the 
	Wald-type tests would be larger than for the classical Wald test. This is not the case here. The estimates in 
	Figure \ref{fig:bias_check} still show the expected robustness. But the power plot does not. This is 
	simply due to which side of the null the true alternative happens to sit on. In every case, the MLE 
	drifts furthest from its starting value, since it lets a few unusual observations dominate the fit more 
	than the DPD-based estimators do.
	\begin{figure}[H]
		\centering
		\includegraphics[width=\textwidth]{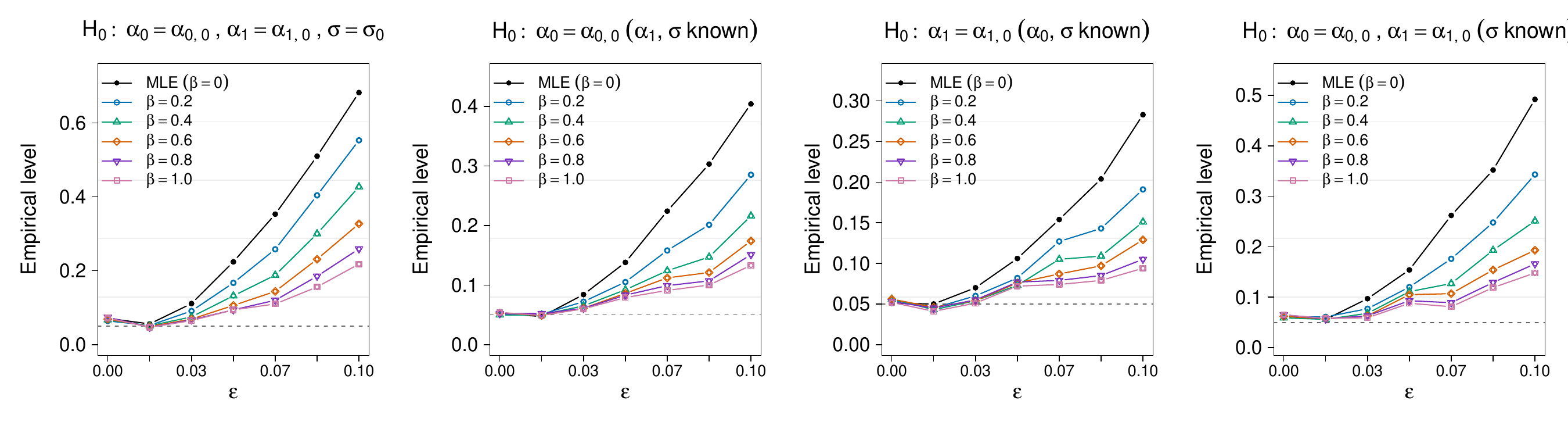}\\[6pt]
		\includegraphics[width=\textwidth]{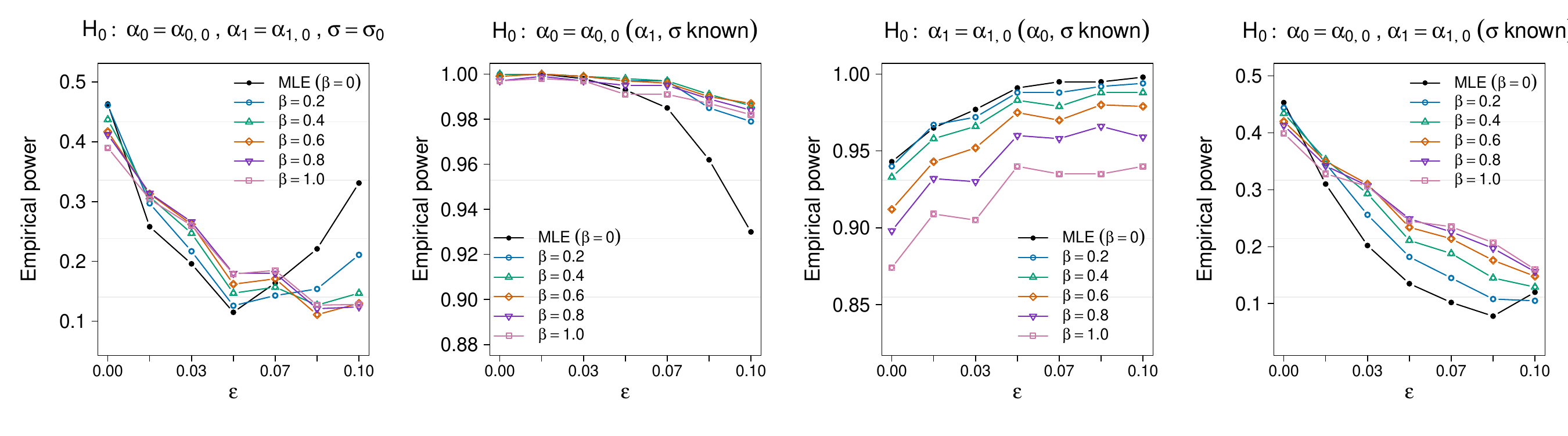}
		\caption{Empirical level (top) and empirical power (bottom) 
			of the Wald-type tests against contamination proportion 
			$\varepsilon$ with $K=200$.}
		\label{fig:sim_contamination}
	\end{figure}
    \begin{figure}[H]
        \centering
        \includegraphics[width=\textwidth]{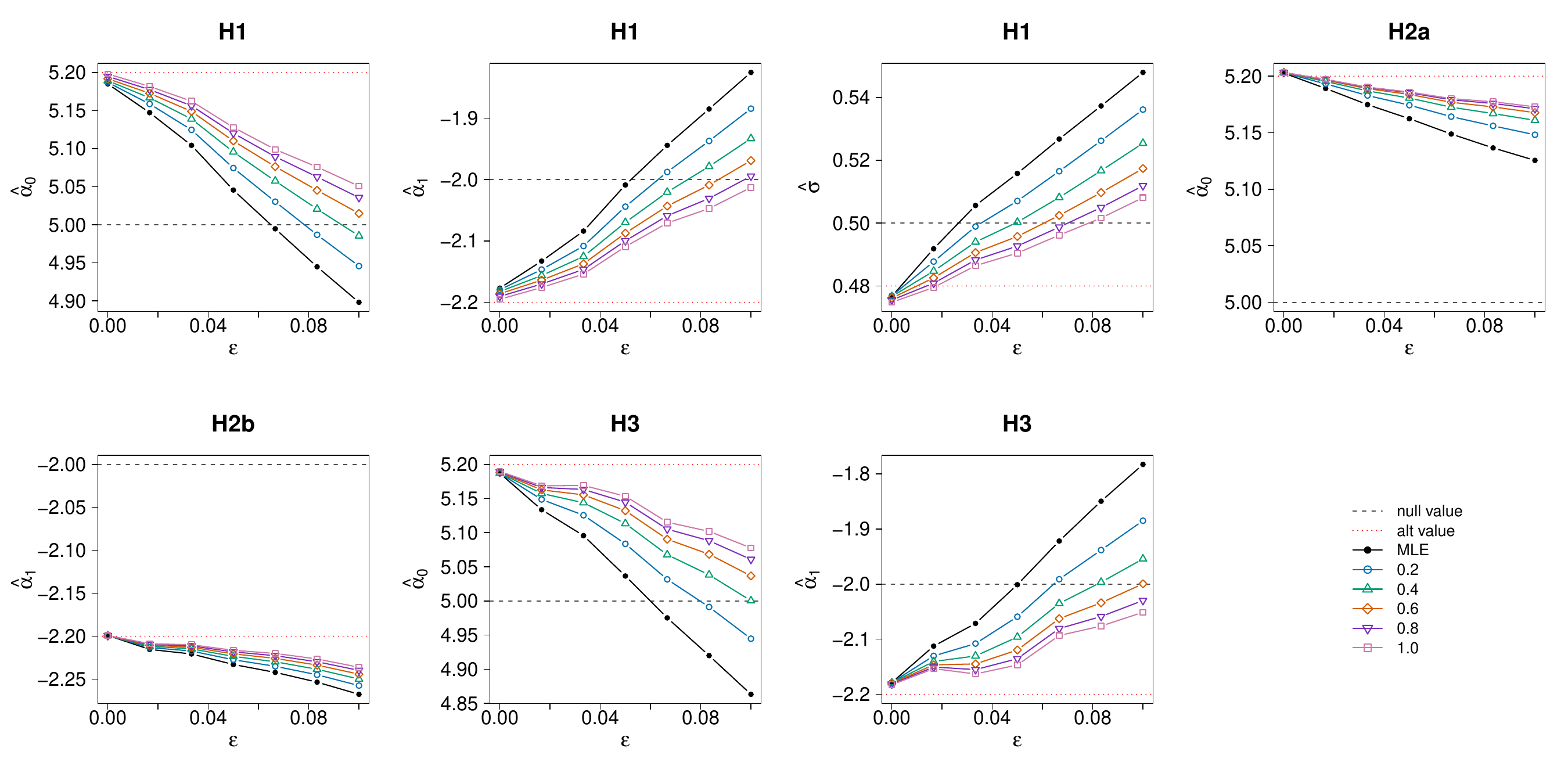}
        \caption{Mean parameter estimates against contamination 
        proportion $\varepsilon$ under the alternative with $K=200$ and $1000$ iterations, for each 
        parameter tested in $H_0^{(1)}$, $H_0^{(2a)}$, $H_0^{(2b)}$, and $H_0^{(3)}$. The dashed line marks the null value; the dotted line marks 
        the true alternative value used to generate the data.}
        \label{fig:bias_check}
    \end{figure}    
    \noindent This pattern shows that the relationship between contamination and power depends on the direction in which contamination biases the estimates relative to the null and the alternative, and it is the level results in Figure \ref{fig:sim_contamination}, unaffected by this issue, that remain the primary evidence of the robustness advantage of the Wald-type tests.

    \subsubsection*{(c) Empirical level and power while varying sample size (composite hypotheses)}
    

    The top panels of Figure \ref{fig:sim_wald_composite} show the empirical level and power for the 
    composite null hypotheses $H_0^{(4a)}$ and $H_0^{(4b)}$, against sample size $K$ with no contamination. 
    As in the simple-hypothesis case, the empirical level stays close to the nominal $0.05$ for all $\beta$, 
    so estimating the extra nuisance parameters does not affect the significance level. Power increases with 
    $K$ for both hypotheses, with the classical Wald test highest and a modest, well-controlled loss as $\beta$ increases, as 
    seen before in Section \ref{subsec:sim_wald_results}(a).

     \begin{figure}[H]
        \centering
        \includegraphics[width=\textwidth]{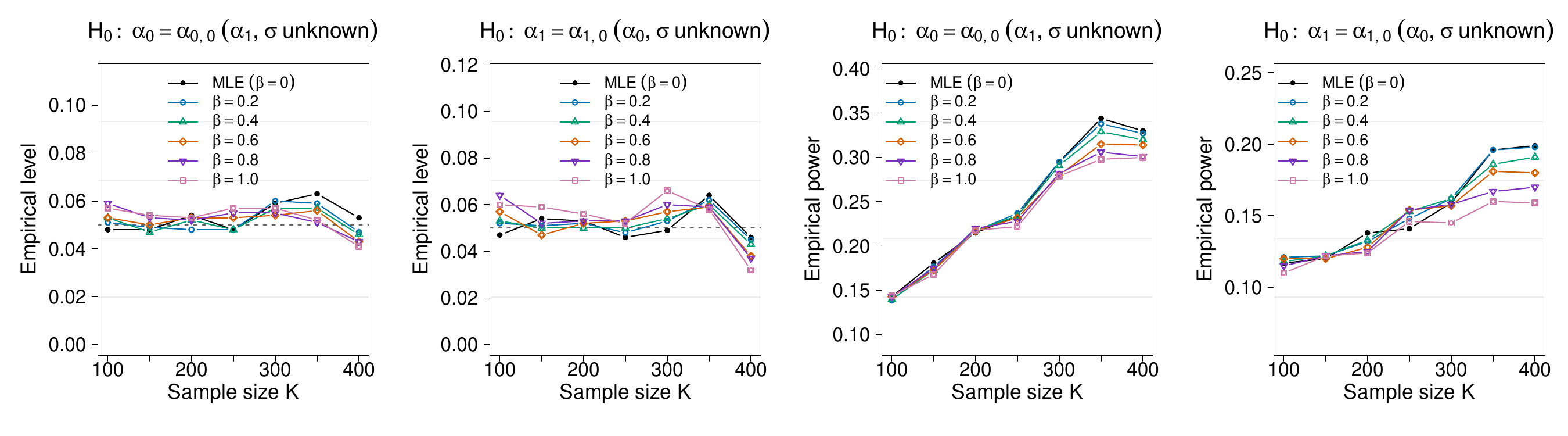}\\[6pt]
        \includegraphics[width=\textwidth]{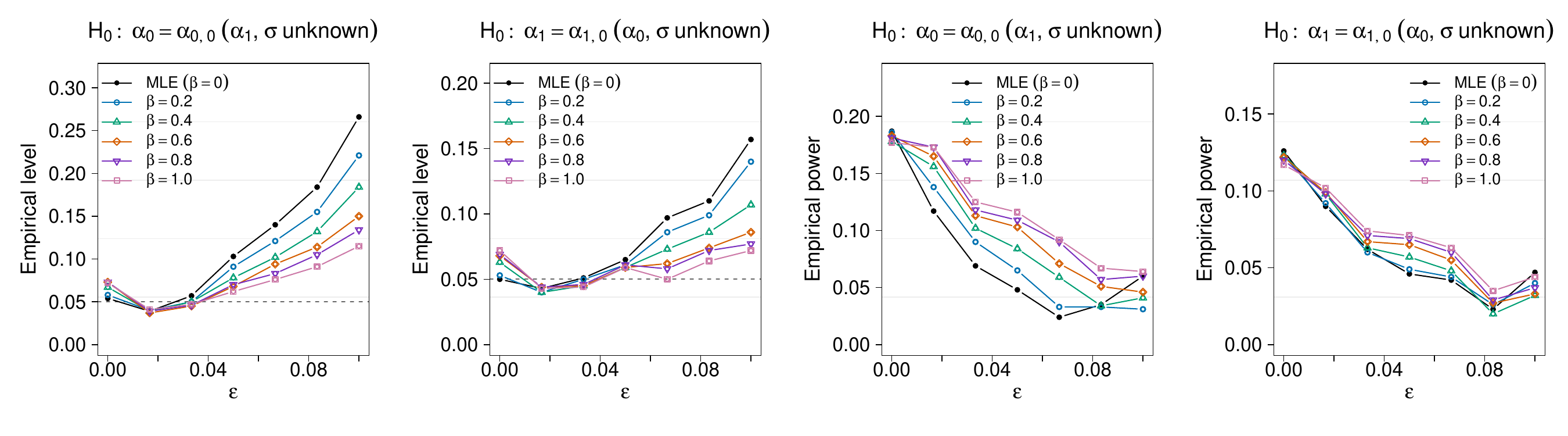}
        \caption{Empirical level and empirical power of the Wald-type tests for the composite null hypotheses $H_0^{(4a)}$ and $H_0^{(4b)}$, 
        against sample size $K$ with no contamination (top) and against contamination proportion $\varepsilon$ when $K=200$ (bottom)}
        \label{fig:sim_wald_composite}
    \end{figure}

`   \subsubsection*{(d) Empirical level and power over increasing contamination (composite hypotheses)}


    The bottom panels of Figure \ref{fig:sim_wald_composite} show the same quantities against $\varepsilon$ 
    at $K=200$. The level results repeat the robustness pattern seen throughout: the classical Wald test inflates fast, while the Wald-type tests with larger $\beta$ stay closer to the nominal level for longer. Power behaves differently for the 
    two hypotheses, for the same reason as in the simple-hypothesis case. For $H_0^{(4a)}$, the estimate 
    crosses the null under contamination, so power first drops; the statistic then grows large enough to 
    reject almost regardless of the truth, so power looks better without the test actually working better. 
    For $H_0^{(4b)}$, the estimate moves away from the null without crossing, so power declines throughout 
    instead, and the robustness of the Wald-type tests is visible directly: power falls more slowly for larger $\beta$ than for the classical Wald test.
    
    It is also worth noting that power for both hypotheses is markedly lower here than for the corresponding 
    simple hypotheses in Section \ref{subsec:sim_wald_results}(a)-(b). This is expected. Unlike $H_0^{(2a)}$ 
    and $H_0^{(2b)}$, where the two nuisance parameters are fixed at their known true values, here they must 
    be estimated jointly with the tested parameter. Estimating extra parameters inflates the variance of the 
    estimator being tested, which reduces power for the same effect size. The level results, unaffected by 
    this loss of power, remain the primary evidence of robustness for the composite hypotheses.

    \subsection{Rao-type tests}
    \label{subsec:sim_rao_results}
    
    In this subsection, we are now going to examine the performance of the Rao-type tests proposed in Section \ref{sec:rao} using the same simulation design, contamination scheme, and alternative parameter vectors described in the Section \ref{subsec:sim_design}. An important computational difference from the Wald-type tests is worth noting here for the simple null hypotheses. The Rao-type statistic does not require computing the WMDPDE from each simulated dataset, since it does not involve an estimate of $\boldsymbol{\theta}$. This makes the Rao-type simulation considerably faster than the Wald-type simulation, where the WMDPDE must be re-estimated by numerical optimisation for every replication.

    For $H_0^{(1)}$ \eqref{W1}, the test statistic is
    \begin{equation}
        R_K^\beta(\boldsymbol{\theta}_0) = K\,
        \boldsymbol{U}_\beta(\boldsymbol{\theta}_0)^\top\,
        \boldsymbol{K}_\beta(\boldsymbol{\theta}_0)^{-1}\,
        \boldsymbol{U}_\beta(\boldsymbol{\theta}_0),
        \label{eq:rao_H1_sim}
    \end{equation}
    compared to $\chi^2_{3,\,0.05}$. For $H_0^{(2a)}$ and $H_0^{(2b)}$ \eqref{W2}, the test statistics reduce to
    \begin{equation}
        R_K^\beta(\alpha_0) = 
        \frac{K\,\left(U_\beta^1(\boldsymbol{\theta}_0)\right)^2}
        {\boldsymbol{K}_\beta^{11}(\boldsymbol{\theta}_0)},
        \qquad
        R_K^\beta(\alpha_1) = 
        \frac{K\,\left(U_\beta^2(\boldsymbol{\theta}_0)\right)^2}
        {\boldsymbol{K}_\beta^{22}(\boldsymbol{\theta}_0)},
        \label{eq:rao_H2_sim}
    \end{equation}
    respectively, each compared to $\chi^2_{1,\,0.05}$. Finally, for $H_0^{(3)}$ \eqref{W3}, the test statistic is
    \begin{equation}
        R_K^\beta(\alpha_0,\alpha_1) = K\,
        \boldsymbol{U}^{1:2}_\beta(\boldsymbol{\theta}_0)^\top\,
        \left(\boldsymbol{K}_\beta^{1:2,1:2}(\boldsymbol{\theta}_0)\right)^{-1}\,
        \boldsymbol{U}^{1:2}_\beta(\boldsymbol{\theta}_0),
        \label{eq:rao_H3_sim}
    \end{equation}
    compared to $\chi^2_{2,\,0.05}$, where $\boldsymbol{U}^{1:2}_\beta(\boldsymbol{\theta}_0) = (U_\beta^1(\boldsymbol{\theta}_0), U_\beta^2(\boldsymbol{\theta}_0))^\top$ denotes the first two components of the estimating function. The empirical level and power are computed exactly as in \eqref{eq:emp_level_power}, replacing $T_K^\beta$ with $R_K^\beta$.
    
    For the composite null hypotheses $H_0^{(4a)}$ and $H_0^{(4b)}$, the Rao-type statistic follows from Definition \ref{def:rao_composite} 
    (eq. \eqref{5}-\eqref{6}), evaluated at the restricted WMDPDE $\widetilde{\boldsymbol{\theta}}_\beta$ obtained under the 
    corresponding null constraint. For $H_0^{(4a)}$, $\boldsymbol{h}(\boldsymbol{\theta}) = \alpha_0 - \alpha_{0,0}$ and $\boldsymbol{H}(\boldsymbol{\theta}) = (1,0,0)^\top$; for $H_0^{(4b)}$, $\boldsymbol{h}(\boldsymbol{\theta}) = \alpha_1 - \alpha_{1,0}$ and $\boldsymbol{H}(\boldsymbol{\theta}) = (0,1,0)^\top$. In this scalar case, $\boldsymbol{Q}^\beta(\widetilde{\boldsymbol{\theta}}_\beta)^\top  \boldsymbol{K}_\beta(\widetilde{\boldsymbol{\theta}}_\beta) \boldsymbol{Q}^\beta(\widetilde{\boldsymbol{\theta}}_\beta)$ reduces to a scalar rather than an $r\times r$ matrix, so no matrix inversion is required at this step; the statistic becomes
    \begin{equation}
        R_K^\beta(\widetilde{\boldsymbol{\theta}}_\beta) = 
        K\,\frac{\left(\boldsymbol{Q}^\beta(\widetilde{\boldsymbol{\theta}}_\beta)^\top
        \boldsymbol{U}_\beta(\widetilde{\boldsymbol{\theta}}_\beta)\right)^2}
        {\boldsymbol{Q}^\beta(\widetilde{\boldsymbol{\theta}}_\beta)^\top
        \boldsymbol{K}_\beta(\widetilde{\boldsymbol{\theta}}_\beta)\,
        \boldsymbol{Q}^\beta(\widetilde{\boldsymbol{\theta}}_\beta)},
        \qquad
        \boldsymbol{Q}^\beta(\widetilde{\boldsymbol{\theta}}_\beta) = 
        \frac{\boldsymbol{J}_\beta(\widetilde{\boldsymbol{\theta}}_\beta)^{-1}
        \boldsymbol{H}}
        {\boldsymbol{H}^\top
        \boldsymbol{J}_\beta(\widetilde{\boldsymbol{\theta}}_\beta)^{-1}
        \boldsymbol{H}},
        \label{eq:rao_H4_sim}
    \end{equation}
    compared to $\chi^2_{1,\,0.05}$. Unlike the simple-hypothesis statistics in \eqref{eq:rao_H2_sim}, this requires computing the restricted WMDPDE $\widetilde{\boldsymbol{\theta}}_\beta$ by numerical optimisation for each replication, so the computational advantage described above, avoiding 
    re-estimation of $\boldsymbol{\theta}$, does not carry over to the composite case. 

    We now present the simulation results for the four simple null hypotheses \eqref{W1}-\eqref{W3} and the composite null hypotheses \eqref{W4a}-\eqref{W4b}, in Subsubsections (a)--(b) and (c)--(d) below, respectively.

    \subsubsection*{(a) Empirical level and power while varying sample size (simple hypotheses)}
    
    Figure \ref{fig:sim_rao_samplesize} shows the empirical level and power against sample size $K$ with no contamination. The empirical 
    level tries to remain close to the nominal level value of $0.05$ for all four hypotheses, with the usual reduction in Monte Carlo noise as 
    $K$ grows. As seen for the Wald-type tests, the classical Rao test also achieves the highest power under clean data, and the Rao-type tests show a small and expected loss of power as $\beta$ increases.
    
    \begin{figure}[H]
        \centering
        \includegraphics[width=\textwidth]{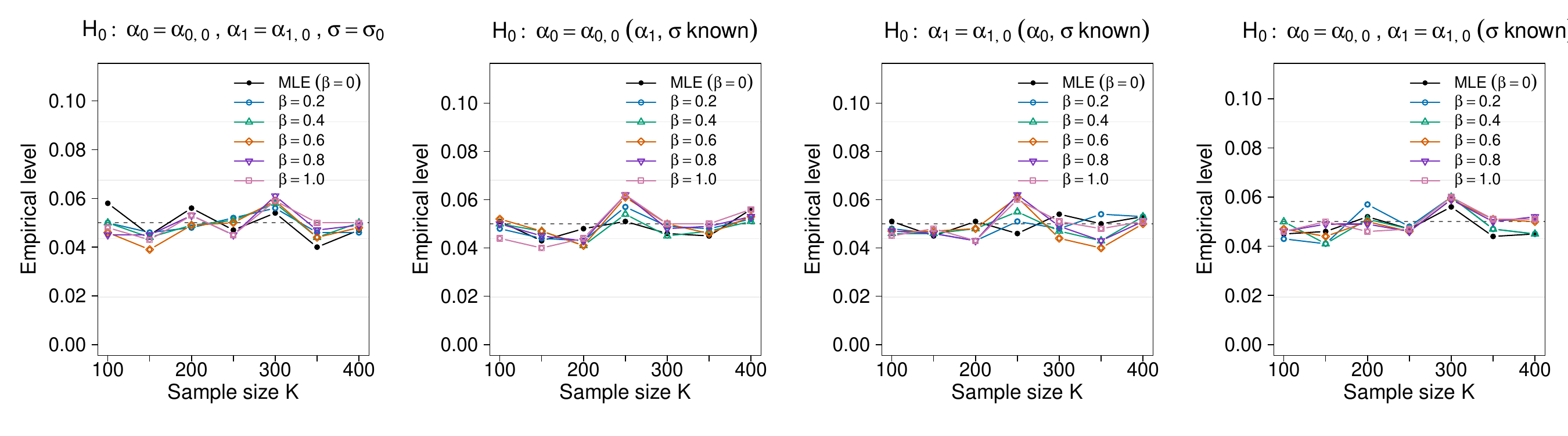}\\[6pt]
        \includegraphics[width=\textwidth]{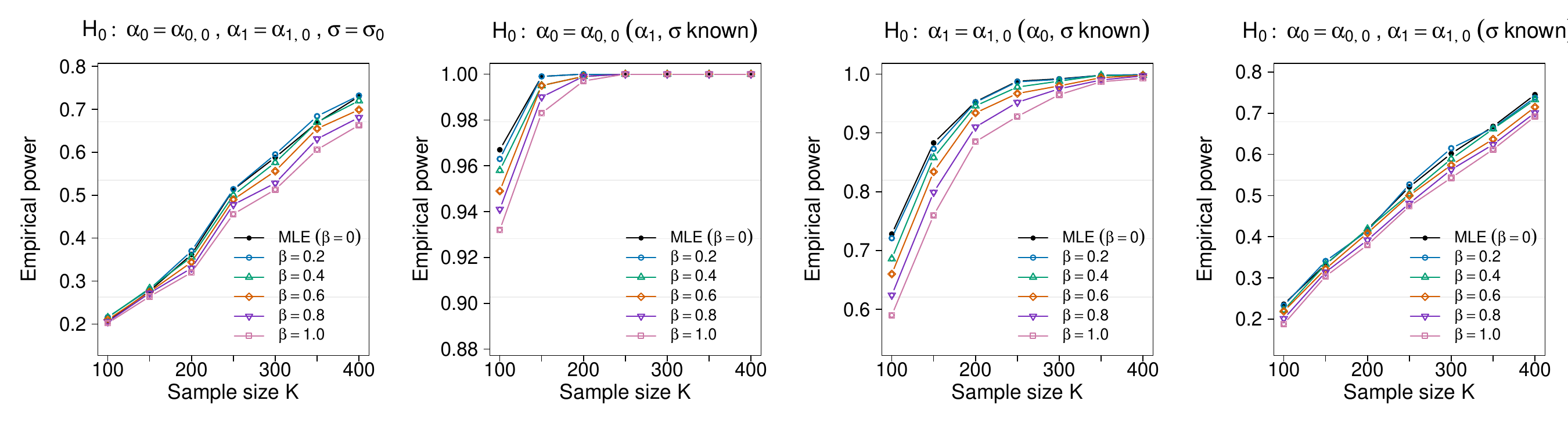}
        \caption{Empirical level (top) and empirical power (bottom) of the Rao-type tests against sample size $K$ under clean data.}
        \label{fig:sim_rao_samplesize}
    \end{figure}
    
    \subsubsection*{(b) Empirical level and power over increasing 
    contamination (simple hypotheses)}
    
    Figure \ref{fig:sim_rao_contamination} shows the empirical 
    level and power against contamination proportion $\varepsilon$ 
    at fixed $K=200$, using the same contamination scheme as for 
    the Wald-type tests.   
    
    The level results confirm the same robustness pattern seen for the Wald-type tests. The classical Rao 
    test inflates quickly as $\varepsilon$ increases, while the robust Rao-type versions of the tests with larger $\beta$ 
    stay closer to the nominal level for longer. 
    
    \begin{figure}[H]
    	\centering
    	\includegraphics[width=\textwidth]{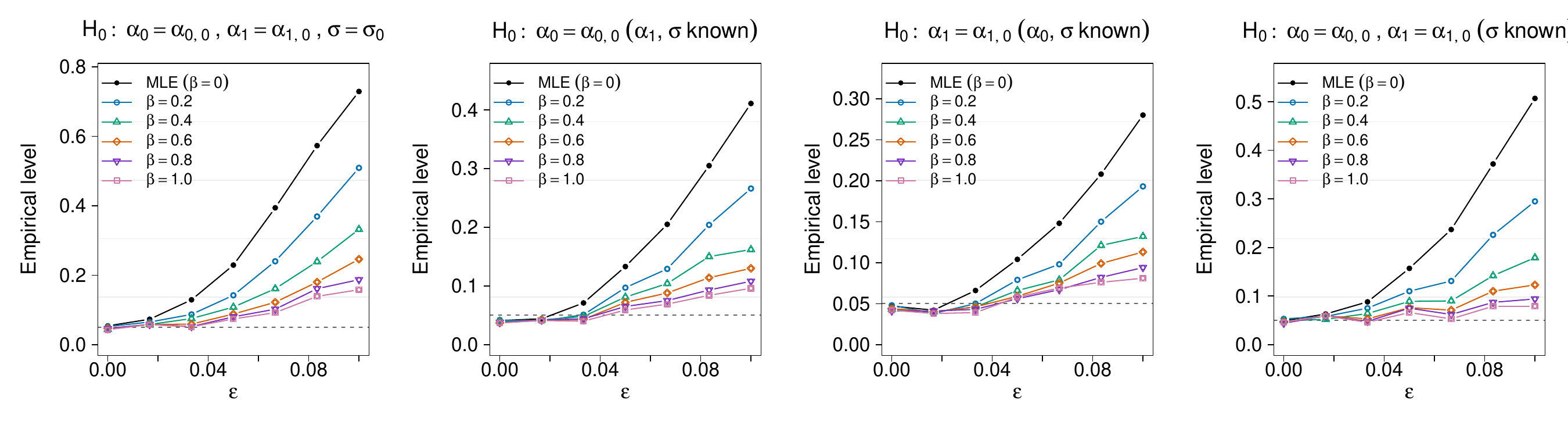}\\[6pt]
    	\includegraphics[width=\textwidth]{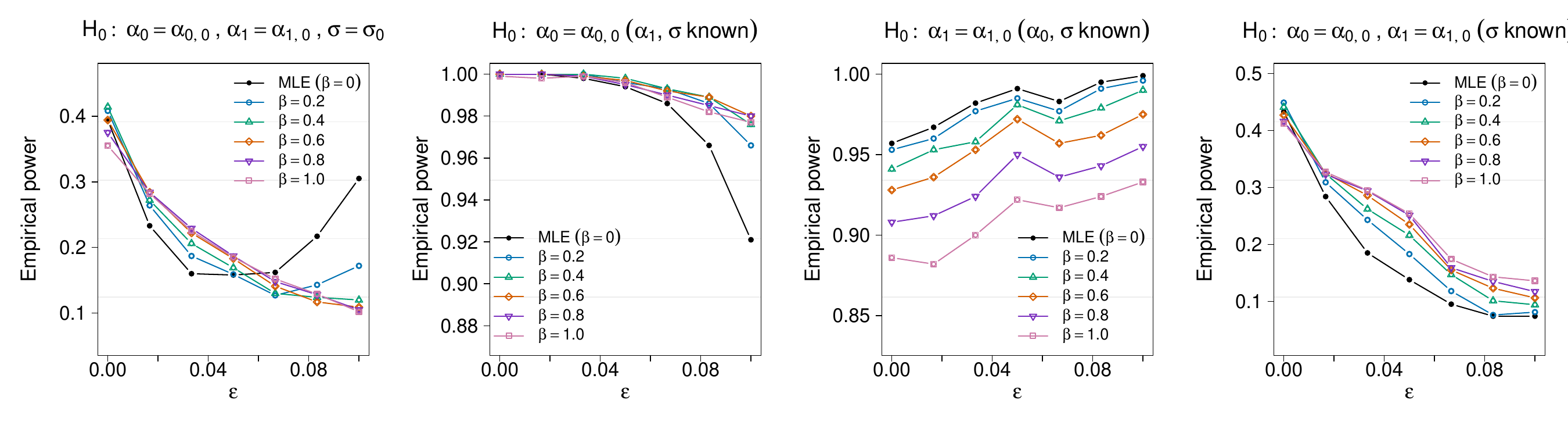}
    	\caption{Empirical level (top) and empirical power (bottom) of the Rao-type tests against contamination proportion $\varepsilon$ with $K=200$ for the four simple null hypotheses.}
    	\label{fig:sim_rao_contamination}
    \end{figure}
    \begin{figure}[H]
    	\centering
    	\includegraphics[width=\textwidth]{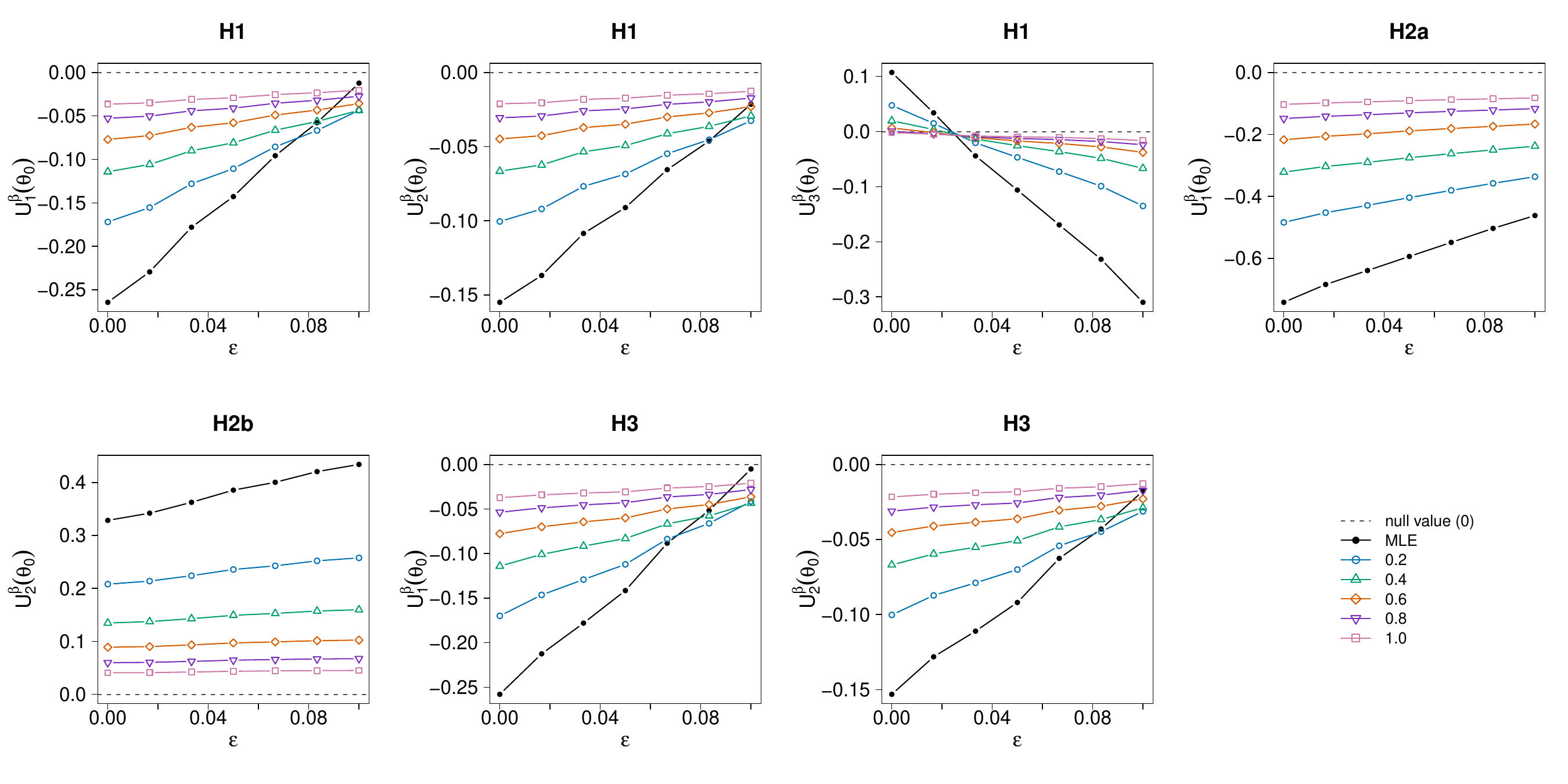}
    	\caption{Mean value of the relevant components of the score function $\boldsymbol{U}_\beta(\boldsymbol{\theta}_0)$ under the alternative-generated data, against contamination proportion $\varepsilon$ with $K=200$, for each hypothesis. The dashed line marks the value $0$, which is the expected value of $\boldsymbol{U}_\beta(\boldsymbol{\theta}_0)$ under $H_0$.}
    	\label{fig:sim_rao_ubias}
    \end{figure}
    
    The Rao-type power results follow a pattern similar to the Wald-type power results seen in Section 
    \ref{subsec:sim_wald_results}(b), here explained through the score function instead. Under $H_0$, 
    $E[\boldsymbol{U}_\beta(\boldsymbol{\theta}_0)] = \mathbf{0}_3$, so zero is the reference value in Figure 
    \ref{fig:sim_rao_ubias}. The figure evaluates $\boldsymbol{U}_\beta(\boldsymbol{\theta}_0)$ on data 
    generated under the alternative used for power, and how far they drift from it is what drives the power 
    of the test. For $H_0^{(1)}$, at least one component crosses zero during the increase in the contamination 
    range, giving the same dip-and-rise seen for the Wald-type test. For $H_0^{(3)}$, the components move 
    toward zero, so power declines for all tests as contamination increases, but it declines more slowly for 
    the DPD-based Rao-type tests than the classical Rao test, showing clear robustness. For $H_0^{(2a)}$ and $H_0^{(2b)}$, only 
    one component enters the statistic, and it behaves as in the Wald case: moving toward zero for 
    $H_0^{(2a)}$, and away from zero throughout for $H_0^{(2b)}$. As before, any late rise in power reflects 
    the statistic growing large enough to reject almost regardless of the truth. The level plots remain the 
    main evidence of robustness.
    \begin{remark}
        Although the Rao-type test statistic in \eqref{eq:rao_H1_sim}-\eqref{eq:rao_H3_sim} 
    does not require parameter estimation of $\boldsymbol{\theta}$, the quantity 
    $\boldsymbol{U}_\beta(\boldsymbol{\theta}_0)$ is not constant across replications. 
    $\boldsymbol{U}_\beta(\boldsymbol{\theta})$ is the DPD estimating equation, whose 
    root defines the WMDPDE $\widehat{\boldsymbol\theta}_\beta$; evaluating it at the 
    fixed $\boldsymbol\theta_0$ therefore indirectly probes the same robustness as DPD 
    estimation, even though $\widehat{\boldsymbol\theta}_\beta$ is never computed here. 
    The score function is
    \begin{equation*}
        \boldsymbol{U}_\beta(\boldsymbol{\theta}_0) = 
        \sum_{i=1}^R \frac{K_i}{K}\,
        \boldsymbol{W}_i(\boldsymbol{\theta}_0)\,
        \boldsymbol{D}_i^{(\beta-1)}(\boldsymbol{\theta}_0)\,
        \big(\mathbf{p}_i(\boldsymbol{\theta}_0) - \widehat{\mathbf{p}}_i\big),
    \end{equation*}
    where $\boldsymbol{W}_i(\boldsymbol{\theta}_0)$, 
    $\boldsymbol{D}_i^{(\beta-1)}(\boldsymbol{\theta}_0)$, and 
    $\mathbf{p}_i(\boldsymbol{\theta}_0)$ depend only on $\boldsymbol{\theta}_0$ 
    and $\beta$; the one data-dependent term is $\widehat{\mathbf{p}}_i = 
    (\widehat p_{i1},\ldots,\widehat p_{iL})^\top$, with $\widehat p_{ij}=n_{ij}/K_i$ 
    the empirical proportion in interval $j$ for group $i$. As $\varepsilon$ increases, 
    $\widehat{\mathbf{p}}_i$ drifts from $\mathbf{p}_i(\boldsymbol{\theta}_0)$ — exactly 
    the pattern in Figure \ref{fig:sim_rao_ubias}, where the MLE's component drifts 
    furthest from zero and robust $\beta$ stays closest, evidencing the robustness 
    of the estimating equation itself.
    \end{remark}
    
    \subsubsection*{(c) Empirical level and power while varying sample size (composite hypotheses)}
    
    The top panels of Figure \ref{fig:sim_rao_composite} show the empirical level and power for $H_0^{(4a)}$ 
    and $H_0^{(4b)}$ against $K$ with no contamination. The level stays close to the nominal level for all 
    $\beta$, and the classical Rao test has the highest power under clean data, with a small, expected loss 
    as $\beta$ increases. These results match the composite Wald-type results in Section 
    \ref{subsec:sim_wald_results}(c) closely, including the lower baseline power compared to the simple 
    hypotheses, for the same reason: the nuisance parameters are estimated jointly here rather than fixed at 
    their known values.

    \subsubsection*{(d) Empirical level and power over increasing contamination (composite hypotheses)}
    
    The bottom panels of Figure \ref{fig:sim_rao_composite} show the same quantities against $\varepsilon$ at 
    $K=200$. The level results confirm the same pattern seen throughout: the classical Rao test inflates 
    quickly, while larger $\beta$ stay closer to the nominal level for longer. Power again splits between the 
    two hypotheses, as in the Wald-type ca
    se in Section \ref{subsec:sim_wald_results}(d). For $H_0^{(4a)}$, 
    power first drops; the statistic then grows large enough to reject almost regardless of the truth. For 
    $H_0^{(4b)}$, power declines throughout instead, with no such turn, and the Rao-type tests decline more slowly than the classical Rao test, showing robustness directly. Unlike the simple-hypothesis case, this cannot be 
    checked directly through $\boldsymbol{U}_\beta(\boldsymbol{\theta}_0)$, since the composite statistic in 
    \eqref{eq:rao_H4_sim} is evaluated at the restricted WMDPDE $\widetilde{\boldsymbol{\theta}}_\beta$ 
    instead. The level results, unaffected by this, remain the primary evidence of robustness.
	\begin{figure}[H]
		\centering
		\includegraphics[width=\textwidth]{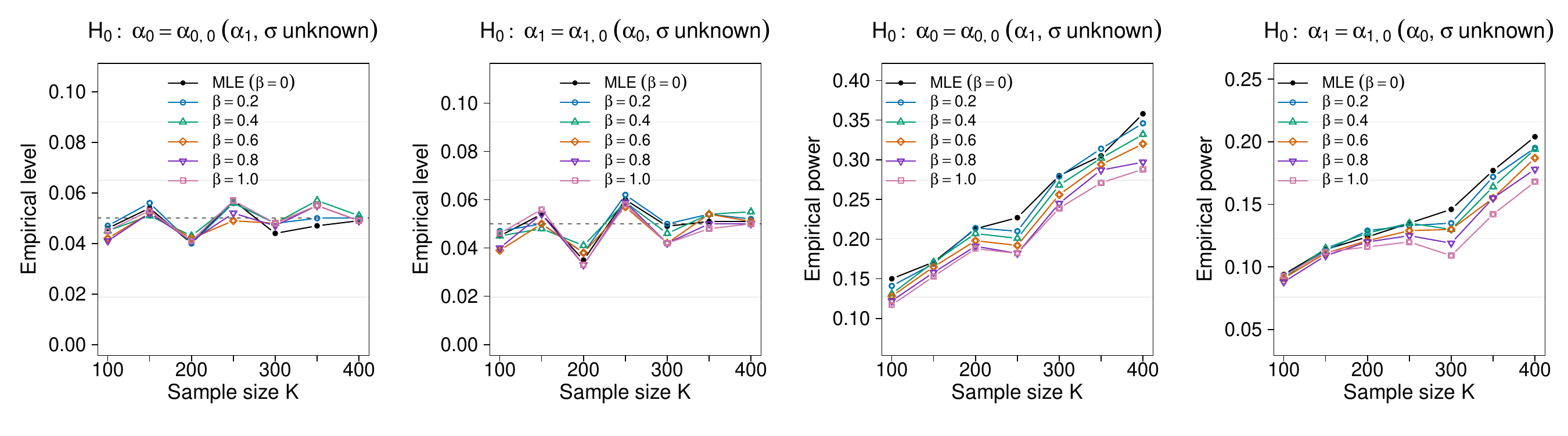}\\[6pt]
		\includegraphics[width=\textwidth]{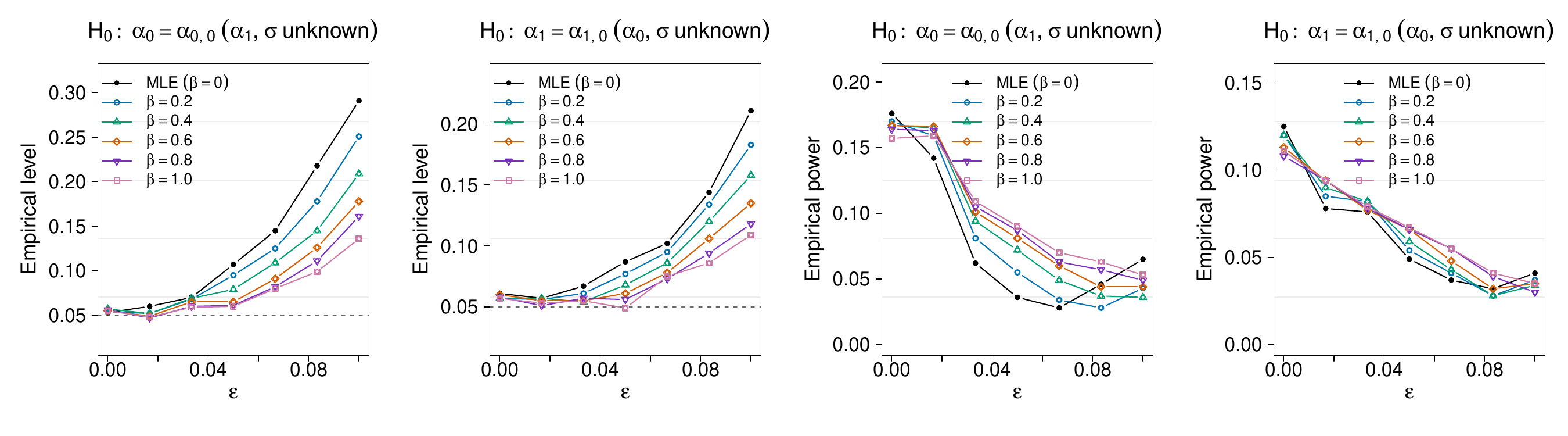}
		\caption{Empirical level and empirical power of the Rao-type tests for the composite null hypotheses $H_0^{(4a)}$ and $H_0^{(4b)}$, 
			against sample size $K$ (top) and against contamination proportion $\varepsilon$ with $K=200$ (bottom).}
		\label{fig:sim_rao_composite}
	\end{figure}
    \section{Application to air-conditioner evaporator data}
    \label{sec:real_data}

    We illustrate the proposed Wald-type and Rao-type tests using 
    the automotive air-conditioner evaporator application described 
    in \cite{kim2021optimal,jaenada2026robust}, where the engineering 
    motivation for the cyclic-stress design and further background on the application are given.
    As in that paper, we use the preliminary experiment on $N=18$ evaporator units to reconstruct 
    a plausible parameter vector before the main CyALT experiment 
    begins. Fitting a log-normal distribution to the observed 
    failure times gives $\sigma = 0.156$. Together with the 
    observed failure proportion $p_h = 0.90$ under the preliminary 
    cyclic-stress condition and the target failure proportion 
    $p_u = 0.001$ at the use condition, this yields the 
    pilot-based parameter vector
    \begin{equation*}
        \boldsymbol{\theta}_0 = 
        (\alpha_{0,0}, \alpha_{1,0}, \sigma_0)^\top = 
        (11.4565,\, -1.0684,\, 0.1560)^\top.
    \end{equation*}
    This vector is treated as one of the null hypotheses in this 
    section. The natural question is whether the main CyALT 
    experiment, once completed, confirms the parameter values on 
    which its own design was based.
    
    The main CyALT experiment, described in detail in 
    \cite{jaenada2026robust}, is designed with $R=2$ cyclic-stress 
    conditions, as shown in Figure \ref{fig:real_design}. Both
     conditions share a common floor level $s_F = 0.30$, above the 
    use-condition ceiling $s_{0C} = 0.27$, ensuring that even the 
    lowest test stress exceeds normal operating conditions. The 
    ceiling levels are set to $s_{1C} = 0.70$ and $s_{2C} = 1.00$, 
    with cyclic fraction $\tau = 0.50$. A total of $K = 200$ units 
    are allocated as $K_1 = 140$ and $K_2 = 60$, following the 
    proportions $\pi_1 = 0.70$ and $\pi_2 = 0.30$. The test runs 
    to a censoring time of $t_c = 65{,}000$ cycles, with $L=6$ 
    pre-fixed inspection times at $25{,}000$, $35{,}000$, 
    $45{,}000$, $50{,}000$, $60{,}000$, and $65{,}000$ cycles.
    
     \begin{figure}[H]
         \centering
         \resizebox{0.9\textwidth}{!}{%

\begin{tikzpicture}[>=stealth, scale=1.0]
	
	%
	
	\fill[gray!13] (0,0.0) rectangle (12,1.62);
	
	\foreach \y in {0.0, 1.62, 1.80, 4.20, 6.0}
	\draw[gray!35, thin, dashed] (0,\y) -- (12,\y);
	
	\draw[gray!55, thick]
	(0,0)--(0,1.62)--(1,1.62)--(1,0)--(2,0)
	--(2,1.62)--(3,1.62)--(3,0)--(4,0)
	--(4,1.62)--(5,1.62)--(5,0)--(6,0)
	--(6,1.62)--(7,1.62)--(7,0)--(8,0)
	--(8,1.62)--(9,1.62)--(9,0)--(10,0)
	--(10,1.62)--(11,1.62)--(11,0)--(12,0);
	
	\draw[blue!90!black, very thick]
	(-0.06,1.80)--(-0.06,4.20)--(0.94,4.20)--(0.94,1.80)--(1.94,1.80)
	--(1.94,4.20)--(2.94,4.20)--(2.94,1.80)--(3.94,1.80)
	--(3.94,4.20)--(4.94,4.20)--(4.94,1.80)--(5.94,1.80)
	--(5.94,4.20)--(6.94,4.20)--(6.94,1.80)--(7.94,1.80)
	--(7.94,4.20)--(8.94,4.20)--(8.94,1.80)--(9.94,1.80)
	--(9.94,4.20)--(10.94,4.20)--(10.94,1.80)--(11.94,1.80);
	
	\draw[red!90!black, very thick]
	(0,1.80)--(0,6.0)--(1,6.0)--(1,1.80)--(2,1.80)
	--(2,6.0)--(3,6.0)--(3,1.80)--(4,1.80)
	--(4,6.0)--(5,6.0)--(5,1.80)--(6,1.80)
	--(6,6.0)--(7,6.0)--(7,1.80)--(8,1.80)
	--(8,6.0)--(9,6.0)--(9,1.80)--(10,1.80)
	--(10,6.0)--(11,6.0)--(11,1.80)--(12,1.80);
	
	\node[font=\large, gray!65] at (12.65, 0.81) {$\cdots$};
	\node[font=\large, gray!65] at (12.65, 3.00) {$\cdots$};
	\node[font=\large, gray!65] at (12.65, 4.50) {$\cdots$};
	
	\draw[thick, ->] (-0.3,0) -- (13.4,0)
	node[right, font=\small] {Cycle};
	\draw[thick, ->] (-0.3,0) -- (-0.3,7.0)
	node[above, font=\small] {Stress $s$};
	
	\draw (-0.45, 0.0)  -- (-0.3, 0.0);
	\draw (-0.45, 1.62) -- (-0.3, 1.62);
	\draw (-0.45, 1.80) -- (-0.3, 1.80);
	\draw (-0.45, 4.20) -- (-0.3, 4.20);
	\draw (-0.45, 6.0)  -- (-0.3, 6.0);
	
	\node[left, font=\small] at (-0.48, 0.0)  {$0.00$};
	
	\draw[gray!50, thin] (-0.45, 1.62) -- (-1.1, 1.30);
	\node[left, font=\scriptsize] at (-1.1, 1.30) {$0.27$};
	
	\draw[gray!50, thin] (-0.45, 1.80) -- (-1.1, 2.15);
	\node[left, font=\scriptsize] at (-1.1, 2.15) {$0.30$};
	
	\node[left, font=\small] at (-0.48, 4.20) {$0.70$};
	\node[left, font=\small] at (-0.48, 6.0)  {$1.00$};
	
	\draw[<->, gray!60, thin] (0,-0.55)--(1,-0.55)
	node[midway, below, font=\footnotesize] {$\tau T$};
	\draw[<->, gray!60, thin] (1,-0.55)--(2,-0.55)
	node[midway, below, font=\footnotesize] {$(1{-}\tau)T$};
	\draw[gray!40, thin, dotted] (1,-0.80)--(1,6.5);
	\draw[decorate,
	decoration={brace, mirror, amplitude=5pt, raise=3pt}]
	(0,-0.85)--(2,-0.85)
	node[midway, below=7pt, font=\footnotesize] {One cycle of $20$ seconds};
	
	\foreach \xx in {4.615, 6.462, 8.308, 9.231, 11.077, 12.0}
	\draw[gray!45, thin, dashed] (\xx,0)--(\xx,6.5);
	
	\draw[thick, ->] (0,-2.0)--(13.5,-2.0)
	node[right, font=\small] {$t$ (${\times}10^3$ cycles)};
	\node[below, font=\footnotesize] at (0,-2.0) {$0$};
	
	\foreach \xx/\lbl/\t in {
		4.615/{$IT_1$}/{25},
		6.462/{$IT_2$}/{35},
		8.308/{$IT_3$}/{45},
		9.231/{$IT_4$}/{50},
		11.077/{$IT_5$}/{60},
		12.0/{$IT_6$}/{65}}
	{
		\draw[thick] (\xx,-1.88)--(\xx,-2.12);
		\node[above, font=\footnotesize] at (\xx,-1.88) {\lbl};
		\node[below, font=\footnotesize] at (\xx,-2.12) {\t};
	}
	
	\node[right, font=\footnotesize, gray!70]       at (12.0, 0.0)
	{$s_{0F}$};
	\node[right, font=\scriptsize,   gray!70, yshift= 4pt] at (12.0, 1.62)
	{$s_{0C}$};
	\node[right, font=\scriptsize,   black!60, yshift=-4pt] at (12.0, 1.80)
	{$s_F$};
	\node[right, font=\footnotesize, blue!90!black] at (12.0, 4.20)
	{$s_{1C} \rightarrow \text{Cond.~1:}\;K_1{=}140$};
	\node[right, font=\footnotesize, red!90!black]  at (12.0, 6.0)
	{$s_{2C} \rightarrow \text{Cond.~2:}\;K_2{=}60$};
	
\end{tikzpicture}}
         \caption{Cyclic-stress ALT design for the real data 
         analysis. The grey band shows the use-condition stress 
         cycling between $s_{0F}=0.00$ and $s_{0C}=0.27$. Vertical 
         dashed lines mark the $L=6$ inspection times.}
         \label{fig:real_design}
     \end{figure}
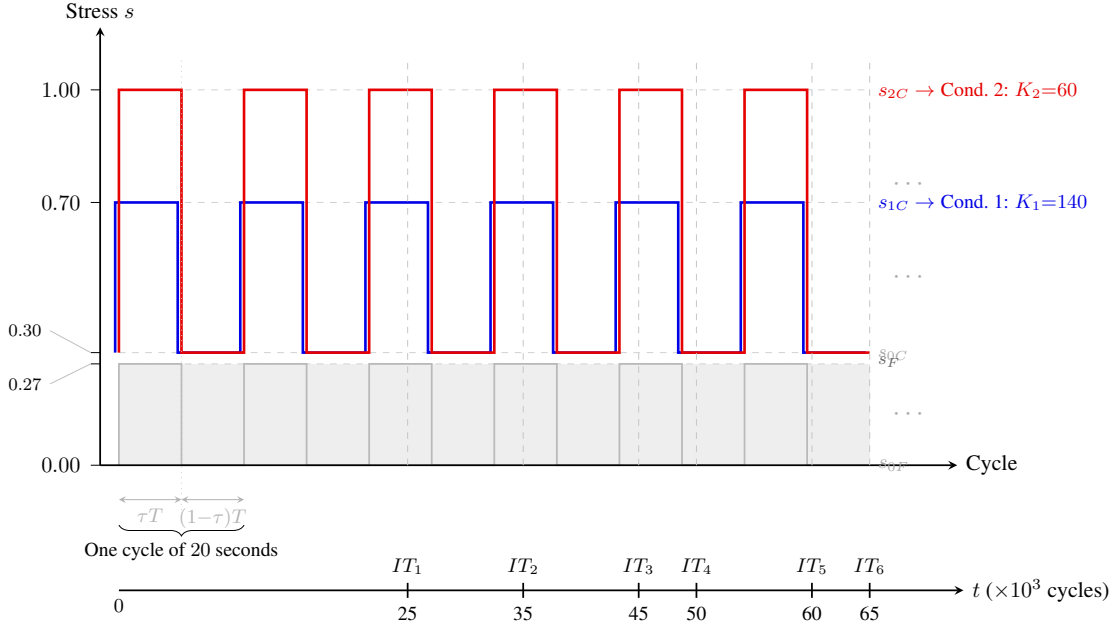
    
    The interval-censored failure counts obtained under this design 
    are reported in \cite{jaenada2026robust} and reproduced in 
    Table \ref{tab:real_data_counts}. Group~1, under ceiling stress 
    $s_{1C}=0.70$, concentrates most failures in the interval 
    $(50{,}000,\,60{,}000]$ with $18$ survivors beyond $t_c$. 
    Group~2, under the higher ceiling stress $s_{2C}=1.00$, fails 
    almost entirely before $45{,}000$ cycles with no survivors. 
    
    \begin{table}[htbp]
        \centering
        \small
        \caption{Interval-censored failure counts from the CyALT 
        experiment with $K_1=140$, $K_2=60$, and $t_c=65{,}000$ 
        cycles.}
        \label{tab:real_data_counts}
        \begin{tabular}{l|cc}
            \toprule
            Inspection interval & Group 1 & Group 2 \\
            (cycles) & ($s_{1C}=0.70$) & ($s_{2C}=1.00$) \\
            \midrule
            $(0,\,25{,}000]$        & 0  & 0  \\
            $(25{,}000,\,35{,}000]$ & 1  & 3  \\
            $(35{,}000,\,45{,}000]$ & 13 & 30 \\
            $(45{,}000,\,50{,}000]$ & 25 & 12 \\
            $(50{,}000,\,60{,}000]$ & 62 & 14 \\
            $(60{,}000,\,65{,}000]$ & 21 & 1  \\
            Survivors ($>65{,}000$) & 18 & 0  \\
            \midrule
            Total & 140 & 60 \\
            \bottomrule
        \end{tabular}
    \end{table}
    
    We apply the Wald-type and Rao-type tests developed in 
    Sections \ref{sec:wald} and \ref{sec:rao} to this dataset, 
    using the pilot-based vector $\boldsymbol{\theta}_0$ obtained 
    above. Three hypotheses are of particular interest for this 
    dataset, each addressing a different practical question.
    
    The first is the simple null hypothesis
    \begin{equation}
        H_0^{(A)}: (\alpha_0,\alpha_1,\sigma) = 
        (\alpha_{0,0},\alpha_{1,0},\sigma_0)
        \quad \text{vs} \quad
        H_1^{(A)}: (\alpha_0,\alpha_1,\sigma) \neq 
        (\alpha_{0,0},\alpha_{1,0},\sigma_0),
        \label{eq:real_HA}
    \end{equation}
    corresponding to \eqref{W1}, computed as in 
\eqref{eq:wald_H1_sim} and \eqref{eq:rao_H1_sim}, which asks whether the completed 
    CyALT experiment is consistent with the parameter values used 
    to design it. Since $\boldsymbol{\theta}_0$ was itself derived 
    from a preliminary experiment before the main test began, this 
    hypothesis validates the pilot-based design.
    
    The second is the composite null hypothesis
    \begin{equation}
        H_0^{(B)}: \alpha_1 = 0        \quad \text{vs} \quad
        H_1^{(B)}: \alpha_1 \neq 0,
        \qquad \text{with } \alpha_0, \sigma \text{ unknown},
        \label{eq:real_HB}
    \end{equation}
    corresponding to \eqref{W4b}, computed as in 
    \eqref{eq:wald_H4b_sim} and \eqref{eq:rao_H4_sim} (with $\boldsymbol{H}(\boldsymbol{\theta})=(0,1,0)^\top$), which addresses a question of direct engineering interest: does the applied cyclic stress have a meaningful effect on evaporator lifetime? A value of $\alpha_1=0$ corresponds to no stress effect at all, i.e., the log-lifetime is insensitive to the applied stress. Unlike $H_0^{(A)}$, this hypothesis leaves $\alpha_0$ and $\sigma$ unknown and estimated, so it also illustrates the composite testing framework of Sections \ref{sec:wald} and \ref{sec:rao} directly on real data.
    
    The third is the composite null hypothesis on the shape 
    parameter,
    \begin{equation}
        H_0^{(C)}: \sigma = \sigma_0
        \quad \text{vs} \quad
        H_1^{(C)}: \sigma \neq \sigma_0,
        \qquad \text{with } \alpha_0, \alpha_1 \text{ unknown},
        \label{eq:real_HC}
    \end{equation}
    corresponding to \eqref{W4c}, computed as in 
    \eqref{eq:wald_H4c_sim} and \eqref{eq:rao_H4_sim} (with $\boldsymbol{H}(\boldsymbol{\theta})=(0,0,1)^\top$), which examines whether the dispersion of the log-normal lifetime distribution, estimated 
    from only $N=18$ preliminary units, remains consistent with 
    the full $K=200$-unit CyALT dataset.

    \begin{table}[htbp]
    \centering
    \small
    \setlength{\tabcolsep}{4pt}
    \caption{Wald-type and Rao-type test statistics and 
    $p$-values, reported as Wald (Rao), for the CyALT evaporator 
    dataset, testing $H_0^{(A)}$ (simple, pilot design 
    consistency), $H_0^{(B)}$ (composite, stress-effect 
    significance), and $H_0^{(C)}$ (composite, shape-parameter 
    consistency).}
    \label{tab:real_tests}
    \begin{tabular}{c|cc cc cc}
        \toprule
        & \multicolumn{2}{c}{$H_0^{(A)}$} 
        & \multicolumn{2}{c}{$H_0^{(B)}$}
        & \multicolumn{2}{c}{$H_0^{(C)}$} \\
        \cmidrule(lr){2-3}\cmidrule(lr){4-5}\cmidrule(lr){6-7}
        $\beta$ & Statistic & $p$ & Statistic & $p$ 
        & Statistic & $p$ \\
        \midrule
        $0$   & 0.69 (0.68) & 0.88 (0.88) & 95.75 (50.14) & 0.000 (0.000) & 0.02 (0.02) & 0.89 (0.89) \\
        $0.2$ & 0.90 (0.88) & 0.83 (0.83) & 94.58 (51.87) & 0.000 (0.000) & 0.06 (0.06) & 0.81 (0.81) \\
        $0.4$ & 0.93 (0.93) & 0.82 (0.82) & 89.01 (52.07) & 0.000 (0.000) & 0.06 (0.06) & 0.80 (0.80) \\
        $0.6$ & 0.88 (0.89) & 0.83 (0.83) & 82.66 (51.48) & 0.000 (0.000) & 0.05 (0.05) & 0.83 (0.82) \\
        $0.8$ & 0.80 (0.82) & 0.85 (0.85) & 76.86 (50.42) & 0.000 (0.000) & 0.03 (0.04) & 0.85 (0.85) \\
        $1.0$ & 0.73 (0.75) & 0.87 (0.86) & 71.76 (48.98) & 0.000 (0.000) & 0.02 (0.02) & 0.88 (0.88) \\
        \bottomrule
    \end{tabular}
    \end{table}
    
    Table \ref{tab:real_tests} reports the resulting Wald-type and 
    Rao-type test statistics and $p$-values for each hypothesis and 
    each value of the tuning parameter $\beta$. Both $H_0^{(A)}$ 
    and $H_0^{(C)}$ are not rejected at any $\beta$, with 
    $p$-values well above conventional significance levels 
    throughout. This indicates no evidence against the pilot-based 
    parameter values on which the main CyALT experiment was 
    designed, and confirms that the assumed dispersion carries over 
    consistently from the preliminary experiment to the full 
    dataset. In sharp contrast, $H_0^{(B)}$ is rejected 
    overwhelmingly by both test families at every value of $\beta$, 
    confirming, as expected from the accelerated life-testing 
    premise, that the applied cyclic stress has a highly 
    significant effect on evaporator lifetime. Together, these 
    results demonstrate that the proposed tests are not merely 
    conservative: they correctly fail to reject a well-calibrated 
    null while decisively rejecting a null that contradicts the 
    basic physics of the experiment, across the full range of the 
    robustness parameter $\beta$.
    

    \section{Concluding remarks}
    \label{sec:conclusion}
    This paper developed robust Wald-type and Rao-type tests for the CyALT model with log-normal lifetimes under interval monitoring, based on the WMDPDE of \cite{jaenada2026robust}. Both test families were derived for simple and composite null hypotheses on the model parameters, with their asymptotic chi-square distribution and consistency established under the non-identically distributed sampling scheme induced by the cyclic-stress design. This extends the classical Wald and Rao testing framework, previously unavailable for CyALT data under interval censoring, to a robust setting governed by a single tuning parameter $\beta$.
    
    The simulation study confirms the expected efficiency-robustness trade-off: under clean data, the classical tests ($\beta=0$) attain the highest power, with only a modest reduction as $\beta$ increases; under contamination, this ordering reverses for the empirical level, with the classical tests losing control of the significance level rapidly while the robust tests remain close to nominal for a substantially wider range of contamination. A rise in power at high contamination is not, by itself, evidence of good performance, since it can reflect bias-driven rejection rather than genuine sensitivity to the alternative; the empirical level therefore provides the more reliable diagnostic of robustness throughout this study. The air-conditioner evaporator application confirms these findings on real data: the proposed tests correctly fail to reject a well-calibrated null while decisively rejecting a null that contradicts the physical premise of the experiment, with close agreement between the Wald-type and Rao-type statistics near the tested value.
    
    The proposed tests provide practitioners with a robust alternative to classical likelihood-based inference for CyALT experiments, particularly relevant when interval-monitored failure data are prone to recording errors or unit heterogeneity. Future work includes extending the framework to other lifetime distributions and stress-life relationships, and developing a data-driven procedure for selecting $\beta$.
    
    \section*{Code and data availability}
    The R code implementing the WMDPDE, the Wald-type and Rao-type test statistics, and the Monte Carlo 
    simulation study, together with the air-conditioner evaporator dataset used in Section~\ref{sec:real_data}, 
    is available at \url{https://github.com/kiranprajapat92/Robust-Inference-WMDPDE-CyALT}.
    
    \bibliographystyle{apalike}
    \bibliography{references_tests}{}
	
	\appendix
    \section{Proofs}
    \label{app:proofs}

    \subsection*{Proof of Theorem \ref{thm:wald_consistency}}

    We define,%
    \begin{equation*}
    L\left( \boldsymbol{\theta }\right) =\left( \boldsymbol{\theta }-\boldsymbol{%
    \theta }_{0}\right) ^{T}\left( \boldsymbol{J}_{\beta }\left( \boldsymbol{%
    \theta }_{0}\right) ^{-1}\boldsymbol{K}_{\beta }\left( \boldsymbol{\theta }%
    _{0}\right) \boldsymbol{J}_{\beta }\left( \boldsymbol{\theta }_{0}\right)
    ^{-1}\right) ^{-1}\left( \boldsymbol{\theta }-\boldsymbol{\theta }%
    _{0}\right) .
    \end{equation*}%
    Taking derivatives, we have
    \begin{equation*}
    \frac{\partial L\left( \boldsymbol{\theta }\right) }{\partial \boldsymbol{%
    \theta }}=2\left( \boldsymbol{J}_{\beta }\left( \boldsymbol{\theta }%
    _{0}\right) ^{-1}\boldsymbol{K}_{\beta }\left( \boldsymbol{\theta }%
    _{0}\right) \boldsymbol{J}_{\beta }\left( \boldsymbol{\theta }_{0}\right)
    ^{-1}\right) ^{-1}\left( \boldsymbol{\theta }-\boldsymbol{\theta }%
    _{0}\right) .
    \end{equation*}%
    Now, a first order Taylor expansion of $L\left( \boldsymbol{\theta }\right) $ in $%
    \boldsymbol{\theta =\theta }^{\ast }$ and evaluated in $\widehat{\boldsymbol{%
    \theta }}_{\beta }$ is given by%
    \begin{equation*}
    L\left( \widehat{\boldsymbol{\theta }}_{\beta }\right) =L\left( \boldsymbol{%
    \theta }^{\ast }\right) +\left( \frac{\partial L\left( \boldsymbol{\theta }%
    \right) }{\partial \boldsymbol{\theta }}\right) _{\boldsymbol{\theta =\theta 
    }^{\ast }}\left( \widehat{\boldsymbol{\theta }}_{\beta }-\boldsymbol{\theta }%
    ^{\ast }\right) +o_{p}\left( \left\Vert \widehat{\boldsymbol{\theta }}%
    _{\beta }-\boldsymbol{\theta }_{0}\right\Vert \right) .
    \end{equation*}%
    Therefore, 
    \begin{eqnarray*}
    L\left( \widehat{\boldsymbol{\theta }}_{\beta }\right) -L\left( \boldsymbol{%
    \theta }^{\ast }\right) &=&2\left( \boldsymbol{\theta }^{\ast }-\boldsymbol{%
    \theta }_{0}\right) ^{T}\left( \boldsymbol{J}_{\beta }\left( \boldsymbol{%
    \theta }_{0}\right) ^{-1}\boldsymbol{K}_{\beta }\left( \boldsymbol{\theta }%
    _{0}\right) \boldsymbol{J}_{\beta }\left( \boldsymbol{\theta }_{0}\right)
    ^{-1}\right) ^{-1}\left( \widehat{\boldsymbol{\theta }}_{\beta }-\boldsymbol{%
    \theta }^{\ast }\right) +o_{p}\left( \left\Vert \widehat{\boldsymbol{\theta }%
    }_{\beta }-\boldsymbol{\theta }_{0}\right\Vert \right) = \\
    &=&2\left( \boldsymbol{\theta }^{\ast }-\boldsymbol{\theta }_{0}\right) ^{T}%
    \boldsymbol{\Sigma }\left( \boldsymbol{\theta }_{0}\right) ^{-1}\left( 
    \widehat{\boldsymbol{\theta }}_{\beta }-\boldsymbol{\theta }^{\ast }\right)
    +o_{p}\left( \left\Vert \boldsymbol{\theta }^{\ast }-\boldsymbol{\theta }%
    _{0}\right\Vert \right) .
    \end{eqnarray*}%
    with 
    \begin{equation*}
    \boldsymbol{\Sigma }\left( \boldsymbol{\theta }\right) =\boldsymbol{J}%
    _{\beta }\left( \boldsymbol{\theta }\right) ^{-1}\boldsymbol{K}_{\beta
    }\left( \boldsymbol{\theta }\right) \boldsymbol{J}_{\beta }\left( 
    \boldsymbol{\theta }\right) ^{-1}.
    \end{equation*}%
    Now,  by Theorem \ref{thm:asymp}, the WMDPDE satisfies
        \begin{equation}
        \sqrt{K}\,\big(
        \widehat{\boldsymbol{\theta}}_\beta - 
        \boldsymbol{\theta}^\ast
        \big)
        \;\xrightarrow[K \to \infty]{\mathcal{L}}\;
        N_3 \Big(\mathbf{0}_3,\;
        \mathbf{J}_\beta(\boldsymbol{\theta}^\ast)^{-1}
        \mathbf{K}_\beta(\boldsymbol{\theta}^\ast)
        \mathbf{J}_\beta(\boldsymbol{\theta}^\ast)^{-1}
        \Big),
        \label{eq:asymp}
    \end{equation}
    and so, after some algebra, 
    \begin{equation*}
    \sqrt{K}\left( L\left( \widehat{\boldsymbol{\theta }}_{\beta }\right)
    -L\left( \boldsymbol{\theta }^{\ast }\right) \right) \overset{L}{\underset{%
    n\rightarrow \infty }{\rightarrow }}N\left( 0,\sigma ^{2}\left( \boldsymbol{%
    \theta }^{\ast }\right) \right)
    \end{equation*}%
    with
    \begin{equation*}
    \sigma ^{2}\left( \boldsymbol{\theta }^{\ast }\right) =4\left( \boldsymbol{%
    \theta }^{\ast }-\boldsymbol{\theta }_{0}\right) ^{T}\boldsymbol{\Sigma }%
    \left( \boldsymbol{\theta }_{0}\right) ^{-1}\boldsymbol{\Sigma }\left( 
    \boldsymbol{\theta }^{\ast }\right) \boldsymbol{\Sigma }\left( \boldsymbol{%
    \theta }_{0}\right) ^{-1}\left( \boldsymbol{\theta }^{\ast }-\boldsymbol{%
    \theta }_{0}\right).
    \end{equation*}
    Then, the power of the test can be approximate as 
    \begin{eqnarray*}
    P_{W_{K}^{\beta }\left( \boldsymbol{\theta }_{0}\right) }\left( \boldsymbol{%
    \theta }^{\ast }\right) &=&\Pr \left( W_{K}^{\beta }\left( \boldsymbol{%
    \theta }_{0}\right) >\chi _{3,\gamma }^{2}\right) =\Pr \left( KL\left( 
    \widehat{\boldsymbol{\theta }}_{\beta }\right) >\chi _{3,\gamma }^{2}\right)
    \\
    &=&\Pr \left( K\text{ }\left( L\left( \widehat{\boldsymbol{\theta }}_{\beta
    }\right) -L\left( \boldsymbol{\theta }^{\ast }\right) \right) >K\chi
    _{3,\gamma }^{2}-K\text{ }L\left( \boldsymbol{\theta }^{\ast }\right) \right)
    \\
    &=&\Pr \left( \frac{\sqrt{K}\left( L\left( \widehat{\boldsymbol{\theta }}%
    _{\beta }\right) -L\left( \boldsymbol{\theta }^{\ast }\right) \right) }{%
    \sigma \left( \boldsymbol{\theta }^{\ast }\right) }>\frac{\sqrt{K}\left( 
    \frac{\chi _{3,\gamma }^{2}}{K}-L\left( \boldsymbol{\theta }^{\ast }\right)
    \right) }{\sigma \left( \boldsymbol{\theta }^{\ast }\right) }\right) \\
    & = &1-\Phi _{N(0,1)}^{\left( K\right) }\left( \frac{\sqrt{K}\left( \frac{\chi
    _{3,\gamma }^{2}}{K}-L\left( \boldsymbol{\theta }^{\ast }\right) \right) }{%
    \sigma \left( \boldsymbol{\theta }^{\ast }\right) }\right)
    \end{eqnarray*}%
    where $\Phi _{N(0,1)}^{\left( K\right) }\left( .\right) $ represents a
    sequence of distribution functions tending uniformly to the standard normal
    distribution. Then, taking limits in $K$
    \begin{equation*}
    \lim_{K\rightarrow \infty }P_{W_{K}^{\beta }\left( \boldsymbol{\theta }%
    _{0}\right) }\left( \boldsymbol{\theta }^{\ast }\right) =1.
    \end{equation*}

    \subsection*{Proof of Theorem \ref{thm:u_beta_mean_var}}
    The DPD attains its minimum at the true value of the model parameter. Therefore, under the null hypothesis, $E[\mathbf{p}_i(\boldsymbol{\theta}_0)] = E[\widehat{\mathbf{p}}_i]$ and    
    the first derivatives of the DPD loss vanish at $\boldsymbol{\theta}_0.$
    Then, $E\left[ \boldsymbol{U}_{\beta }\left( \boldsymbol{\theta }_{0}\right) \right] =\boldsymbol{0}_{3}.$ 
    Now, we prove that 
    $Cov\left[ \boldsymbol{U}_{\beta }\left( \boldsymbol{\theta }_{0}\right) %
    \right] =\frac{1}{K}\boldsymbol{K}_{\beta }\left( \boldsymbol{\theta }%
    _{0}\right).$ First, we have, 
    \begin{equation*}
    Cov\left[ \mathbf{p}_i\left( \boldsymbol{\theta }_{0}\right) -\widehat{%
    \mathbf{p}}_{i}\right] =\frac{1}{K_{i}}\left( \boldsymbol{D}_{i}^{\left(
    -1\right) }\left( \boldsymbol{\theta }_{0}\right) -\mathbf{p}_i\left( 
    \boldsymbol{\theta }_{0}\right) \mathbf{p}_i\left( \boldsymbol{\theta }%
    _{0}\right) ^{T}\right).
    \end{equation*}%
    Therefore, 
    \begin{eqnarray*}
    Cov\left[ \boldsymbol{U}_{\beta }\left( \boldsymbol{\theta }_{0}\right) %
    \right]  &=&\tsum \limits_{i=1}^{R}\frac{K_{i}}{K}\boldsymbol{W}_{i}\left( 
    \boldsymbol{\theta }_{0}\right) \boldsymbol{D}_{i}^{\left( \beta -1\right)
    }\left( \boldsymbol{\theta }_{0}\right) Cov\left[ \mathbf{p}_i\left( 
    \boldsymbol{\theta }_{0}\right) -\widehat{\mathbf{p}}_{i}\right] \left( 
    \frac{K_{i}}{K}\boldsymbol{W}_{i}\left( \boldsymbol{\theta }_{0}\right) 
    \boldsymbol{D}_{i}^{\left( \beta -1\right) }\left( \boldsymbol{\theta }%
    _{0}\right) \right) ^{T} \\
    &=&\tsum \limits_{i=1}^{R}\left( \frac{K_{i}}{K}\right) ^{2}\boldsymbol{W}%
    _{i}\left( \boldsymbol{\theta }_{0}\right) \boldsymbol{D}_{i}^{\left( \beta
    -1\right) }\left( \boldsymbol{\theta }_{0}\right) \frac{1}{K_{i}}\left( 
    \boldsymbol{D}_{i}^{\left( -1\right) }\left( \boldsymbol{\theta }_{0}\right)
    -\mathbf{p}_i\left( \boldsymbol{\theta }_{0}\right) \mathbf{p}%
    _{i}\left( \boldsymbol{\theta }_{0}\right) ^{T}\right)     
	\end{eqnarray*}
	\begin{eqnarray*}
    &&\times \left( \frac{K_{i}}{K}\boldsymbol{W}_{i}\left( \boldsymbol{\theta }%
    _{0}\right) \boldsymbol{D}_{i}^{\left( \beta -1\right) }\left( \boldsymbol{%
    \theta }_{0}\right) \right) ^{T} \\
    &=&\frac{1}{K}\tsum \limits_{i=1}^{R}\frac{K_{i}}{K}\boldsymbol{W}_{i}\left( 
    \boldsymbol{\theta }_{0}\right) \left \{ \boldsymbol{D}_{i}^{\left( 2\beta
    -1\right) }\left( \boldsymbol{\theta }_{0}\right) -\boldsymbol{D}%
    _{i}^{\left( \beta \right) }\left( \boldsymbol{\theta }_{0}\right) 
    \boldsymbol{1}_{L+1}\boldsymbol{1}_{L+1}^{T}\boldsymbol{D}_{i}^{\left( \beta
    \right) }\left( \boldsymbol{\theta }_{0}\right) \right \} \boldsymbol{W}%
    _{i}\left( \boldsymbol{\theta }_{0}\right) ^{T} \\
    &=&\frac{1}{K}\boldsymbol{K}_{\beta }\left( \boldsymbol{\theta }_{0}\right) .
    \end{eqnarray*}    
    \subsection*{Proof of theorem \ref{thm:rao_simple_chisq}}
    The previous theorem established that
    \begin{equation*}
    \sqrt{K}\boldsymbol{U}_{\beta }\left( \boldsymbol{\theta }_{0}\right) 
    \underset{K\rightarrow \infty }{\overset{L}{\rightarrow }}N(\boldsymbol{0}%
    _{3},\boldsymbol{K}_{\beta }\left( \boldsymbol{\theta }_{0}\right) ).
    \end{equation*}%
    Therefore, 
    \begin{equation*}
    \sqrt{K}\boldsymbol{U}_{\beta }\left( \boldsymbol{\theta }_{0}\right) 
    \boldsymbol{K}_{\beta }\left( \boldsymbol{\theta }_{0}\right) ^{-1/2}%
    \underset{K\rightarrow \infty }{\overset{L}{\rightarrow }}N(\boldsymbol{0}%
    _{3},\boldsymbol{I}).
    \end{equation*}%
    Then,%
    \begin{equation*}
    R_{K}^{\beta }\left( \boldsymbol{\theta }_{0}\right) =\left( \sqrt{K}%
    \boldsymbol{U}_{\beta }\left( \boldsymbol{\theta }_{0}\right) \boldsymbol{K}%
    _{\beta }\left( \boldsymbol{\theta }_{0}\right) ^{-1/2}\right) \left( \sqrt{K%
    }\boldsymbol{U}_{\beta }\left( \boldsymbol{\theta }_{0}\right) \boldsymbol{K}%
    _{\beta }\left( \boldsymbol{\theta }_{0}\right) ^{-1/2}\right) ^{T}\underset{%
    K\rightarrow \infty }{\overset{L}{\rightarrow }}\chi _{3}^{2}.
    \end{equation*}
    
    \subsection*{Proof of Theorem \ref{thm:rmdpde_asymp}}

    The Lagrangian function for our problem is given by,%
    \begin{equation*}
    L_{\beta }(\boldsymbol{\theta ,\lambda })=\boldsymbol{H}_{\beta }(%
    \boldsymbol{\theta })+\boldsymbol{\lambda }^{T}\boldsymbol{h(\theta ),}
    \end{equation*}%
    where $\boldsymbol{\lambda }$ is the Lagrange multiplier. The partial
    derivatives of the Lagrangian function with respect to $\boldsymbol{\lambda }
    $ and $\boldsymbol{\widetilde{\boldsymbol{\theta }}_{\beta }}$ must vanish
    
    \begin{equation*}
    \left \{ 
    \begin{array}{l}
    \left( \frac{\partial L_{\beta }(\boldsymbol{\theta ,\lambda })}{\partial 
    \boldsymbol{\theta }}\right) _{\left( \widetilde{\boldsymbol{\theta }}%
    \boldsymbol{,}\widetilde{\boldsymbol{\lambda }}_{\beta }\right) }=%
    \boldsymbol{U}_{\beta }\left( \widetilde{\boldsymbol{\theta }}_{\beta
    }\right) +H\boldsymbol{(\widetilde{\boldsymbol{\theta }}_{\beta })}%
    \widetilde{\boldsymbol{\lambda }}_{\beta }=\boldsymbol{0}_{3} \\ 
    \left( \frac{\partial L_{\beta }(\boldsymbol{\theta ,\lambda })}{\partial 
    \boldsymbol{\lambda }}\right) _{\left( \widetilde{\boldsymbol{\theta }}%
    \boldsymbol{,}\widetilde{\boldsymbol{\lambda }}_{\beta }\right) }=%
    \boldsymbol{h(\widetilde{\boldsymbol{\theta }}_{\beta })=0}_{3}%
    \end{array}%
    \right. .
    \end{equation*}%
    A first order Taylor expansion of $\boldsymbol{U}_{\beta }\left( \boldsymbol{%
    \theta }\right) $ in $\boldsymbol{\theta }_{0}$ and evaluated at $%
    \boldsymbol{\widetilde{\boldsymbol{\theta }}_{\beta },}$ gives%
    \begin{equation*}
    \boldsymbol{U}_{\beta }\left( \widetilde{\boldsymbol{\theta }}_{\beta
    }\right) =\boldsymbol{U}_{\beta }\left( \boldsymbol{\theta }_{0}\right)
    +\left( \frac{\partial \boldsymbol{U}_{\beta }\left( \boldsymbol{\theta }%
    \right) }{\partial \boldsymbol{\theta }}\right) _{\boldsymbol{\theta =%
    \widetilde{\boldsymbol{\theta }}_{\beta }}}\left( \boldsymbol{\widetilde{%
    \boldsymbol{\theta }}_{\beta }-\theta }_{0}\right) +o_{p}\left( \left \Vert 
    \boldsymbol{\widetilde{\boldsymbol{\theta }}_{\beta }-\theta }%
    _{0}\right \Vert \right)
    \end{equation*}%
    and applying the strong law of large numbers we get 
    \begin{equation*}
    \left( \frac{\partial \boldsymbol{U}_{\beta }\left( \boldsymbol{\theta }%
    \right) }{\partial \boldsymbol{\theta }}\right) _{\boldsymbol{\theta =%
    \widetilde{\boldsymbol{\theta }}_{\beta }}}\underset{K\rightarrow \infty }{%
    \overset{P}{\rightarrow }}-\boldsymbol{J}_{\beta }\left( \boldsymbol{\theta }%
    _{0}\right) .
    \end{equation*}%
    Therefore we can write, 
    \begin{equation}
    \boldsymbol{U}_{\beta }\left( \widetilde{\boldsymbol{\theta }}_{\beta
    }\right) =\boldsymbol{U}_{\beta }\left( \boldsymbol{\theta }_{0}\right) -%
    \boldsymbol{J}_{\beta }\left( \boldsymbol{\theta }_{0}\right) \left( 
    \boldsymbol{\widetilde{\boldsymbol{\theta }}_{\beta }-\theta }_{0}\right)
    +o_{p}\left( \left \Vert \boldsymbol{\widetilde{\boldsymbol{\theta }}_{\beta
    }-\theta }_{0}\right \Vert \right)  \label{0}
    \end{equation}%
    and 
    \begin{equation*}
    \boldsymbol{U}_{\beta }\left( \boldsymbol{\theta }_{0}\right) -\boldsymbol{J}%
    _{\beta }\left( \boldsymbol{\theta }_{0}\right) \left( \boldsymbol{%
    \widetilde{\boldsymbol{\theta }}_{\beta }-\theta }_{0}\right) +H\boldsymbol{(%
    \widetilde{\boldsymbol{\theta }}_{\beta })}\widetilde{\boldsymbol{\lambda }}%
    _{\beta }+o_{p}\left( \left \Vert \boldsymbol{\widetilde{\boldsymbol{\theta }}%
    _{\beta }-\theta }_{0}\right \Vert \right) =\boldsymbol{0}_{3}
    \end{equation*}%
    or%
    \begin{equation}
    \sqrt{K}\boldsymbol{J}_{\beta }\left( \boldsymbol{\theta }_{0}\right) \left( 
    \boldsymbol{\widetilde{\boldsymbol{\theta }}_{\beta }-\theta }_{0}\right) -%
    \sqrt{K}H\boldsymbol{(\widetilde{\boldsymbol{\theta }}_{\beta })}\widetilde{%
    \boldsymbol{\lambda }}_{\beta }+o_{p}(1)=\sqrt{k}\boldsymbol{U}_{\beta
    }\left( \boldsymbol{\theta }_{0}\right)  \label{una}
    \end{equation}%
    On the other hand 
    \begin{equation*}
    H\boldsymbol{(\widetilde{\boldsymbol{\theta }}_{\beta })=H\boldsymbol{(%
    \boldsymbol{\theta }}}_{0}\boldsymbol{\boldsymbol{)}+O}_{p}\left( \left \Vert 
    \boldsymbol{\widetilde{\boldsymbol{\theta }}_{\beta }-\theta }%
    _{0}\right \Vert \right) =\boldsymbol{H\boldsymbol{(\widetilde{\boldsymbol{%
    \theta }}_{\beta })+O}}_{p}(K^{-1/2}).
    \end{equation*}%
    Under $H_{0}$ we have $\widetilde{\boldsymbol{\lambda }}_{\beta }=%
    \boldsymbol{\boldsymbol{O}}_{p}(K^{-1/2}).$ Therefore, 
    \begin{equation*}
    \widetilde{\boldsymbol{\lambda }}_{\beta }H\boldsymbol{(\widetilde{%
    \boldsymbol{\theta }}_{\beta })=}\widetilde{\boldsymbol{\lambda }}_{\beta }%
    \boldsymbol{H\boldsymbol{(\boldsymbol{\theta }}}_{0}\boldsymbol{\boldsymbol{%
    )=O}}_{p}(K^{-1/2})\boldsymbol{\boldsymbol{O}}_{p}(K^{-1/2})=\boldsymbol{%
    \boldsymbol{O}}_{p}(K^{-1})=o_{p}\left( K^{-1/2}\right)
    \end{equation*}%
    and 
    \begin{equation*}
    \sqrt{K}\widetilde{\boldsymbol{\lambda }}_{\beta }H\boldsymbol{(\widetilde{%
    \boldsymbol{\theta }}_{\beta })=}\sqrt{K}\widetilde{\boldsymbol{\lambda }}%
    _{\beta }\boldsymbol{H\boldsymbol{(\boldsymbol{\theta }}}_{0}\boldsymbol{%
    \boldsymbol{)+o}}_{p}(1).
    \end{equation*}%
    Then (\ref{una}) can be written by 
    \begin{equation}
    \sqrt{K}\boldsymbol{J}_{\beta }\left( \boldsymbol{\theta }_{0}\right) \left( 
    \boldsymbol{\widetilde{\boldsymbol{\theta }}_{\beta }-\theta }_{0}\right) -%
    \sqrt{K}\boldsymbol{H\boldsymbol{(\boldsymbol{\theta }}}_{0}\boldsymbol{%
    \boldsymbol{)}}\widetilde{\boldsymbol{\lambda }}_{\beta }+o_{p}(1)=\sqrt{K}%
    \boldsymbol{U}_{\beta }\left( \boldsymbol{\theta }_{0}\right).  \label{2}
    \end{equation}%
    A first order Taylor expansion of $\boldsymbol{h(\boldsymbol{\theta })}$ at $%
    \boldsymbol{\boldsymbol{\boldsymbol{\theta }}}_{0}$ and evaluated in $%
    \boldsymbol{\widetilde{\boldsymbol{\theta }}_{\beta }}$ is given by,%
    \begin{equation*}
    \boldsymbol{h(\widetilde{\boldsymbol{\theta }}_{\beta })=}\boldsymbol{h(%
    \boldsymbol{\boldsymbol{\boldsymbol{\theta }}}_{0})+H\boldsymbol{(%
    \boldsymbol{\theta }}}_{0})^{T}\left( \boldsymbol{\widetilde{\boldsymbol{%
    \theta }}_{\beta }-\theta }_{0}\right) +o_{p}\left( \left \Vert \boldsymbol{%
    \widetilde{\boldsymbol{\theta }}_{\beta }-\theta }_{0}\right \Vert \right)
    \end{equation*}%
    But $\boldsymbol{h(\widetilde{\boldsymbol{\theta }}_{\beta })=}\boldsymbol{h(%
    \boldsymbol{\boldsymbol{\boldsymbol{\theta }}}_{0})=0}_{r},$then 
    \begin{equation}
    \boldsymbol{H\boldsymbol{(\boldsymbol{\theta }}}_{0})^{T}\sqrt{K}\left( 
    \boldsymbol{\widetilde{\boldsymbol{\theta }}_{\beta }-\theta }_{0}\right)
    +o_{p}(1)=0.  \label{3}
    \end{equation}%
    Now we are going to write equations (\ref{2}) and (\ref{3}) \ in a matrix form
    \begin{equation}
    \left( 
    \begin{array}{cc}
    \boldsymbol{J}_{\beta }\left( \boldsymbol{\theta }_{0}\right) & -\boldsymbol{%
    H\boldsymbol{(\boldsymbol{\theta }}}_{0}\boldsymbol{\boldsymbol{)}} \\ 
    \boldsymbol{H\boldsymbol{(\boldsymbol{\theta }}}_{0}\boldsymbol{\boldsymbol{)%
    }}^{T} & \boldsymbol{0}_{r\times r}%
    \end{array}%
    \right) \left( 
    \begin{array}{c}
    \sqrt{K}\left( \boldsymbol{\widetilde{\boldsymbol{\theta }}_{\beta }-\theta }%
    _{0}\right) \\ 
    \sqrt{K}\widetilde{\boldsymbol{\lambda }}_{\beta }%
    \end{array}%
    \right) =\left( 
    \begin{array}{c}
    \sqrt{K}\boldsymbol{U}_{\beta }\left( \boldsymbol{\theta }_{0}\right) \\ 
    \boldsymbol{0}_{r}%
    \end{array}%
    \right) +\left( 
    \begin{array}{c}
    o_{p}(1) \\ 
    o_{p}(1)%
    \end{array}%
    \right)  \label{4}
    \end{equation}
    
    The inverse matrix of , 
    \begin{equation*}
    \boldsymbol{R}\left( \boldsymbol{\theta }_{0}\right) =\left( 
    \begin{array}{cc}
    \boldsymbol{J}_{\beta }\left( \boldsymbol{\theta }_{0}\right) & -\boldsymbol{%
    H\boldsymbol{(\boldsymbol{\theta }}}_{0}\boldsymbol{\boldsymbol{)}} \\ 
    \boldsymbol{H\boldsymbol{(\boldsymbol{\theta }}}_{0}\boldsymbol{\boldsymbol{)%
    }}^{T} & \boldsymbol{0}_{r\times r}%
    \end{array}%
    \right)
    \end{equation*}%
    is given by 
    \begin{equation*}
    \boldsymbol{R}\left( \boldsymbol{\theta }_{0}\right) ^{-1}=\left( 
    \begin{array}{cc}
    \boldsymbol{R}_{11}\left( \boldsymbol{\theta }_{0}\right) & \boldsymbol{R}%
    _{21}\left( \boldsymbol{\theta }_{0}\right) \\ 
    \boldsymbol{R}_{12}\left( \boldsymbol{\theta }_{0}\right) & \boldsymbol{R}%
    _{22}\left( \boldsymbol{\theta }_{0}\right)%
    \end{array}%
    \right)
    \end{equation*}%
    with $\boldsymbol{R}_{11}\left( \boldsymbol{\theta }_{0}\right) =\boldsymbol{%
    J}_{\beta }\left( \boldsymbol{\theta }_{0}\right) ^{-1}-\boldsymbol{J}%
    _{\beta }\left( \boldsymbol{\theta }_{0}\right) ^{-1}\boldsymbol{H%
    \boldsymbol{(\boldsymbol{\theta }}}_{0}\boldsymbol{\boldsymbol{)}}\left( 
    \boldsymbol{H\boldsymbol{(\boldsymbol{\theta }}}_{0}\boldsymbol{\boldsymbol{)%
    }}^{T}\boldsymbol{\boldsymbol{J}_{\beta }\left( \boldsymbol{\theta }%
    _{0}\right) ^{-1}H\boldsymbol{(\boldsymbol{\theta }}}_{0}\boldsymbol{%
    \boldsymbol{)}}\right) ^{-1}\boldsymbol{H\boldsymbol{(\boldsymbol{\theta }}}%
    _{0}\boldsymbol{\boldsymbol{)}}^{T}\boldsymbol{J}_{\beta }\left( \boldsymbol{%
    \theta }_{0}\right) ^{-1}.$ If we consider the projection matrix 
    \begin{equation*}
    \boldsymbol{P\boldsymbol{(\boldsymbol{\theta }}}_{0}\boldsymbol{\boldsymbol{%
    )=I}}_{3\times3}-\boldsymbol{J}_{\beta }\left( \boldsymbol{\theta }_{0}\right)
    ^{-1}\boldsymbol{H\boldsymbol{(\boldsymbol{\theta }}}_{0}\boldsymbol{%
    \boldsymbol{)}}\left( \boldsymbol{H\boldsymbol{(\boldsymbol{\theta }}}_{0}%
    \boldsymbol{\boldsymbol{)}}^{T}\boldsymbol{\boldsymbol{J}_{\beta }\left( 
    \boldsymbol{\theta }_{0}\right) ^{-1}H\boldsymbol{(\boldsymbol{\theta }}}_{0}%
    \boldsymbol{\boldsymbol{)}}\right) ^{-1}\boldsymbol{H\boldsymbol{(%
    \boldsymbol{\theta }}}_{0}\boldsymbol{\boldsymbol{)}}^{T}
    \end{equation*}%
    and 
    \begin{equation*}
    \boldsymbol{R}_{11}\left( \boldsymbol{\theta }_{0}\right) =P\boldsymbol{%
    \boldsymbol{(\boldsymbol{\theta }}}_{0}\boldsymbol{\boldsymbol{)}J}_{\beta
    }\left( \boldsymbol{\theta }_{0}\right) ^{-1}
    \end{equation*}%
    Based on (\ref{4}) we have,$\boldsymbol{0}$ 
    \begin{equation*}
    \left( 
    \begin{array}{c}
    \sqrt{K}\left( \boldsymbol{\widetilde{\boldsymbol{\theta }}_{\beta }-\theta }%
    _{0}\right) \\ 
    \sqrt{K}\widetilde{\boldsymbol{\lambda }}_{\beta }%
    \end{array}%
    \right) =\left( 
    \begin{array}{cc}
    \boldsymbol{R}_{11}\left( \boldsymbol{\theta }_{0}\right) & \boldsymbol{R}%
    _{21}\left( \boldsymbol{\theta }_{0}\right) \\ 
    \boldsymbol{R}_{12}\left( \boldsymbol{\theta }_{0}\right) & \boldsymbol{R}%
    _{22}\left( \boldsymbol{\theta }_{0}\right)%
    \end{array}%
    \right) \left( 
    \begin{array}{c}
    \sqrt{K}\boldsymbol{U}_{\beta }\left( \boldsymbol{\theta }_{0}\right) \\ 
    \boldsymbol{0}_{r}%
    \end{array}%
    \right) +\left( 
    \begin{array}{c}
    o_{p}(1) \\ 
    o_{p}(1)%
    \end{array}%
    \right)
    \end{equation*}%
    Therefore, 
    \begin{eqnarray}
    \sqrt{K}\left( \boldsymbol{\widetilde{\boldsymbol{\theta }}_{\beta }-\theta }%
    _{0}\right) &=&\boldsymbol{R}_{11}\left( \boldsymbol{\theta }_{0}\right) 
    \sqrt{K}\boldsymbol{U}_{\beta }\left( \boldsymbol{\theta }_{0}\right)
    +o_{p}(1)  \label{01} \\
    &=&P\boldsymbol{\boldsymbol{(\boldsymbol{\theta }}}_{0}\boldsymbol{%
    \boldsymbol{)}J}_{\beta }\left( \boldsymbol{\theta }_{0}\right) ^{-1}\sqrt{K}%
    \boldsymbol{U}_{\beta }\left( \boldsymbol{\theta }_{0}\right) +o_{p}(1) 
    \notag
    \end{eqnarray}%
    and we have,%
    \begin{equation*}
    \sqrt{K}\left( \boldsymbol{\widetilde{\boldsymbol{\theta }}_{\beta }-\theta }%
    _{0}\right) \underset{K\rightarrow \infty }{\overset{L}{\rightarrow }}N(%
    \boldsymbol{0,\Sigma }\left( \boldsymbol{\theta }_{0}\right)
    \end{equation*}%
    with 
    \begin{equation*}
    \boldsymbol{\Sigma }\left( \boldsymbol{\theta }_{0}\right) =P\boldsymbol{%
    \boldsymbol{(\boldsymbol{\theta }}}_{0}\boldsymbol{\boldsymbol{)}J}_{\beta
    }\left( \boldsymbol{\theta }_{0}\right) ^{-1}\boldsymbol{K}_{\beta }\left( 
    \boldsymbol{\theta }_{0}\right) P\boldsymbol{\boldsymbol{(\boldsymbol{\theta 
    }}}_{0}\boldsymbol{\boldsymbol{)}J}_{\beta }\left( \boldsymbol{\theta }%
    _{0}\right) ^{-1}.
    \end{equation*}

    \subsection*{Proof of Theorem \ref{thm:rao_composite_chisq}}

    By (\ref{4}) we have 
    \begin{equation}
    \sqrt{K}\boldsymbol{U}_{\beta }\left( \widetilde{\boldsymbol{\theta }}%
    _{\beta }\right) =\sqrt{K}\boldsymbol{U}_{\beta }\left( \boldsymbol{\theta }%
    _{0}\right) -\boldsymbol{J}_{\beta }\left( \boldsymbol{\theta }_{0}\right) 
    \sqrt{K}\left( \boldsymbol{\widetilde{\boldsymbol{\theta }}_{\beta }-\theta }%
    _{0}\right) +o_{p}\left( \boldsymbol{1}\right)  \label{T5}
    \end{equation}%
    and by (\ref{01}), 
    \begin{equation}
    \sqrt{K}\left( \boldsymbol{\widetilde{\boldsymbol{\theta }}_{\beta }-\theta }%
    _{0}\right) =P\boldsymbol{\boldsymbol{(\boldsymbol{\theta }}}_{0}\boldsymbol{%
    \boldsymbol{)}J}_{\beta }\left( \boldsymbol{\theta }_{0}\right) ^{-1}\sqrt{K}%
    \boldsymbol{U}_{\beta }\left( \boldsymbol{\theta }_{0}\right) +o_{p}(1).
    \label{T6}
    \end{equation}%
    Now we replace (\ref{T6}) in (\ref{T5}) and we have, 
    \begin{eqnarray*}
    \sqrt{K}\boldsymbol{U}_{\beta }\left( \widetilde{\boldsymbol{\theta }}%
    _{\beta }\right) &=&\sqrt{K}\boldsymbol{U}_{\beta }\left( \boldsymbol{\theta 
    }_{0}\right) -\boldsymbol{J}_{\beta }\left( \boldsymbol{\theta }_{0}\right) P%
    \boldsymbol{\boldsymbol{(\boldsymbol{\theta }}}_{0}\boldsymbol{\boldsymbol{)}%
    J}_{\beta }\left( \boldsymbol{\theta }_{0}\right) ^{-1}\sqrt{K}\boldsymbol{U}%
    ^{\beta }\left( \boldsymbol{\theta }_{0}\right) +o_{p}(1) \\
    &=&\left[ \boldsymbol{I}_{3\times3}-\boldsymbol{J}_{\beta }\left( \boldsymbol{%
    \theta }_{0}\right) P\boldsymbol{\boldsymbol{(\boldsymbol{\theta }}}_{0}%
    \boldsymbol{\boldsymbol{)}J}_{\beta }\left( \boldsymbol{\theta }_{0}\right)
    ^{-1}\right] \sqrt{K}\boldsymbol{U}_{\beta }\left( \boldsymbol{\theta }%
    _{0}\right) +o_{p}(1).
    \end{eqnarray*}%
    Now we are going to study the term%
    \begin{equation*}
    \boldsymbol{A=}\left[ \boldsymbol{I}_{3\times3}-\boldsymbol{J}_{\beta }\left( 
    \boldsymbol{\theta }_{0}\right) P\boldsymbol{\boldsymbol{(\boldsymbol{\theta 
    }}}_{0}\boldsymbol{\boldsymbol{)}J}_{\beta }\left( \boldsymbol{\theta }%
    _{0}\right) ^{-1}\right] .
    \end{equation*}%
    We had 
    \begin{equation*}
    \boldsymbol{P\boldsymbol{(\boldsymbol{\theta }}}_{0}\boldsymbol{\boldsymbol{%
    )=I}}_{3\times3}-\boldsymbol{J}_{\beta }\left( \boldsymbol{\theta }_{0}\right)
    ^{-1}\boldsymbol{H\boldsymbol{(\boldsymbol{\theta }}}_{0}\boldsymbol{%
    \boldsymbol{)}}\left( \boldsymbol{H\boldsymbol{(\boldsymbol{\theta }}}_{0}%
    \boldsymbol{\boldsymbol{)}}^{T}\boldsymbol{\boldsymbol{J}_{\beta }\left( 
    \boldsymbol{\theta }_{0}\right) ^{-1}H\boldsymbol{(\boldsymbol{\theta }}}_{0}%
    \boldsymbol{\boldsymbol{)}}\right) ^{-1}\boldsymbol{H\boldsymbol{(%
    \boldsymbol{\theta }}}_{0}\boldsymbol{\boldsymbol{)}}^{T}
    \end{equation*}%
    therefore, 
    \begin{eqnarray*}
    \boldsymbol{A} &=&\boldsymbol{I}_{3\times3}-\boldsymbol{J}_{\beta }\left( 
    \boldsymbol{\theta }_{0}\right) \left[ \boldsymbol{\boldsymbol{I}}_{3\times3}-%
    \boldsymbol{J}_{\beta }\left( \boldsymbol{\theta }_{0}\right) ^{-1}%
    \boldsymbol{H\boldsymbol{(\boldsymbol{\theta }}}_{0}\boldsymbol{\boldsymbol{)%
    }}\left( \boldsymbol{H\boldsymbol{(\boldsymbol{\theta }}}_{0}\boldsymbol{%
    \boldsymbol{)}}^{T}\boldsymbol{\boldsymbol{J}_{\beta }\left( \boldsymbol{%
    \theta }_{0}\right) ^{-1}H\boldsymbol{(\boldsymbol{\theta }}}_{0}\boldsymbol{%
    \boldsymbol{)}}\right) ^{-1}\boldsymbol{H\boldsymbol{(\boldsymbol{\theta }}}%
    _{0}\boldsymbol{\boldsymbol{)}}^{T}\right] \boldsymbol{J}_{\beta }\left( 
    \boldsymbol{\theta }_{0}\right) ^{-1} \\
    &=&\boldsymbol{I}_{3\times3}-\boldsymbol{I}_{3\times3}+\boldsymbol{J}_{\beta }\left( 
    \boldsymbol{\theta }_{0}\right) \boldsymbol{J}_{\beta }\left( \boldsymbol{%
    \theta }_{0}\right) ^{-1}\boldsymbol{H\boldsymbol{(\boldsymbol{\theta }}}_{0}%
    \boldsymbol{\boldsymbol{)}}\left( \boldsymbol{H\boldsymbol{(\boldsymbol{%
    \theta }}}_{0}\boldsymbol{\boldsymbol{)}}^{T}\boldsymbol{\boldsymbol{J}%
    _{\beta }\left( \boldsymbol{\theta }_{0}\right) ^{-1}H\boldsymbol{(%
    \boldsymbol{\theta }}}_{0}\boldsymbol{\boldsymbol{)}}\right) ^{-1}%
    \boldsymbol{H\boldsymbol{(\boldsymbol{\theta }}}_{0}\boldsymbol{\boldsymbol{)%
    }}^{T}\boldsymbol{J}_{\beta }\left( \boldsymbol{\theta }_{0}\right) ^{-1}.
    \end{eqnarray*}%
    But $\boldsymbol{Q}^{\beta }\left( \boldsymbol{\widetilde{\boldsymbol{\theta 
    }}_{\beta }}\right) =\boldsymbol{J}_{\beta }\left( \boldsymbol{\theta }%
    _{0}\right) ^{-1}\boldsymbol{H\boldsymbol{(\boldsymbol{\theta }}}_{0}%
    \boldsymbol{\boldsymbol{)}}\left( \boldsymbol{H\boldsymbol{(\boldsymbol{%
    \theta }}}_{0}\boldsymbol{\boldsymbol{)}}^{T}\boldsymbol{\boldsymbol{J}%
    _{\beta }\left( \boldsymbol{\theta }_{0}\right) ^{-1}H\boldsymbol{(%
    \boldsymbol{\theta }}}_{0}\boldsymbol{\boldsymbol{)}}\right) ^{-1}.$%

    \noindent Therefore, 
    \begin{equation*}
    A=\boldsymbol{J}_{\beta }\left( \boldsymbol{\theta }_{0}\right) \boldsymbol{Q%
    }^{\beta }\left( \boldsymbol{\widetilde{\boldsymbol{\theta }}_{\beta }}%
    \right) \boldsymbol{H\boldsymbol{(\boldsymbol{\theta }}}_{0}\boldsymbol{%
    \boldsymbol{)}}^{T}\boldsymbol{J}_{\beta }\left( \boldsymbol{\theta }%
    _{0}\right) ^{-1}
    \end{equation*}%
    and 
    \begin{equation*}
    \sqrt{K}\boldsymbol{U}_{\beta }\left( \widetilde{\boldsymbol{\theta }}%
    _{\beta }\right) =\boldsymbol{J}_{\beta }\left( \boldsymbol{\theta }%
    _{0}\right) \boldsymbol{Q}^{\beta }\left( \boldsymbol{\widetilde{\boldsymbol{%
    \theta }}_{\beta }}\right) \boldsymbol{H\boldsymbol{(\boldsymbol{\theta }}}%
    _{0}\boldsymbol{\boldsymbol{)}}^{T}\boldsymbol{J}_{\beta }\left( \boldsymbol{%
    \theta }_{0}\right) ^{-1}\sqrt{K}\boldsymbol{U}_{\beta }\left( \boldsymbol{%
    \theta }_{0}\right) +o_{p}(1),
    \end{equation*}%
    and multiplying by $\boldsymbol{Q}^{\beta }\left( \boldsymbol{\widetilde{%
    \boldsymbol{\theta }}_{\beta }}\right) ^{T}$, we get 
    \begin{equation*}
    \sqrt{K}\boldsymbol{Q}^{\beta }\left( \boldsymbol{\widetilde{\boldsymbol{%
    \theta }}_{\beta }}\right) ^{T}\boldsymbol{U}_{\beta }\left( \widetilde{%
    \boldsymbol{\theta }}_{\beta }\right) =\boldsymbol{Q}^{\beta }\left( 
    \boldsymbol{\widetilde{\boldsymbol{\theta }}_{\beta }}\right) ^{T}%
    \boldsymbol{J}_{\beta }\left( \boldsymbol{\theta }_{0}\right) \boldsymbol{Q}%
    ^{\beta }\left( \boldsymbol{\widetilde{\boldsymbol{\theta }}_{\beta }}%
    \right) \boldsymbol{H\boldsymbol{(\boldsymbol{\theta }}}_{0}\boldsymbol{%
    \boldsymbol{)}}^{T}\boldsymbol{J}_{\beta }\left( \boldsymbol{\theta }%
    _{0}\right) ^{-1}\sqrt{K}\boldsymbol{U}_{\beta }\left( \boldsymbol{\theta }%
    _{0}\right) +o_{p}(1).
    \end{equation*}%
    Now we are going to study the expression%
    \begin{eqnarray*}
    \boldsymbol{B} &=&\boldsymbol{Q}^{\beta }\left( \boldsymbol{\widetilde{%
    \boldsymbol{\theta }}_{\beta }}\right) ^{T}\boldsymbol{J}_{\beta }\left( 
    \boldsymbol{\theta }_{0}\right) \boldsymbol{Q}^{\beta }\left( \boldsymbol{%
    \widetilde{\boldsymbol{\theta }}_{\beta }}\right) \boldsymbol{H\boldsymbol{(%
    \boldsymbol{\theta }}}_{0}\boldsymbol{\boldsymbol{)}}^{T}\boldsymbol{J}%
    _{\beta }\left( \boldsymbol{\theta }_{0}\right) ^{-1} \\
    &=&\left( \boldsymbol{H\boldsymbol{(\boldsymbol{\theta }}}_{0}\boldsymbol{%
    \boldsymbol{)}}^{T}\boldsymbol{\boldsymbol{J}_{\beta }\left( \boldsymbol{%
    \theta }_{0}\right) ^{-1}H\boldsymbol{(\boldsymbol{\theta }}}_{0}\boldsymbol{%
    \boldsymbol{)}}\right) ^{-1}\boldsymbol{H\boldsymbol{(\boldsymbol{\theta }}}%
    _{0}\boldsymbol{\boldsymbol{)}}^{T}\boldsymbol{J}_{\beta }\left( \boldsymbol{%
    \theta }_{0}\right) ^{-1}\boldsymbol{J}_{\beta }\left( \boldsymbol{\theta }%
    _{0}\right) \\
    &&\times\boldsymbol{J}_{\beta }\left( \boldsymbol{\theta }_{0}\right) ^{-1}%
    \boldsymbol{H\boldsymbol{(\boldsymbol{\theta }}}_{0}\boldsymbol{\boldsymbol{)%
    }}\left( \boldsymbol{H\boldsymbol{(\boldsymbol{\theta }}}_{0}\boldsymbol{%
    \boldsymbol{)}}^{T}\boldsymbol{\boldsymbol{J}_{\beta }\left( \boldsymbol{%
    \theta }_{0}\right) ^{-1}H\boldsymbol{(\boldsymbol{\theta }}}_{0}\boldsymbol{%
    \boldsymbol{)}}\right) ^{-1}\boldsymbol{H\boldsymbol{(\boldsymbol{\theta }}}%
    _{0}\boldsymbol{\boldsymbol{)}}^{T}\boldsymbol{J}_{\beta }\left( \boldsymbol{%
    \theta }_{0}\right) ^{-1} \\
    &=&\left( \boldsymbol{H\boldsymbol{(\boldsymbol{\theta }}}_{0}\boldsymbol{%
    \boldsymbol{)}}^{T}\boldsymbol{\boldsymbol{J}_{\beta }\left( \boldsymbol{%
    \theta }_{0}\right) ^{-1}H\boldsymbol{(\boldsymbol{\theta }}}_{0}\boldsymbol{%
    \boldsymbol{)}}\right) ^{-1}\boldsymbol{H\boldsymbol{(\boldsymbol{\theta }}}%
    _{0}\boldsymbol{\boldsymbol{)}}^{T}\boldsymbol{J}_{\beta }\left( \boldsymbol{%
    \theta }_{0}\right) ^{-1} \\
    &=&\boldsymbol{Q}^{\beta }\left( \boldsymbol{\widetilde{\boldsymbol{\theta }}%
    _{\beta }}\right) ^{T}.
    \end{eqnarray*}%
    Therefore, 
    \begin{equation*}
    \sqrt{K}\boldsymbol{Q}^{\beta }\left( \boldsymbol{\widetilde{\boldsymbol{%
    \theta }}_{\beta }}\right) ^{T}\boldsymbol{U}_{\beta }\left( \widetilde{%
    \boldsymbol{\theta }}_{\beta }\right) =\boldsymbol{Q}^{\beta }\left( 
    \boldsymbol{\widetilde{\boldsymbol{\theta }}_{\beta }}\right) ^{T}\sqrt{K}%
    \boldsymbol{U}_{\beta }\left( \boldsymbol{\theta }_{0}\right) +o_{p}(1).
    \end{equation*}%
    But 
    \begin{equation*}
    \sqrt{K}\boldsymbol{U}_{\beta }\left( \boldsymbol{\theta }_{0}\right) 
    \underset{K\rightarrow \infty }{\overset{L}{\rightarrow }}N(\boldsymbol{0}%
    _{3},\boldsymbol{K}_{\beta }\left( \boldsymbol{\theta }_{0}\right) )
    \end{equation*}%
    and 
    \begin{equation*}
    \sqrt{K}\boldsymbol{Q}^{\beta }\left( \boldsymbol{\widetilde{\boldsymbol{%
    \theta }}_{\beta }}\right) ^{T}\boldsymbol{U}_{\beta }\left( \widetilde{%
    \boldsymbol{\theta }}_{\beta }\right) \underset{K\rightarrow \infty }{%
    \overset{L}{\rightarrow }}N\left( \boldsymbol{0,\boldsymbol{Q}^{\beta
    }\left( \boldsymbol{\widetilde{\boldsymbol{\theta }}_{\beta }}\right) ^{T}%
    \boldsymbol{K}_{\beta }\left( \boldsymbol{\theta }_{0}\right) \boldsymbol{Q}%
    ^{\beta }\left( \boldsymbol{\widetilde{\boldsymbol{\theta }}_{\beta }}%
    \right) }\right) \boldsymbol{.}
    \end{equation*}%
    Now we define the random vector
    \begin{equation*}
    \boldsymbol{Y}_{K}=\left[ \boldsymbol{\boldsymbol{Q}^{\beta }\left( 
    \boldsymbol{\widetilde{\boldsymbol{\theta }}_{\beta }}\right) ^{T}%
    \boldsymbol{K}_{\beta }\left( \boldsymbol{\theta }_{0}\right) \boldsymbol{Q}%
    ^{\beta }\left( \boldsymbol{\widetilde{\boldsymbol{\theta }}_{\beta }}%
    \right) }\right] ^{-1/2}\sqrt{K}\boldsymbol{Q}^{\beta }\left( \boldsymbol{%
    \widetilde{\boldsymbol{\theta }}_{\beta }}\right) ^{T}\boldsymbol{U}_{\beta
    }\left( \widetilde{\boldsymbol{\theta }}_{\beta }\right)
    \end{equation*}%
    and we have 
    \begin{equation*}
    \boldsymbol{Y}_{K}\underset{K\rightarrow \infty }{\overset{L}{\rightarrow }}%
    N(\boldsymbol{0}_{3},\boldsymbol{I}_{r \times r})
    \end{equation*}%
    and on the other hand 
    \begin{equation*}
    R_{K}^{\beta }\left( \boldsymbol{\widetilde{\boldsymbol{\theta }}_{\beta }}%
    \right) =\boldsymbol{Y}_{K}^{T}\boldsymbol{Y}_{K}+o_{p}(1).
    \end{equation*}%
    Therefore, we get 
    \begin{equation*}
    R_{K}^{\beta }\left( \boldsymbol{\widetilde{\boldsymbol{\theta }}_{\beta }}%
    \right) \underset{K\rightarrow \infty }{\overset{L}{\rightarrow }}\chi
    _{r}^{2}.
    \end{equation*}

    
    
    
    
    
    

\end{document}